\documentclass[pagesize,paper=A4,11pt,bibliography=totoc,oneside,DIV=11]{scrartcl}
\usepackage[utf8]{inputenc}
\usepackage{enumitem}
\usepackage{graphicx}
\usepackage[dvipsnames]{xcolor}
\usepackage{amssymb,amsmath}
\usepackage{amsthm}
\usepackage{thmtools}
\usepackage{booktabs}
\usepackage{csquotes}
\usepackage{placeins}
\usepackage{tablefootnote}
\usepackage{microtype}
\usepackage{multicol}

\usepackage[font=small,labelfont=bf,format=plain]{caption} 
\usepackage[font=footnotesize]{subcaption}

\usepackage{xcolor}
\definecolor{darkgreen}{rgb}{0,.5,0}
\definecolor{darkblue}{rgb}{0,0,.5}
\definecolor{darkred}{rgb}{.4,0,0}

\usepackage[
backend=biber,
style=phys,
date=year,
eprint=true,
doi=false,
url=false,
isbn=false
]{biblatex}
\usepackage[colorlinks=true,linkcolor=ForestGreen,citecolor=red,hypertexnames=false]{hyperref}
\usepackage[nameinlink]{cleveref}
\usepackage{orcidlink}

\usepackage{todonotes}

\usepackage{tikz}
\usetikzlibrary{arrows.meta}
\usetikzlibrary{calc}
\usetikzlibrary{decorations.markings}
\usetikzlibrary{decorations.pathreplacing}
\tikzstyle{edge} = [black,line width=.3mm]
\tikzstyle{vertex}=[circle,minimum size=1.5mm, draw=black, fill=black, inner sep=0mm]

\tikzstyle{v} = [circle, draw=black, line width=.2pt, fill=black, inner sep=0pt, minimum size=1.5mm]
\tikzstyle{vb} = [circle, inner sep=0.1pt, fill=jred, minimum size=1mm]
\tikzstyle{e} =	[draw=jred,line width=1pt]
\tikzstyle{eb}= [draw=jgreen,line width=1pt]

\usetikzlibrary {patterns.meta}
\pgfdeclarepattern{
	name=hatch,
	parameters={\hatchsize,\hatchangle,\hatchlinewidth},
	bottom left={\pgfpoint{-.1pt}{-.1pt}},
	top right={\pgfpoint{\hatchsize+.1pt}{\hatchsize+.1pt}},
	tile size={\pgfpoint{\hatchsize}{\hatchsize}},
	tile transformation={\pgftransformrotate{\hatchangle}},
	code={
		\pgfsetlinewidth{\hatchlinewidth}
		\pgfpathmoveto{\pgfpoint{-.1pt}{-.1pt}}
		\pgfpathlineto{\pgfpoint{\hatchsize+.1pt}{\hatchsize+.1pt}}
		\pgfpathmoveto{\pgfpoint{-.1pt}{\hatchsize+.1pt}}
		\pgfpathlineto{\pgfpoint{\hatchsize+.1pt}{-.1pt}}
		\pgfusepath{stroke}
	}
}

\tikzset{
	hatch size/.store in=\hatchsize,
	hatch angle/.store in=\hatchangle,
	hatch line width/.store in=\hatchlinewidth,
	hatch size=5pt,
	hatch angle=0pt,
	hatch line width=.5pt,
}

\usepackage{pgfplots}
\pgfplotsset{
	compat=1.18,
	table/search path={plotdata},
}
\usepgfplotslibrary{statistics}

\usepackage[framemethod=tikz,
roundcorner=2pt,
linewidth=.5pt,
skipabove=10pt plus 20pt minus 2pt,
skipbelow=12pt plus 20pt minus 2pt]{mdframed}
\mdfsetup{linewidth=1.0pt, roundcorner=5pt}

\theoremstyle{plain}
\newtheorem{theorem}{Theorem}
\newtheorem{lemma}{Lemma}
\newtheorem{proposition}{Proposition}

\newtheorem{remark}[lemma]{Remark}

\theoremstyle{definition}

\numberwithin{equation}{section}

\newcounter{example}
\NewDocumentEnvironment {example} {o}
{ \refstepcounter{example}
	\pagebreak[2]
	\begin{mdframed}[linewidth=0.8pt,  linecolor=black,
		bottomline=false,topline=false,rightline=false, startcode=\needspace{4\baselineskip}]
		\noindent \textbf{Example  \theexample.}~}
	{ \end{mdframed}
}

\newcommand{\Graph}{G}
\newcommand{\trop}{\mathrm{tr}}

\newcommand{\R}{\mathbb{R}}

\DeclareMathOperator{\Res}{Res}
\newcommand{\asyO}[1]{\mathcal{O}\left(#1\right)}
\newcommand{\defas}{\mathrel{\mathop:}=}

\newcommand{\Hepp}{\mathcal{H}}
\newcommand{\abs}[1]{\left|#1\right|}

\DeclareMathOperator{\Aut}{Aut}
\newcommand{\firstsymanzik}{ \mathcal{U}}
\newcommand{\secondsymanzik}{\mathcal{F}}
\renewcommand{\d}{ \,\mathrm{d}}
\newcommand{\loopnumber}{\ell}
\newcommand{\feynmanintegral}{\Phi}
\newcommand{\tropicalintegral}{\Theta}
\newcommand{\period}{\mathcal{P}}
\newcommand{\renop}{\hat{\mathcal{R}}}
\newcommand{\ren}{{\mathcal{R}}}

\newcommand{\correctiontoscaling}{\psi}

\newcommand{\borel}[1]{\mathfrak{B}{\hskip-.2em}\left[#1\right]}

\newcommand{\nauty}{\href{http://pallini.di.uniroma1.it/}{\texttt{\textup{nauty}}}}

\newcommand\scalemath[2]{\scalebox{#1}{\mbox{\ensuremath{\displaystyle #2}}}}

\title{Renormalized tropical field theory}

\newcommand{\email}[1]{\href{mailto:#1}{#1}}
\author{
\thanks{Mathematical Institute, University of Oxford, OX2 6GG, UK, \email{paul-hermann.balduf@maths.ox.ac.uk}}
    Paul-Hermann Balduf
    \orcidlink{0000-0003-4475-3031}
    \and
    \thanks{Mathematical Institute, University of Oxford, OX2 6GG, UK, \email{erik.panzer@maths.ox.ac.uk}}
    Erik Panzer
    \orcidlink{0000-0002-9897-5812}
}

\begin{document}

\maketitle

\begin{abstract}
    We introduce tropical scalar field theory as a model for renormalizable quantum field theory, and examine in detail the case of quartic self-interaction and internal $O(N)$ symmetry. This model arises in a formally zero-dimensional limit of critical long-range models, but nevertheless its Feynman integrals exhibit strong numerical correlations with the ordinary 4-dimensional theory. The tropical theory retains the full complexity of renormalization with nested and overlapping vertex subdivergences and infinitely many primitive graphs.

    We compute the perturbation series of the tropical renormalization group functions exactly to 400 loops and study their asymptotic growth.
    In the minimal subtraction scheme, we find only an arithmetic sequence of singularities on the negative real axis in the Borel plane. These singularities are confluent and imply that the large-order perturbative asymptotics contain logarithmic and fractional power corrections. The absence of any further singularities suggests these series are Borel summable.
    In contrast, in a kinematic subtraction scheme, the singularity structure on the negative axis changes, and we find additional singularities on the positive real axis.
\end{abstract}

\tableofcontents

\newpage
\section{Introduction}

\subsection{Motivation and goals} \label{sec:motivation}

Tropical  field theory is a class of model quantum field theories where Feynman integrals are replaced by rational functions that reproduce their logarithmic UV singularities. There are multiple conceptual perspectives to motivate tropical field theory, they all lead to mathematically equivalent descriptions. In the present paper, we discuss three of them: The first, to be discussed in \cref{sec:longrange}, is to view tropical field theory as a  limit of long range quantum field theory, where propagator powers and spacetime dimensions simultaneously approach zero. The second one, discussed in \cref{sec:combinatorial_formulas},  is motivated by an approximation, or bound, of Feynman integrals through the leading monomial in their parametric integrand, which amounts to a \emph{tropicalization} of the Symanzik polynomials. This \emph{Hepp bound} has   originally been introduced  to prove renormalizability  \cite{hepp_proof_1966}, and then subsequently used     to give factorial bounds to renormalized integrals \cite{decalan_local_1981}. Its properties have been analysed in detail recently by the second author  \cite{panzer_hepps_2022}.   The third perspective is purely combinatorial and will become most clear only later in  \cref{sec:PDE_other}, in this perspective tropical field theory is a natural generalization of the zero dimensional field theory models that have been used to enumerate Feynman graphs, which endows these models with physically meaningful combinatorial Feynman rules.

The first goal of the present paper is to review and connect these three perspectives, and establish that tropical field theory is a physically meaningful model theory. For concreteness, we will  concentrate on the tropical version of $O(N)$-symmetric $\phi^4_4$ theory, although other scalar theories work analogously.  

Our second goal is to establish how to renormalize tropical field theory, and to demonstrate that it shows the usual phenomena, such as renormalization group flow and anomalous scale dependence of mass terms.

The third goal is to analyse the large-order growth of perturbation series in tropical field theory as a theory that includes infinitely many primitive graphs and vertex subdivergences.
As is well known, the renormalized perturbation series of most interacting quantum field theories are factorially divergent, in particular this is the case for $\phi^4$ theory \cite{decalan_local_1981,magnen_lipatov_1987,david_large_1988}. Two distinct mechanisms are known to cause this behaviour: \emph{instantons} are classical solutions to the field equations \cite{brezin_perturbation_1977a,brezin_perturbation_1977b,brezin_critical_1978,mckane_vacuum_1979,mckane_instanton_1978,mckane_perturbation_2019}, typically identified with the factorial growth of the total number of Feynman graphs \cite{bender_anharmonic_1969,bender_statistical_1976,bender_asymptotic_1978,cvitanovic_number_1978,bollobas_asymptotic_1982,borinsky_renormalized_2017}, whereas \emph{renormalons} arise from integrals over large powers of logarithms \cite{parisi_borel_1979,parisi_singularities_1978,beneke_renormalons_1999} and are typically identified with specific classes of Feynman graphs with many nested subdivergences \cite{olesen_vacuum_1978,broadhurst_exact_2001,clingerman_asymptotic_2025}. Despite decades of intense research, it remains a notorious problem to establish precisely how these mechanisms interact beyond leading-order estimates, which quantities they affect in which theories and which renormalization schemes, and to what extent they are sufficient to describe the entire large-order behaviour and non-perturbative physics (e.g.\ \cite{bergere_ambiguities_1984,grunberg_renormalons_1996,brezin_should_2023,marino_new_2022,dunne_instantons_2022} and many other).
What has prevented a definite conclusion on such questions has been that only very limited classes of Feynman integrals --- such as nested multi-edges or zig-zags \cite{brown_singlevalued_2015} --- can be computed analytically at all orders, and at the same time numerical integration of large graphs is hard and mostly restricted to primitive graphs \cite{balduf_statistics_2023, balduf_predicting_2024,BorinskyFavorito:logGamma,borinsky_tropicalized_2025}.

In tropical field theory, individual Feynman integrals can be computed efficiently, and moreover the entire amplitudes satisfy a partial differential equation --- the \emph{tropical loop equation} \cite{borinsky_tropicalized_2025} --- which allows to determine them exactly to very high loop order. Crucially, this computation includes all graphs of $\phi^4$ theory on an equal footing and makes no assumption about dominance of particular classes. Hence, one can measure  the growth rates of renormalized perturbation series at large loop order, and infer the singularities of their Borel transforms, in order to explicitly demonstrate whether it coincides with the expectations in the literature \cite{thooft_can_1977,beneke_renormalons_1999}.  Doing this for the renormalization group functions and critical exponents of renormalized tropical $\phi^4$ theory  is the third goal for the present article. 

In practical terms, the latter computation requires a large amount of numerical data analysis. We find that the structure of asymptotic series in tropical theory is rather complicated and requires tools that go beyond e.g. ordinary Richardson extrapolation.  Since these methods as such have no explicit relation with tropical field theory themselves, but at the same time they can be relevant for a different audience from the perspective of data analysis, the present work contains merely the physically significant results. A complete exposition of the methodology, together with a discussion of reliability and accuracy in controlled examples, are therefore outsourced to the accompanying manuscript \cite{balduf_asymptotic_2026}.

Finally, we want to mention that tropical field theory   has already found a fruitful application in numerical integration algorithms for primitive amplitudes \cite{borinsky_tropical_2023,borinsky_tropical_2023a,borinsky_tropicalized_2025}, and in the construction of counterterms \cite{salvatori_tropical_2024}.  Furthermore, the \emph{surfaceology}  framework \cite{arkani-hamed_all_2024,arkani-hamed_all_2024a,salvatori_all_2025} is closely related. We leave it to further work to discuss the precise correspondences and  mutual implications.

\subsection{Content and results}\label{sec:content}

Since this paper touches upon a number of different topics, not all of which might be relevant to every reader, we give in the present section a self-contained overview of the main results. Precise definitions and derivations are to be found in the corresponding sections. 

\Cref{sec:tropical_field_theory} is a detailed introduction to tropical field theory, collecting existing results from the literature and providing a large number of   example calculations for individual tropical Feynman integrals. 
Firstly, in \cref{sec:longrange}, we recall properties of long range scalar quantum field theory, where a Lagrangian has the general form \cref{longrange_lagrangian,longrange_lagrangian_reg},
\begin{align}
	\mathcal L &= \frac 12 \phi \left( \partial_\mu \partial^\mu- \bar\mu^{2}  \right) ^{\xi} \phi +\frac 12 m^{2\xi}\phi^2+ \frac{(4\pi)^2g_0}{4!}\phi^4.
\end{align}
We are interested in the parameter range $0<\xi <1$, in which case the infrared regulator $\bar \mu$ is distinct from the mass term $m$. In the $(D,\xi)$ phase space, tropical field theory amounts to the blow-up of the corner $(0,0)$, see \cref{fig:tropical_phase_space}. In \cref{sec:tropical_limit}, we take this tropical limit $\xi \rightarrow 0, D=\xi(4-2\epsilon)\rightarrow 0$ for individual massless IR-regularized Feynman integrals in their parametric representation. A key observation is that  external momenta naturally vanish in this limit unless they are specifically scaled as well. This is consistent with the physical interpretation of zero spacetime dimensions, where conventional momenta are not defined, although we stress that tropical field theory is distinct from the models known as \emph{zero-dimensional field theories}, which merely count graphs.   Finally, in \cref{sec:combinatorial_formulas}, we give a combinatorial formula, which has been known for the \emph{Hepp bound} \cite{panzer_hepps_2022} of primitive graphs: The tropical Feynman integral $\tropicalintegral[\Graph]$ of a graph is given by a sum over the superficial degrees of convergence $\omega_\gamma$ of all possible subgraphs   according to \cref{def:tropical_combinatorial},
\begin{align*} 
	\tropicalintegral[\Graph] &= \sum_{\sigma \in S_{\abs{E_\Graph}}} \frac{1}{\omega_{G^\sigma_1} \cdot \omega_{G^\sigma_2} \cdots \omega_{G^\sigma_{\abs{E_\Graph}-1}}\omega_{G^\sigma_{\abs{E_\Graph}}}}.
\end{align*}
For primitive graphs, this formula  is a bound to the non-tropical integrals \cite{hepp_proof_1966}. Since the Hepp bound is correlated with the true 4-dimensional value of primitive integrals within a few per cent accuracy, tropical field theory can equivalently been viewed as an analytic continuation of this approximation scheme to also include integrals with subdivergences.

In \cref{sec:recurrence_relations}, we discuss the recurrence relations that determine the bare Green's functions in tropical field theory, and illustrate their mechanism by examples comparing to the computation of individual tropical integrals in the previous section. As first discovered by Michael Borinsky \cite{borinsky_tropicalized_2025} and reviewed in \cref{sec:recurrence}, these recurrences take the form of a partial differential equation for the effective potential $\mathcal G(x,t,\kappa)$, the tropical loop equation. Here, $x$ counts external legs, $t$ counts loop order, and $\kappa$ counts number of mass insertions. The tropical loop equation represents the Dyson-Schwinger equations of tropical field theory. In \cref{sec:PDE_ON}, we generalize the tropical loop equation to include internal $O(N)$-symmetry for the case of $\phi^4$ theory in $D=4-2\epsilon$ (tropical) dimensions (\cref{thm:G_PDE_N}):
\begin{align*} 
	\Big( 2 \epsilon ~t\partial_t + x \partial_x + 2 \kappa \partial_\kappa   \Big) \mathcal G -4\mathcal G  &=  t \cdot \left( \frac{1 }{1-\partial_x^2 \mathcal G } + \frac{ (N-1)x}{x-\partial_x \mathcal G  } -N \right) .
\end{align*}
By elementary algebraic transformations demonstrated in \cref{sec:PDE_other}, one obtains equivalent partial differential equations for other classes of amplitudes, and based on vertex count instead of loop number. Particularly instructive  is the tropical  partition function  $\bar{\mathcal Z}$, which generates all graphs as a function of vertices $g$ and external legs $j$. It satisfies \cref{thm:Z_PDE_g},
\begin{align*}
	2\epsilon g   \partial_g  \bar{\mathcal Z}   + \left( (3-\epsilon )  j -\frac{N-1}{j} \right) \partial_j \bar{\mathcal Z}  +2\kappa \partial_\kappa \bar{\mathcal Z} + 2(\epsilon-2) \bar{\mathcal Z}  \ln  \bar{\mathcal Z}  &=  \partial_j^2 \bar{\mathcal Z}-N {\bar{\mathcal Z}}.
\end{align*}
The only non-linear term in this equation is $\bar{\mathcal Z} \ln \bar{\mathcal Z}$. In the  special case $\epsilon=2$, this non-linear term disappears and the differential equation reduces to a Dyson-Schwinger equation of zero-dimensional field theory, and $\bar{\mathcal Z}(\epsilon=2)$ merely enumerates graphs. This is the third perspective mentioned in \cref{sec:motivation}: Tropical field theory constitutes a non-linear generalization of zero dimensional field theory where the Feynman rules are much richer than merely counting graphs.

In \cref{sec:renormalization}, we turn to the divergences  and their renormalization in the tropical theory. The physics can best be understood from the divergence structure  of long-range (non-tropical) theory reviewed in \cref{sec:longrange_divergences}. Two important conclusions are that tropical field theory has vanishing field anomalous dimension, i.e. the field renormalization counter term $Z_2$ is finite, and that the infrared regulation parameter introduced in the tropical limit is not a mass, and the so-obtained tropical field theory should be interpreted as a IR-regularized massless theory. In \cref{sec:MS}, we  consider the minimal subtraction scheme (MS). The only non-trivial counterterm is that of the coupling, it can be computed  in the usual way, and gives rise to a non-trivial beta function. A mass can be included in the form of explicit 2-valent vertices, their renormalization with a new independent counterterm gives rise to a non-vanishing mass anomalous dimension.  The renormalization group functions and  renormalized Green functions satisfy the Callan-Symanzik equation, where the scale-dependence refers to the tropical mass scale. 
It is possible to use other renormalization schemes, in \cref{sec:kinematic_renormalization} we consider  one possible version of kinematic renormalization at zero external momentum (MOM). It gives rise to a different set of  renormalization group functions, which have a non-trivial dependence on higher orders of $\epsilon$.  
In \cref{sec:critical_exponents} we introduce two critical exponents, the correction to scaling exponent and the mass anomalous dimension. As it should be, they do not depend on the choice of renormalization scheme. 

While example results for illustrative purposes are included in all previous sections, \cref{sec:numerics} presents numerical results at larger loop orders in a more systematic fashion. In \cref{sec:correlation}, we focus on individual tropical Feynman integrals. Even when they contain subdivergences, their individual value is correlated with the one in non-tropical $\phi^4_4$ theory (\cref{fig:counterterm_correlation}); this indicates that tropical field theory can be a meaningful model for qualitative features of a non-tropical theory. For the contributions of individual graphs to the $\beta$-function, we compare the MS scheme with MOM at $\loopnumber=8$ loops (\cref{fig:beta_histogram}). In MS, most graphs contribute a comparably small positive value, while in the MOM scheme there is a notable fraction of graphs that contributes a negative value, or a larger positive value. Heuristically, this suggests that in MOM there are overall stronger cancellation effects than in MS.
Finally, numerical data  indicates that at large loop order, primitive graphs contribute a finite fraction, but not 100\% of the   value of the beta function in MS (\cref{fig:beta_MS_primitive_ratio}). Hence, the conjecture \cite{mckane_perturbation_2019} that this function   is asymptotically dominated by primitive graphs is unlikely to be true in the tropical theory.

We have computed power series solutions with rational coefficients to the tropical loop equation, and thereby obtained the exact MS beta function and correction to scaling exponent to $\loopnumber=400$ loops, as well as the MOM beta function to $\loopnumber=200$, and the mass anomalous dimension and mass critical exponent to $\loopnumber=180$. These and additional files will be made available online. We analyzed these  power series with a number of methods   that are described in detail in \cite{balduf_asymptotic_2026}. We find that their large-order asymptotic expansion is an overall factorial (i.e. Gevrey-1) growth as expected, and contains subleading corrections proportional to third roots and logarithms of the series index $n$.

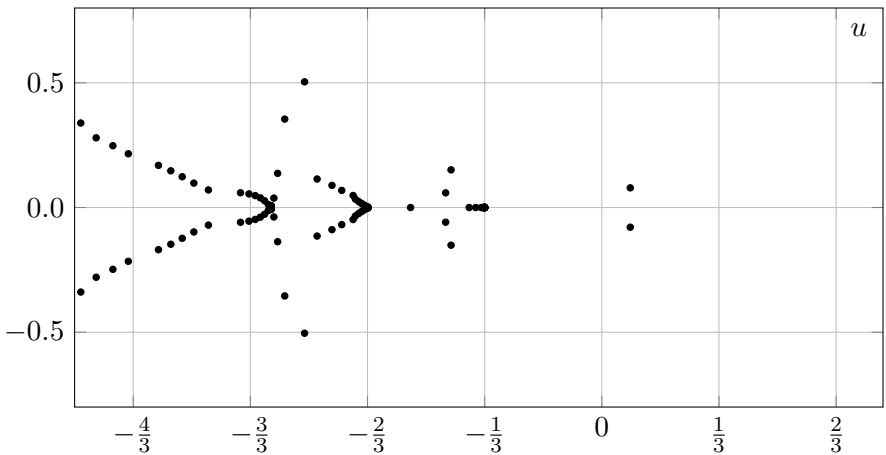
\begin{figure}[htb]
	\centering 
	\begin{tikzpicture}
		\begin{axis}[
			width= .8\linewidth, height=.45\linewidth,
			title={Poles   of Padé  of $\borel{\beta^\text{(MS)}}(u)$ in the complex Borel plane},
			grid=major,
			ymin=-.8 , ymax=.8,
			xmin=-1.5, xmax=.8,
			yticklabel style={
				anchor=east,
				/pgf/number format/precision=1,
				/pgf/number format/fixed,
				/pgf/number format/fixed zerofill,
			},
			xtick={-2.6666,-2.3333,-2,-1.6666,-1.3333,-1,-0.6666,-0.3333, 0,.3333, 0.66666},
			xticklabels={$-\frac 83$,$ -\frac 73$, $-\frac 63$,$ -\frac 53$,$ -\frac 43$, $-\frac 33$,$-\frac 23$, $-\frac 13$, 0, $\frac 13$, $\frac 23$},
			]
			\addplot [black, only marks, mark size=1.2pt] table {polesPadConformal2BetaMappedBack400.txt};

			\node[black, fill=white, below left=1mm] at (axis cs:.8,.8){$u$};
			
			\node[black, fill=white, above left] at (axis cs:.8, -1.5){\tiny Padé of 400-term series};
			
		\end{axis}
	\end{tikzpicture}
	\caption{Poles (dots) and zeros (crosses) of the order-200 Padé approximant of the 400-loop tropical beta function in minimal subtraction. The actual singularities are accumulation points of poles, they lie on the negative real Borel axis, at multiples of $u=-\frac 13$. Scattered individual poles are artefacts. The interpretation of such plots is discussed in e.g. \cite{costin_conformal_2021,balduf_asymptotic_2026}.  }
	\label{fig:poles_Borel_beta_conformal_back}
\end{figure}

With a combination of conformal mappings and Padé approximants, our data at 400 loops  of the tropical MS beta function is sufficient to explicitly resolve the first singularities in the Borel plane. \Cref{fig:poles_Borel_beta_conformal_back} shows the plot of this data, which contains the sum of more than $10^{800}$ vertex-type Feynman graphs. We interpret the regularly spaced algebraic singularities on the negative real Borel axis as instantons. Crucially, we do not find any other singularities except for these instantons and an essential singularity at $u=+\infty$. 
The leading instanton at $u=-1/3$ is a 6-fold confluent singularity with local exponents differing by integer multiples of $1/3$, forming three pairs that differ by an integer.
The corresponding large-order asymptotic expansion has the form
\begin{equation}\label{beta-asymptotics_general}
	\beta_n \sim A \cdot  (-3)^n  \Gamma \left(n + \frac{N+4}{2} \right) \left( 1 +\sum_{k>0} \frac{b_k}{n^{k/3}} + \frac{1}{n} \sum_{k\geq 0}\frac{c_k \log n}{n^{k/3}} \right), \qquad n \rightarrow \infty,
\end{equation}
where $\lbrace A,b_k,c_k\rbrace$ are numerical constants that depend on $N$.
Although logarithms and fractional powers have appeared in the literature in the context of renormalons and truncated Dyson-Schwinger equations (e.g.\ \cite{borinsky_resonant_2022,borinsky_treetubings_2024}), these functional forms are not typically considered in the existing literature concerning subleading corrections to instanton asymptotics. It would be interesting to confirm in future work if a refined instanton analysis can indeed reproduce \eqref{beta-asymptotics_general}.

In the MOM scheme, the asymptotics of renormalization group functions is notably different: Most strikingly, their Borel transform contains additional singularities on the positive real axis, see \cref{fig:poles_Borel_beta_MOM_conformal_back}. We confirm that the location, as function of the $O(N)$ parameter, of the leading singularity on the positive axis are consistent with the expectation for a leading renormalon. Qualitatively, this appears to confirm the picture of  \cite{thooft_can_1977,beneke_renormalons_1999}, that $\phi^4_4$ theory should have regularly spaced instanton singularities on the negative real axis and additional renormalons on the positive real axis.

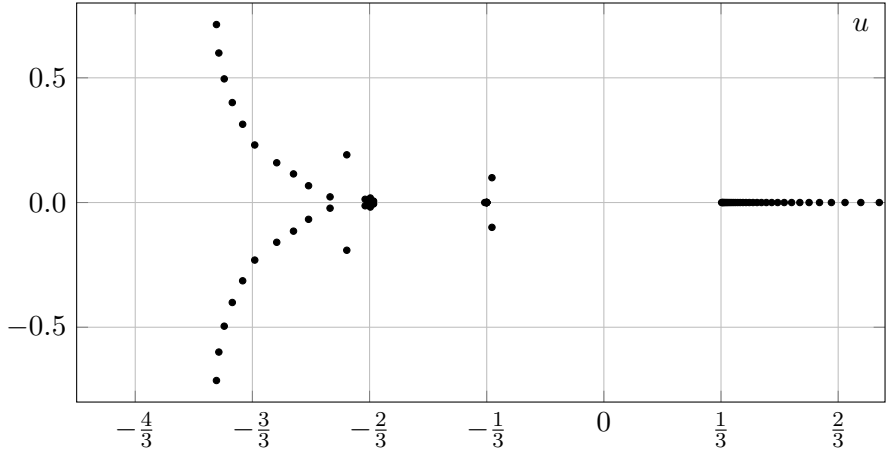
\begin{figure}[htb]
	\centering 
	\begin{tikzpicture}
		\begin{axis}[
			width= .8\linewidth, height=.45\linewidth,
			title={Poles   of Padé  of $\borel{\beta^\text{(MOM)}}(u)$ in the complex Borel plane},
			grid=major,
			ymin=-.8 , ymax=.8,
			xmin=-1.5, xmax=.8,
			yticklabel style={
				anchor=east,
				/pgf/number format/precision=1,
				/pgf/number format/fixed,
				/pgf/number format/fixed zerofill,
			},
			xtick={-2.6666,-2.3333,-2,-1.6666,-1.3333,-1,-0.6666,-0.3333, 0,.3333, 0.66666},
			xticklabels={$-\frac 83$,$ -\frac 73$, $-\frac 63$,$ -\frac 53$,$ -\frac 43$, $-\frac 33$,$-\frac 23$, $-\frac 13$, 0, $\frac 13$, $\frac 23$},
			]
			\addplot [black, only marks, mark size=1.2pt] table {polesPadConformal2BetaMOMMappedBack200.txt};

			\node[black, fill=white, below left=1mm] at (axis cs:.8,.8){$u$};
			
			\node[black, fill=white, above left] at (axis cs:.8, -1.5){\tiny Padé of 400-term series};
			
		\end{axis}
	\end{tikzpicture}
	\caption{Poles   of the order-100 Padé approximant of the 200-loop tropical beta function in a kinematic  renormalization scheme . There are still singularities at multiples of $u=-\frac 13$ on the negative real axis, but in addition there is a singularity on the positive real axis at  $u=+\frac 13$, which is the expected  location for the leading renormalon. }
	\label{fig:poles_Borel_beta_MOM_conformal_back}
\end{figure}

However, we also find renormalon-like singularities on the negative Borel axis. In fact, when $N>1$, the renormalon on the negative axis dominates the large-order asymptotics of the MOM beta function (for $N\leq 1$, the instanton dominates). The exponents of both the positive and the negative leading renormalon are rational functions of $N$, but they are not the same, and neither of them agrees with the prediction for the exponent from the literature \cite{dunne_instantons_2022}.

Even for the leading instanton-like singularity, we find that its exponent in the MOM scheme is larger than in the MS scheme. Concretely, for $-8<N\leq 1$, the coefficients $\beta_n$ of the tropical $\beta$ function in the MS scheme grow faster, by factor $n^{(N+8)/9}$, than in the MOM scheme. This was a surprise because the instanton computations in 4-dimensional $\phi^4$ theory predict both the location and the exponent of the leading instanton to be identical between MOM and MS.\footnote{In contrast, it is well known that the overall coefficient in the leading asymptotics is scheme dependent.} We leave it to future work to clarify which of these singularities correspond to which physical effect, and how their structure depends on choices of renormalization schemes more generally. 

Despite the many open questions surrounding the precise characteristics of individual Borel singularities,  our explicit numerical data supports two conjectures that are widely assumed, but to the best of our knowledge have never been   proved: Firstly, that the renormalization group functions in MS, and hence the critical exponents, do not contain renormalons, and secondly, that in neither scheme any of these functions has Borel singularities off the real axis.

\subsection{Acknowledgements}

We have greatly benefited from countless discussions  with Michael Borinsky over the last years.
We furthermore acknowledge discussions and comments from David Broadhurst, John Gracey, Veronica Fantini, and Giulio Salvatori.

The authors are funded through the Royal Society grant URF\textbackslash R1\textbackslash 201473.  We are also affiliated with MaScAmp (ERC Synergy grant 101167287) and  enjoyed discussions with the other members.
We are grateful to the
Mainz Institute for Theoretical Physics (MITP) of the DFG
Cluster of Excellence PRISMA+ (Project ID 390831469),
for its hospitality and partial support during the MathemAmplitudes workshop 2025, where part of this work was done.  Other parts were done during a visit at Perimeter Institute in September 2025. 

For the purpose of Open Access, the authors have applied a CC BY public copyright licence to any Author Accepted Manuscript (AAM) version arising from this preprint.

\FloatBarrier

\section{Tropical field theory}\label{sec:tropical_field_theory}

\subsection{Basic definitions}\label{sec:basics}
In massless bosonic QFT, the momentum-space propagator is $1/(p^2+i0)$. We restrict ourselves to Euclidean momenta $p^2>0$, and leave out the Feynman prescription $+i0$. We will consider Feynman integrals where all propagators appear with the same, potentially non-integer,   power $\xi \in \R$, which is expressed in Schwinger parametric form as
\begin{align}\label{schwinger_trick}
	   \frac{1}{(p^2)^\xi} &= \frac{1}{\Gamma(\xi)}   \int  \limits_{0}^\infty  \d \alpha \;  \alpha^{\xi-1} e^{-\alpha p^2} .
\end{align} 
Here, $\Gamma(\xi)$ is the Euler gamma function. 
Let $\Graph$ be a Feynman graph with set of internal edges $E_\Graph$, vertex set $V_\Graph$, and loop number $\loopnumber = \abs{E_\Graph}-\abs{V_\Graph}+ c_\Graph$, where $c_\Graph$ is the number of connected components. 
The superficial degree of convergence of  $\Graph$ is defined as 
\begin{align}\label{def:sdd}
\bar \omega_\Graph &\defas  \xi\cdot \abs{E_\Graph} - \loopnumber \frac{D}{2}.
\end{align}
Let $T\subseteq E_\Graph$ be a spanning tree and $F \subseteq E_\Graph$ a spanning 2-forest, and $p_F$  the sum of all external momenta entering precisely one of the  components of $F$. Then the first and second Symanzik polynomials (e.g.\ \cite{bogner_feynman_2010}) are the homogeneous polynomials given by   
\begin{align}\label{def:firstsymanzik}
\firstsymanzik &\defas \hspace{-1em}\sum_{T \text{ spanning tree}} ~\prod_{e \notin T} \alpha_e, \qquad  \qquad \secondsymanzik \defas  \hspace{-1em} \sum_{F \text{ spanning 2-forest}} p_F^2 \prod_{e \notin F} a_e.
\end{align}
With these definitions,  the scalar Feynman integral in parametric representation reads
\begin{align}\label{feynman_integral}
\feynmanintegral[\Graph] &= \frac{ \Gamma(\bar \omega_\Graph)}{(4\pi)^{\ell D/2}} \int  \Omega_\Graph  \prod_{e\in E_\Graph} \frac{  \alpha_e^{\xi_e-1}}{\Gamma(\xi_e)}    \frac{\secondsymanzik^{-\bar \omega_\Graph}}{\firstsymanzik^{\frac D 2-\bar \omega_\Graph}}.
\end{align}
The projective measure can be written  e.g. $\Omega_\Graph= \delta \big( \sum_e \alpha_e -1 \big)  \prod_{e} \d \alpha _e $.
The gamma function $\Gamma(\bar \omega_\Graph)$ in \cref{feynman_integral} diverges when $\bar \omega_\Graph\in \left \lbrace 0,-1,\ldots \right \rbrace $, it represents the superficial ultraviolet divergence of the integral. We use dimensional regularization   \cite{bollini_dimensional_1972,thooft_regularization_1972} with $D=4-2\epsilon$, so that superficial divergences take the form of poles in $\epsilon$ and/or the propagator powers $\xi$, to be discussed in more detail in \cref{sec:longrange_divergences}.
If $\Graph$ has generic Euclidean external momenta and no UV subdivergences,  $\Gamma(\bar \omega_\Graph)$ is the only divergent quantity in \cref{feynman_integral}.
In the special case where $\xi=1$ and  $\bar \omega_\Graph=\loopnumber \epsilon$, the  residue of the simple pole is  the  \emph{Feynman period}  \cite{broadhurst_knots_1995,schnetz_quantum_2010}
\begin{align}\label{def:period}
\period [\Graph] &\defas   \int \Omega_\Graph\;     \firstsymanzik^{-\frac D 2 }.
\end{align}
The period is a positive, finite, real number, and independent of kinematics; the sum of all periods constitutes the primitive contribution to the beta function. Periods have  served as prototypical examples for Feynman integrals regarding integration methods, number theoretic content, symmetries, and numerical growth rates. Specifically in $\phi^4_4$ theory,  hundreds of periods are known analytically, and more than 2 million  numerically  \cite{schnetz_hyperlogprocedures_2023,balduf_primitive_2024,balduf_predicting_2024,balduf_primitive_2024,BorinskyFavorito:logGamma,borinsky_tropicalized_2025}.

\subsection{Long-range field theory} \label{sec:longrange}

Long range massless $\phi^4$ theory  is defined through the Lagrangian density
\begin{align}\label{longrange_lagrangian}
	\mathcal L &= \frac 12 \phi \left( \partial_\mu \partial^\mu \right) ^{\xi} \phi + \frac{(4\pi)^2  g_0}{4!}\phi^4.
\end{align}
Here, $\xi\in \mathbb R$ is a fixed parameter, and $\xi=1$ amounts to conventional (\enquote{short-range}) massless $\phi^4$ theory. The non-integer derivative operator is defined as a Riesz derivative, i.e. through its Fourier transform, where it becomes multiplication by a non-integer power $\abs{p}^{2\xi}$ of momentum $p$. See \cite{liu_radial_2016,bayin_definition_2016,diethelm_hilbert_2022,li_liouvilletype_2023} for more details regarding  complex analysis and differential equations of such operators. 
The Lagrangian \cref{longrange_lagrangian}  describes a theory where the  propagator in momentum-space and in position-space  is
\begin{align}\label{longrange_propagator_positionspace}
	P(p) \propto \frac{1}{p^{2\xi}}, \qquad 	P(x) &\propto \frac{1}{r^{D-2\xi}}, \quad r=\abs{x}=\sqrt{x^2}.
\end{align}
Notice that instead of a non-integer power of $\partial_\mu \partial^\mu$ in the Lagrangian \cref{longrange_lagrangian},   the statistical physics literature often prefers a non-local interaction term in position space
\begin{align}\label{def:nonlocal_interaction}
	H_\text{LR}\defas \iint \d^D x \d^D r\; \frac{\phi(x)\phi(x+r)}{\abs{r}^{D+\sigma}},
\end{align}
which is related\footnote{The equivalence of the various setups refers to macroscopic properties at the critical point, which e.g. for a field theory  implies vanishing mass, and for the Ising-like models   critical temperature and  infinitely large lattices.  Being at the critical point is in general a non-trivial  assumption, which  in high energy physics is related to the \emph{naturalness} or \emph{hierarchy} problem \cite{weinberg_implications_1976,weinberg_implications_1976a,susskind_dynamics_1979,branchina_dimensional_2022}. In the present work, we make no claim regarding naturalness, or regarding equivalence of systems away from the critical point.} to   \cref{longrange_propagator_positionspace} according to $\sigma=+2\xi$, see e.g. \cite{chai_longrange_2021} or \cite[(1.58)]{mukamel_notes_2009}.

\medskip 
Theories with long-range interaction have a long history and many applications in statistical physics and solid state physics, see e.g. \cite{mukamel_notes_2009} and references therein.
For example, a classical result is that the 1-dimensional Ising model allows for a phase transition when $\xi<\frac 12$ \cite{dyson_existence_1969,dyson_nonexistence_1969,fisher_critical_1972}, but not for larger $\xi$   \cite{ising_beitrag_1925,peierls_bemerkungen_1934}. Typically, one is interested in $0<\xi<1$, since the choice   $  \xi < 0$ can lead to substantially different thermodynamic properties \cite{dauxois_dynamics_2002}, and non-causal solutions  to the equations of motion is causal \cite[Remark~2.8]{diethelm_hilbert_2022}.

Apart from statistical long-range models, another perspective on non-integer derivatives that has been studied  in recent years is  the  path integral formulation of \emph{fractional quantum mechanics}  \cite{laskin_fractional_2000,naber_time_2004},  where a   general Lévy process, for Lévy index $\alpha=2\xi$, replaces the Gaussian distributions of the Wiener process in the ordinary 1-dimensional path integral \cite{klafter_brownian_1996}.  Applications in optics are reviewed in \cite{malomed_basic_2024}.
Finally, we  remark in passing that   theories with $\xi>1$,   in particular $\xi=2$,  have been  studied in QFT  as  in the context of perturbatively renormalizable  quantum gravity \cite{stelle_renormalization_1977,mannheim_solution_2007,motohashi_quantum_2020,salvio_agravity_2014,ganz_reconsidering_2021,herrero-valea_status_2023}.

\medskip

For the action $\int\d^D x \; \mathcal L $ to be dimensionless,  the Lagrangian density must have mass dimension $[\mathcal L]=D$. Since $\left[ \partial_\mu \right] =1$, the  kinetic term in \cref{longrange_lagrangian}  implies that the field has mass dimension $[\phi]= \frac{D-2\xi}{2}$, so that the monomials in the action have
\begin{align}\label{longrange_operators_mass_dimension}
	\left[ \int \d^D x \; \phi (\partial_\mu \partial^\mu)^\xi \phi  \right] =  0, \qquad	\left[ \int \d^D x \; \phi^4 \right] =   D -4\xi, \qquad \left[ \int \d^D x \; \phi^2 \right] =  -2 \xi .
\end{align}
This implies that the coupling has mass dimension   $[g_0]= 4 \xi -D$, and  if one were to define a massive theory analogous to \cref{longrange_lagrangian}, the mass term would need to be $\frac 12 m^{2\xi}\phi^2$ in order for $m$ to have the dimension of a mass, $[m]=1$. 

Generally, the true scaling dimension differs from the mass dimension by \emph{anomalous dimensions} introduced through interactions. 
An operator is called \emph{irrelevant} when its scaling dimension is positive, because then its importance decreases as energy is lowered, so that this operator is (nearly) absent in the observed low-energy theory.   Specifically for   $\phi^4$ theory at $\xi=1$ in exactly $D=4$ dimensions,  the coupling term $\phi^4$ has mass dimension zero, but it  becomes irrelevant through interactions, which is known as  \emph{triviality} of $\phi_4^4$ theory  \cite{frohlich_triviality_1982,callaway_triviality_1988,aizenman_marginal_2021,kopper_asymptotically_2022,romatschke_what_2023,wang_triviality_2025}. 

Another relevant case where quantum corrections qualitatively alter the naive tree-level expectation concerns the   transition between the short-range ($\xi=1$) and long-range ($\xi <1$) theory at fixed dimension $D$ \cite{brezin_crossover_2014,behan_longrange_2017,behan_scaling_2017,chai_longrange_2021}. 
In the short-range theory,   the field variable obtains an anomalous dimension $\gamma_\phi$ and a corresponding critical exponent $\eta$ (precise definitions are in \cref{sec:critical_exponents}), which is positive, such that the full scaling dimension of $\phi$ becomes $  [\phi]+\eta= \frac {D-2}2+\eta >\frac{D-2}{2}$. Conversely, the long-range theory  has vanishing anomalous dimension (to be discussed in \cref{sec:longrange_divergences}), such that the field scaling dimension stays at its classical value $  [\phi]=\frac{D-2\xi}{2}$. The apparent inconsistency of these expressions for $\xi \rightarrow 1$ indicates the existence of a crossover value $\xi_\star$ \cite{sak_recursion_1973}, defined by $\frac{D-2}{2}+\eta = \frac{D-2\xi_\star}{2}$, that is, $\xi_\star=1-\eta<1$.  When $1>\xi \geq \xi_\star$,  the theory is effectively short range since the short range kinetic term $\phi \partial_\mu \partial^\mu \phi$, including its quantum corrections, is more relevant than the long-range one $\phi (\partial_\mu \partial^\mu)^\xi \phi$. Only when $\xi <\xi_\star$,  the theory becomes effectively long range. Equivalently, $\xi_\star$ can be determined from the anomalous dimensions in the long-range theory   \cite{honkonen_crossover_1989,honkonen_critical_1990}; numerical values\footnote{Note that the statistical physics literature often uses conventions that refer to $\abs{x}$ and $\abs{p}$ instead of $x^2$ and $p^2$, they differ by a factor two, $D=4-\epsilon$ and $\xi_\star=1-\frac 12 \eta$ or $\sigma_\star=2-\eta$.}   are visualized in \cref{fig:tropical_phase_space}.

The crossover $\xi_\star$ refers to the balance between the two kinetic terms. For a fixed dimension $D_0$, a second crossover is implied by the mass dimension of the $\phi^4$ term in \cref{longrange_operators_mass_dimension}: For $\xi <\frac {D_0} 4$, the $\phi^4$ interaction is irrelevant and the model becomes a free field with long-range propagator.  Close to the threshold, a Wilson-Fisher like fixed point  \cite{wilson_critical_1972,fisher_renormalization_1974} can be accessed perturbatively in $D_0-2\epsilon$ for  positive $\epsilon$, which has been done up to three loops \cite{benedetti_longrange_2020,benedetti_corrigendum_2024}.  All three fixed points, the short-range one, the long-range one, and the Gaussian one, represent conformally invariant theories  \cite{paulos_conformal_2016}.

\begin{figure}[htb]
	\centering
	\begin{tikzpicture}[>=Stealth, scale=1.7]
		
		\node[ scale=1,align=center] at (2, 2.6){Phase structure of $\phi^4$ theory};
		
		\draw[line width=0pt, color=gray!20!white,fill=gray!20!white] (0,0) -- (4.2,2.1) -- (4.2,0) --cycle;
		\node[text width=2cm, scale=.8,align=right] at (3.2,.8){long range free field};
		
		\draw[color=green!50!black,line width=.5mm, fill=green!50!white] (0.2,0) arc [start angle=0,   end angle=90,
		x radius=.2, y radius=.2] --(0,0) -- cycle; 
		
		\draw[line width=.2mm, dashed] (0,0) -- (4.2, 2.1);
		\node[text width=2cm, scale=.9, rotate=27] at (1.7,.67){ $D=4\xi$};
		
		\draw[line width=.5mm, <->]  (0,2.3)--(0,0) -- (4.5,0);
		\node at (4.6,0){$D$};
		\node at (0,2.5){$\xi$};
		\draw[line width=.3mm] (1,0)--+(0,.1);
		\draw[line width=.3mm] (2,0)--+(0,.1);
		\draw[line width=.3mm] (3,0)--+(0,.1);
		\draw[line width=.3mm] (4,0)--+(0,.1);
		\node at (1,-.3){\small 1};
		\node at (2,-.3){\small 2};
		\node at (3,-.3){\small 3};
		\node at (4,-.3){\small 4};
		\draw[line width=.3mm] (0,1)--+(.1,0);
		\draw[line width=.3mm] (0,2)--+(.1,0);
		\node at (-.2,1){\small 0.5};
		\node[align=right] at (-.2,2){\small 1};
	
		\draw[color=blue,line width=.3mm](4,2)--(3.75,1.999) --(3.5,1.993) --(3.25,1.982) --(3,1.963)  -- (2.75,1.933)--(2.5,1.88) -- (2.25,1.83) --(2,1.75) --(1.875, 1.7)-- (1.75,1.65)-- (1.5,1.5) -- (1.25,1.35) --(1,1) ;

		\node[color=blue] at (1,1.25){$\xi_\star$};
	
		\node[color=blue,text width=2cm,scale=.8 , rotate=27] at (2,1.35){  long range interacting};
		
		\node[text width=2cm, scale=.8] at (.8,1.8){short range interacting};
		
		\fill[color=red]  (3.2, 1.6) circle[radius=1.5pt];
		\draw[color=red, ->, line width=.5mm] (3.2, 1.6) -- + (-.5,0);
		\node[color=red] at  (3.0, 1.7) {$2\epsilon$};

		\node[ scale=1,align=center,color=green!50!black] at (-2.5, 2.3){Tropical $\phi^4$ theory};
		
		\draw[color=green!50!black,line width=.5mm] (-2,0) arc [start angle=0,   end angle=90, 		x radius=1.5, y radius=1.5]; 
		\draw[line width=.5mm, dashed] (-2.5,0.5) -- (-1.5, 1);
		\draw[line width=.5mm ] (-3.5,0 ) -- (-3.5, 2);
		\node[  scale=.7,rotate=90] at (-3.65,1.2){0-dimensional $\phi^4$ theory};
		
		\draw[color=orange,line width=.7mm, ->] (-2.16,.68) arc [start angle=28,   end angle=50, x radius=1.5, y radius=1.5]; 
		\fill[color=orange]  (-2.16,.68) circle[radius=2pt];
		\node[color=orange, scale=.8, rotate=26] at  (-1.93, .9) {$\epsilon=0$};
		\node[color=orange, scale=.8, rotate=58] at  (-2.6, 1.49) {$\epsilon=1$};
		\node[color=orange, scale=.8, rotate=89] at  (-3.41, 1.75) {$\epsilon=2$};
		\draw[color=orange,line width=.3mm ] (-3.5,0)+(26:1.4) -- +(26:1.5); 
		\draw[color=orange,line width=.3mm ] (-3.5,0)+(58:1.4) -- +(58:1.5); 
		\draw[color=orange,line width=.3mm ] (-3.5,0)+(89:1.4) -- +(89:1.5);

		\draw[black, line width=.3mm, ->] (.1,-.1) to [out=240, in=270] (-2.5,0);
		\node[black] at (-1.3,-.4){blow-up};
		
	\end{tikzpicture}
	\caption{Schematic phase diagram of long-range $\phi^4$ theory and emergence of tropical $\phi^4$ theory as a blow-up of the origin. The blue  curve indicates the crossover $\xi_\star = 1-\eta$ (data from \cite{onsager_crystal_1944,guillou_accurate_1985,guillou_accurate_1987,kompaniets_minimally_2017}). The red dot and arrow indicates the perturbative expansion in $D=D_0-2\epsilon$ dimensions for a fixed $D_0$ and $\xi$. \\
	Tropical field theory amounts to a blow-up of the origin, as indicated as a green arc in the left panel. The parameter $\epsilon$ of the perturbative expansion of tropical theory corresponds to the angle of approaching the origin in the non-tropical theory, indicated in orange. $\epsilon=2$ is a vertical approach and reproduces ordinary zero-dimensional $\phi^4$ theory, compare \cref{sec:PDE_other}.}
	\label{fig:tropical_phase_space}
\end{figure}
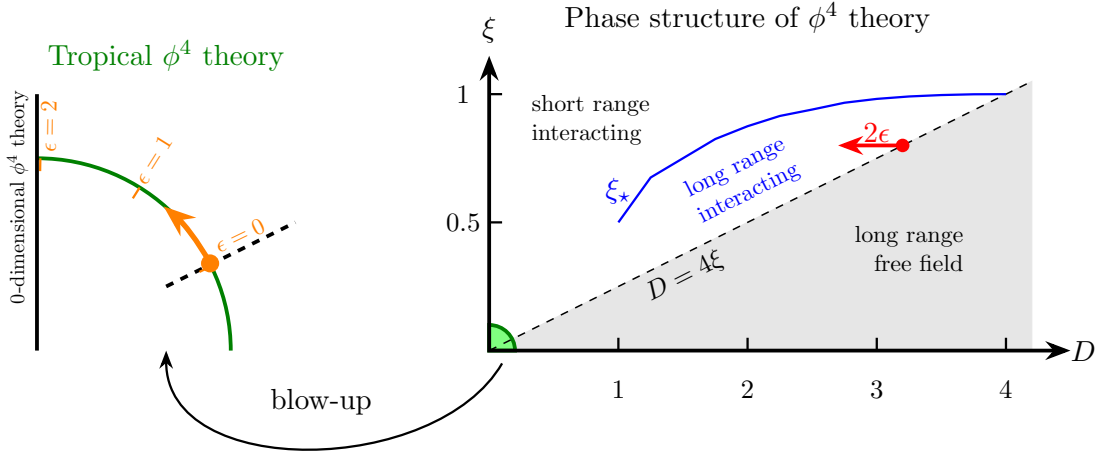

\subsection{Feynman integrals in the tropical limit} \label{sec:tropical_limit}

If $D=4\xi$ in the long range theory \cref{longrange_lagrangian}, Feynman integrals for the four-point function are logarithmically divergent, and we   introduce a dimensional regularization parameter $\epsilon$. We adopt the slightly unconventional notion that  $\xi$ is given, so that $D$ is determined through 
\begin{align}\label{dimension_xi}
	D  =D(\xi)&=  \xi\cdot  (4-2\epsilon) .
\end{align}
This choice of dimension implies that the superficial degree of convergence (\cref{def:sdd}) of all graphs is proportional to $\xi$. We denote by $\omega_\Graph$ the degree of convergence at $\xi=1$, and let $\bar \omega_\Graph:=\xi \omega_\Graph$ be the degree of convergence at arbitrary $\xi$, concretely
\begin{align}\label{def:sdd_scaled}
	\bar \omega_\Graph &= \xi \abs{E_\Graph} -\loopnumber \frac{D}{2} = \xi \big( \abs{E_\Graph}- \loopnumber(2-\epsilon)   \big) =:\xi \cdot  \omega_\Graph.
\end{align}
Tropical field theory is the limit $\xi \rightarrow 0$ of the long range theory, subject to \cref{dimension_xi}, as visualized in \cref{fig:tropical_phase_space}.  We recall parts of the computation of \cite{panzer_hepps_2022,borinsky_tropicalized_2025} in order to highlight the role of mass terms   which will be important in \cref{sec:renormalization}. Concretely,   we introduce a regularization parameter $ \bar \mu $ similar to what is done in  \cite{benedetti_longrange_2020}, such that \cref{longrange_lagrangian} becomes
\begin{align}\label{longrange_lagrangian_reg}
	\mathcal L &= \frac 12 \phi \left( \partial_\mu \partial^\mu- \bar\mu^{2}  \right) ^{\xi} \phi + \frac{(4\pi)^2g_0}{4!}\phi^4.
\end{align}
The momentum-space propagator becomes $(p^2+  \bar\mu^2)^{-\xi}$, which is finite for all Euclidean momenta. The quantity $\bar\mu$ is an arbitrary, but fixed, scale of mass dimension $\left[ \bar \mu \right] =1$, and $k$ is dimensionless. This functional form is different from the   propagator of a massive long-range theory, which would be $(p^{2\xi} + m^{2\xi})^{-1}$. Instead, the IR-regulated propagator formally coincides with the massive propagator of a  short-range theory, taken to the $\xi$\textsuperscript{th} power.
For the massive short-range theory, the second Symanzik polynomial  obtains an extra term compared to \cref{def:firstsymanzik}, 
\begin{align}\label{def:secondsymanzik_massive}
	\secondsymanzik &= \sum_{F \text{ spanning 2-forest}} p_F^2 \prod_{e \notin F} \alpha_e + \Big( \sum_e  \bar \mu^2  \alpha_e \Big)\cdot  \firstsymanzik.
\end{align} 
The $\xi$\textsuperscript{th} power of a massive propagator has a parametric representation given by \cref{schwinger_trick}, the full parametric Feynman integral (\cref{feynman_integral} is 
\begin{align}\label{feynmanintegral_tropicalized_1}
	\feynmanintegral[\Graph] &= \left( \prod_{e\in E_\Graph} \int_0^\infty  \frac{\d \alpha_e \; \alpha_e^{\xi-1}}{\Gamma(\xi)} \right)  \frac{e^{ -\frac{\secondsymanzik}{\firstsymanzik}}} {\firstsymanzik^{\xi (2-\epsilon)}}= \Gamma(\bar \omega_\Graph) \int \Omega_\Graph \left( \prod_{e\in E_\Graph} \frac{  \alpha_e^{\xi-1}}{\Gamma(\xi)} \right)  \frac{\secondsymanzik^{-\bar \omega_\Graph}}{\firstsymanzik^{\xi(2-\epsilon)-\bar \omega_\Graph}}
\end{align}
We introduce new integration variables $y_e= (\bar \mu^2 \alpha_e)^\xi$ to expose the scale dependence of the integral. This change amounts to $\alpha_e \mapsto y_e^{\frac 1 \xi}  \bar \mu^{-2}$ and $\alpha_e^{\xi-1} \d \alpha_e = \frac 1 \xi \d \alpha^\xi_e \mapsto   \frac 1 \xi \bar \mu^{-2\xi} \d y_e$. The first Symanzik polynomial is homogeneous of degree $\loopnumber$, therefore $\firstsymanzik(\vec \alpha) \mapsto  \bar \mu^{-2\loopnumber} \firstsymanzik(\vec y^{\frac 1 \xi})$. A spanning 2-forest has $(\loopnumber+1)$ edges not contained in it, hence   the first summand in \cref{def:secondsymanzik_massive} obtains one more factor $\mu^{-2}$ than the first Symanzik polynomial, and we have 
\begin{align}\label{FU_rescaled}
	\frac{\secondsymanzik}{\firstsymanzik} &=     \frac{\sum_F \frac{p_F^2}{\bar \mu^2} \prod_{e\notin F} y_e^{\frac 1 \xi}}{\firstsymanzik\left(\vec y^{\frac 1 \xi}\right)} + \sum_{e\in E_\Graph}   y_e ^{\frac 1 \xi}.
\end{align}
As could be expected, the whole integral \cref{feynmanintegral_tropicalized_1} obtains an overall prefactor $ \frac{\bar\mu^{-2 \xi \abs{E}}}{\bar\mu^{-2 \loopnumber\xi (2-\epsilon)}} = \bar\mu^{-2\bar \omega}$, where $\bar \omega$ is the scaled superficial degree of convergence (\cref{def:sdd_scaled}).   In the tropical limit $\xi \rightarrow 0$, this prefactor becomes unity, which is physically consistent since an integral in spacetime dimension zero can not have non-vanishing mass dimensions.  However, in order to keep track of the  would-be scale dependence, it is useful to introduce
\begin{align}\label{m_mu}
	\mu &:= \bar \mu^\xi.
\end{align}
This quantity $\mu$ has mass dimension $\xi$, so it is not physically a mass. The prefactor of the Feynman integral is then $\mu ^{-2  \omega_\Graph}$, for all values of $\xi$.

The projective integration measure in \cref{feynmanintegral_tropicalized_1} can, for example, be realized by a delta function $\delta\left( \sum_e \alpha_e-1 \right) $, which becomes after the change of variables
\begin{align}\label{feynmanintegral_tropicalized_2}
	\feynmanintegral[\Graph] &= \mu^{-2 \omega_\Graph } \Gamma(\bar \omega_\Graph) \left( \prod_{e\in E_\Graph} \int_0^\infty  \frac{\d y_e }{\xi \Gamma(\xi)} \right)  \frac{\delta \left( \mu^{-\frac{2}\xi} \sum_e y_e^{\frac 1 \xi} -1  \right)  }{\firstsymanzik^{\xi (2-\epsilon)}}  \left(  \frac{\secondsymanzik}{\firstsymanzik} \right)^{-\xi  \omega_\Graph}  .
\end{align}
Recall that the $L_\infty$ norm amounts to taking the maximum of a sum of terms, independent of finite coefficients $c_e>0$. The $L_\infty$ norm  of $\mathcal U$ is defined as the \emph{tropicalized first Symanzik polynomial} $ \firstsymanzik^\trop $ according to
\begin{align}\label{def:tropical_limit}
	\lim_{\xi \rightarrow 0} \Big( \sum_e c_e \cdot y_e^{\frac 1 \xi} \Big) ^\xi = \max \left \lbrace y_e \right \rbrace , \qquad  \qquad  \lim_{\xi \rightarrow 0} \Big(\firstsymanzik\big(\vec y^{\frac 1 \xi } \big)\Big)^\xi &=: \firstsymanzik^\trop (\vec y).
\end{align}
For a given numerical value $\vec y$,  $\firstsymanzik^\trop(\vec y)$ is   a monomial, but \emph{which} monomial depends on $\vec y$.  
 
In the limit $\xi \rightarrow 0$, the delta function in \cref{feynmanintegral_tropicalized_2} imposes the condition ~$\max \left \lbrace y_e \right \rbrace =1$.   A delta function of a non-linear transform of the argument involves a derivative, $\delta(f(x))= \frac{\delta(x-x_0)}{\abs{f'(x_0)}}$. For the change of variables $\alpha_e \mapsto y_e^{\frac 1 \xi}  \bar \mu^{-2}$, this derivative contributes a factor$\frac 1 \xi$ to the denominator of the Feynman integral. This factor combines  with $\Gamma(\bar \omega_G)$ in the numerator,  
\begin{align}\label{tropical_limit_gamma}
	\lim_{\xi \rightarrow 0}  \xi \Gamma(\xi \omega_\Graph )=\frac{1}{\omega_\Graph }.
\end{align}

Finally, we need to analyze the expression $\frac{\secondsymanzik}{\firstsymanzik}$ of \cref{FU_rescaled}. As before, the tropical limit of \cref{feynmanintegral_tropicalized_2} will retain only the maximum term of this sum. Fix a  spanning 2-forest $F$, and observe that it can be made into a spanning tree $T_e$ by adding exactly one edge $e\notin F$. The terms in the Symanzik polynomials are complements of trees and forests, hence we have  $\bar F \defas \prod_{e\notin F} y_e^{\frac 1 \xi}=y_e \prod_{f\notin T_e} y_f^{\frac 1 \xi} =: y_e^{\frac 1 \xi} \bar T_e$.  
The first Symanzik polynomial  $\firstsymanzik(\vec y^{\frac 1 \xi})$ in the denominator of $\frac{\secondsymanzik}{\firstsymanzik}$ is the sum of all spanning trees, it contains the set $T_e$ as a subset. The contribution of  a fixed 2-forest $F$ to $\frac \secondsymanzik \firstsymanzik$ is 
\begin{align}\label{FU_limit_term}
	\frac{p_F^2}{\mu^{2}}  \frac{\sum_{e \notin F} y_e^{\frac 1 \xi} \cdot \bar T_e }{ \ldots + \sum_{f \notin F} \bar T_f }= \frac{p_F^2}{\mu^{2}} \sum_{e \notin F} y_e^{\frac 1 \xi} \underbrace{\frac{    \bar T_e }{ \ldots + \sum_{f \notin F} \bar T_f }}_{\leq 1}.
\end{align} 
The factor $\frac{p_F^2}{\mu^2}$ is irrelevant in the tropical limit. The domain of integration is $y_e\in [0,1]$, therefore $0 \leq \bar T_e \leq 1$ for all edges $e$. The fraction on the right in \cref{FU_limit_term} is bounded by 1 as indicated.  Hence,   the whole sum \cref{FU_limit_term} takes the form $\sum_{e\notin F} c_e \cdot y_e^{\frac 1 \xi}$ where the coefficients $0 \leq c_e \leq 1$ stay bounded in the tropical limit.   On the other hand, the second sum in  \cref{FU_rescaled} contains every  $y_e^{\frac 1 \xi}$ with finite, non-zero coefficient (namely, unity) as well. We conclude that $\frac{\secondsymanzik}{\firstsymanzik}$ contains every $y_e^{\frac 1 \xi}$ with strictly positive finite coefficient, therefore
\begin{align*}
	\lim_{\xi \rightarrow 0} \left(\frac{\secondsymanzik}{\firstsymanzik}\right)^\xi &= \max \left \lbrace  y_e \right \rbrace = 1.  
\end{align*}
This says that $\frac{\secondsymanzik}{\firstsymanzik}$ does not contribute to the tropical limit of the integral \cref{feynmanintegral_tropicalized_2}. In particular, the tropical limit would have been the same if we had set all external momenta to zero from the start. This computation shows that tropical Feynman integrals naturally have vanishing external momenta, unless one introduces particular scaling transformations to retain them in the tropical limit. 
All in all, we now have 
\begin{align}\label{def:tropical_integral}
	\lim_{\xi \rightarrow 0}\feynmanintegral[\Graph] 
	&=  \mu ^{-2\omega_\Graph} \frac{1 }{ \omega_\Graph  }\left( \prod_{e \neq \star} \int_0^1  \d y_e    \right) \frac{1 }{\firstsymanzik^\trop(\vec y )^{ 2-\epsilon}} \Big|_{y_\star =1} =: \mu^{-2\omega_\Graph} \tropicalintegral[G].
\end{align}
This defines the \emph{tropical Feynman integral}  $\tropicalintegral[\Graph]$.

\begin{example}\label{ex:tropical_multiedge}
	Let $\Graph$ be the 1-loop multiedge. It has $\firstsymanzik= a_1 a_2$ and $\abs{E_\Graph}=2$ and $\loopnumber=1$. At zero momentum in the IR regularized theory,  \cref{feynmanintegral_tropicalized_1}) is
	\begin{align*}
		\feynmanintegral[\Graph]&=  \bar\mu^{-2\bar \omega_\Graph} \iint_0^\infty \frac{\d a_1 \d a_2}{\Gamma(\xi)^2} \alpha_1^{\xi-1} \alpha_2^{\xi-1} \frac{1}{(\alpha_1 + \alpha_2)^{\xi(2-\epsilon)}} e^{-\alpha_1 - \alpha_2}\\
		&= \frac{\bar \mu^{-2\bar \omega_\Graph}}{\Gamma(\xi)^2} \Gamma(\bar \omega_\Graph) \int_0^1 \alpha_1^{\xi-1} (1 - \alpha_1)^{\xi-1} ~ = \frac{\mu^{-2 \epsilon}~ \Gamma(\xi \epsilon)}{\Gamma(2\xi)}.
	\end{align*}
	For comparison, with external momentum $p\neq 0$, but without IR regulator,  one has
	\begin{align*}
		\feynmanintegral_2[\Graph]&= \iint_0^\infty \frac{\d \alpha_1 \d \alpha_2}{\Gamma(\xi)^2} \alpha_1^{\xi-1} \alpha_2^{\xi-1} \frac{e^{-\frac{\alpha_1 \alpha_2  }{\alpha_1+\alpha_2} p^2}} {(\alpha_1 + \alpha_2)^{\xi(2-\epsilon)}}= p^{-2\xi \epsilon} \frac{\Gamma(\xi \epsilon) \Gamma \left( \xi(1-\epsilon) \right) ^2}{\Gamma(\xi)^2 \Gamma \left( 2\xi(1-\epsilon) \right)  }.
	\end{align*}
	The two are distinct. However, the simple pole of a series expansion in $\epsilon$ agrees:
	\begin{align*}
		\mu^{2\xi \epsilon}\feynmanintegral[\Graph] &= \frac{2}{  \Gamma(1+2\xi) \epsilon }- \frac{\gamma_E}{\Gamma(2\xi)}+ \asyO{\epsilon^1}, \\
		p^{2\xi \epsilon} \feynmanintegral_2[\Graph] &= \frac{2}{\Gamma(1+2\xi)\epsilon} - \frac{\gamma_E+2 \psi(\xi) - 2 \psi(2\xi)}{\Gamma(2\xi)}+\asyO{\epsilon^1}.
	\end{align*}
	This is expected because different IR regularization do not affect UV singularities.  
	
	In the tropical limit, the first expression becomes the tropical integral \cref{def:tropical_integral}, $\lim_{\xi \rightarrow 0} \feynmanintegral[\Graph]=\frac 2 \epsilon =\tropicalintegral[\Graph]$. The tropical limit of the second integral is   $\frac{2}{\epsilon(1-\epsilon)}$, which has the same pole $\frac{2}{\epsilon}$, but additional higher-order terms in $\epsilon$. This calculation illustrates that the UV pole term of the tropical theory is universal, but the the full functional dependence on $\epsilon$ depends on our particular definitions. 
	
\end{example}

\begin{example}\label{ex:tropical_cycle}
	Consider the cycle graphs $r_n$ on $n$ edges. In $\phi^4$ theory, they have $2n$ legs.
	
	\begin{center}
		\begin{tikzpicture}
			\node  (x) at (0,0){};
			\node at ($(x) + (-1, 0) $) {$r_1=$};
			\node[vertex](v1) at ($(x)+(270:.15)$) {};
			\draw[edge] (v1) .. controls +(145:1.2) and +(45:1.2) .. (v1);
			\draw[edge] (v1) -- +(210:.4);
			\draw[edge] (v1) -- +(330:.4);
			
			\node  (x) at (3,0){};
			\node at ($(x) + (-1, 0) $) {$r_2=$};
			\node[vertex](v1) at ($(x)+(90:.3)$) {};
			\node[vertex](v2) at ($(x)+(270:.3)$) {};
			\draw[edge, bend angle=40, bend left] (v1) to  (v2);
			\draw[edge, bend angle=40, bend right] (v1) to  (v2);
			\draw[edge] (v2) -- +(210:.4);
			\draw[edge] (v2) -- +(330:.4);
			\draw[edge] (v1) -- +(30:.4);
			\draw[edge] (v1) -- +(150:.4);
			
			\node  (x) at (6,0){};
			\node at ($(x) + (-1.2, 0) $) {$r_3=$};
			\node[vertex](v1) at ($(x)+(90:.3)$) {};
			\node[vertex](v2) at ($(x)+(210:.3)$) {};
			\node[vertex](v3) at ($(x)+(330:.3)$) {};
			\draw[edge, bend angle=20, bend right] (v1) to  (v2);
			\draw[edge, bend angle=20, bend right] (v2) to  (v3);
			\draw[edge, bend angle=20, bend right] (v3) to  (v1);
			\draw[edge] (v1) -- +(50:.4);
			\draw[edge] (v1) -- +(130:.4);
			\draw[edge] (v2) -- +(170:.4);
			\draw[edge] (v2) -- +(250:.4);
			\draw[edge] (v3) -- +(10:.4);
			\draw[edge] (v3) -- +(290:.4);
			
			\node  (x) at (9,0){};
			\node at ($(x) + (-1.2, 0) $) {$r_4=$};
			\node[vertex](v1) at ($(x)+(45:.35)$) {};
			\node[vertex](v2) at ($(x)+(135:.35)$) {};
			\node[vertex](v3) at ($(x)+(225:.35)$) {};
			\node[vertex](v4) at ($(x)+(315:.35)$) {};
			\draw[edge, bend angle=10, bend right] (v1) to  (v2);
			\draw[edge, bend angle=10, bend right] (v2) to  (v3);
			\draw[edge, bend angle=10, bend right] (v3) to  (v4);
			\draw[edge, bend angle=10, bend right] (v4) to  (v1);
			\draw[edge] (v1) -- +(25:.4);
			\draw[edge] (v1) -- +(65:.4);
			\draw[edge] (v2) -- +(115:.4);
			\draw[edge] (v2) -- +(155:.4);
			\draw[edge] (v3) -- +(205:.4);
			\draw[edge] (v3) -- +(245:.4);
			\draw[edge] (v4) -- +(295:.4);
			\draw[edge] (v4) -- +(335:.4);

		\end{tikzpicture}
	\end{center}
	They have $\bar\omega_{r_n} = \xi(n-2+\epsilon)$. The Feynman integral \cref{feynmanintegral_tropicalized_1} at zero external momentum can be computed as in \cref{ex:tropical_multiedge}, alternatively one realizes that a sequence of $n$ edges is equivalent to a single edge with exponent $n\cdot \xi$, therefore  
	\begin{align*}
		\feynmanintegral[r_n] &= \mu^{-2(n-2+\epsilon)}  \frac {\Gamma (\xi(n-2+\epsilon))} {\Gamma(n\xi)}.
	\end{align*}
	As always, this formula does not include symmetry factors of the graph. The tropical integral \cref{def:tropical_integral} is 
	\begin{align*}
		\tropicalintegral[r_n]&= \lim_{\xi \rightarrow 0} \frac {\Gamma (\xi(n-2+\epsilon))} {\Gamma(n\xi)} = \frac{n}{n-2+\epsilon}.
	\end{align*}
\end{example}

\begin{remark}
	For finite $\xi$, Feynman integrals such as \cref{feynmanintegral_tropicalized_1} describe an IR-regulated massless long range theory, which has been studied in the literature for  $\xi\neq 0$ independently of tropicalization, e.g. \cite{benedetti_longrange_2020,benedetti_corrigendum_2024}.
	We stress  that the series expansions in $\epsilon$ for those integrals  can usually not be translated to the tropical setting because the limits $\epsilon \rightarrow 0$ and $\xi \rightarrow 0$ do not commute. 
	
\end{remark}

\subsection{Combinatorial formulas and approximation of Feynman integrals}\label{sec:combinatorial_formulas}

In \cref{ex:tropical_multiedge,ex:tropical_cycle}, the tropical integral \cref{def:tropical_integral} evaluates to a simple rational function, namely $\tropicalintegral[r_n]=\frac 1 {\omega_{r_n}}$. As shown in \cite{panzer_hepps_2022}, the tropical integral can be solved combinatorially for all graphs.  To this end, let $\sigma \in S_{\abs{E_\Graph}}$ be one of the $\abs{E_\Graph}!$ permutations of the edges of $\Graph$, and let $G^\sigma_k\defas \left \lbrace \sigma(1), \sigma(2), \ldots, \sigma(k) \right \rbrace \subset \Graph$ be the subgraph consisting of the first $k$ edges, and $\omega_{G^\sigma_{k}}$ its superficial degree of divergence. Then  \cite[(2.5)]{panzer_hepps_2022} 
\begin{align}\label{def:tropical_combinatorial}
	\tropicalintegral[\Graph] &= \sum_{\sigma \in S_{\abs{E_\Graph}}} \frac{1}{\omega_{G^\sigma_1} \cdot \omega_{G^\sigma_2} \cdots \omega_{G^\sigma_{\abs{E_\Graph}-1}}\omega_{G^\sigma_{\abs{E_\Graph}}}}.
\end{align}

\begin{example}\label{ex:tropical_cycle_combinatorial}
	For the 1-loop multiedge from \cref{ex:tropical_multiedge}, the overall degree of divergence is $\omega_\Graph=\epsilon$. There are two subgraphs, each of which is a single edge and has $\omega_e=1$. Consequently, \cref{def:tropical_combinatorial} reproduces $\frac 1 \epsilon \frac 1 1 + \frac 1 \epsilon \frac 1 1 = \frac 2 \epsilon$. 
	Similarly, the $n$-edge cycle $r_n$ of \cref{ex:tropical_cycle} has  $\omega_{r_n}= \epsilon+n-2$. There are $n$ ways to choose the first edge, and the remaining graph is a tree. A tree or forest of $k$ edges has $\omega=k$, and there are $k$ ways to remove one of them, each of which produces a forest with $(k-1)$ edges. A simple induction argument shows that therefore $\tropicalintegral[T]=1$ for every tree.  Therefore, the combinatorial formula \cref{def:tropical_combinatorial} for a cycle graph boils down to $n$ copies of the same term, and we reproduce
	$	\tropicalintegral[r_n]  = n\cdot \frac{1}{n-2+\epsilon}\cdot \frac 1 1  = \frac n {n-2+\epsilon}$ as expected.
\end{example}

By subdividing and regrouping the set of  \emph{all} permutations of edges  in different ways, \cref{def:tropical_combinatorial} can be rewritten as a sum over different types of subgraphs. Such formulas have been proved in \cite[Sec.~3]{panzer_hepps_2022}, we merely recall some immediate observations:
\begin{lemma}\label{lem:tropical_tree} 
	The quantity $\tropicalintegral[\Graph]$ given by \cref{def:tropical_combinatorial} satisfies
	\begin{enumerate}
		\item $	\tropicalintegral[\Graph]  = \frac 1 {\omega_\Graph} \sum_{e\in E_\Graph} \tropicalintegral[\Graph \setminus e]$.
		\item 	If $\Graph$ is a  tree graph, then $\tropicalintegral[\Graph]=1$.
		\item  If $\Graph$ has 1PI components $\gamma_j$ (i.e. $\Graph$ consists of one or multiple trees, the vertices of which are replaced by the 2-connected subgraphs $\gamma_j$), then $\tropicalintegral[\Graph]=\prod_j \tropicalintegral[\gamma_j]$.
	\end{enumerate}
\end{lemma}

\begin{example}\label{ex:tropical_twopoint_twoloop}
	The 1PI 2-point function at two loops  is the sum of two graphs, 
	\begin{center}
		\begin{tikzpicture}
			\node at (-.5,0){$S:=$};
			\node[vertex](v1) at (.5,0){};
			\node[vertex](v2) at (1.8,0){}; 
			
			\draw[edge, bend angle=50,bend left](v1) to (v2);
			\draw[edge ](v1) to (v2);
			\draw[edge, bend angle=50,bend right](v1) to (v2);
			\draw[edge] (v1) to +(-.3,0);
			\draw[edge] (v2) to +( .3,0);
			
			\node at (2.5,-.1){$,$};
			\node at (4,0){$P:=$};
			\node[vertex](v1) at (5,-.3){};
			\node[vertex](v2) at (5,.3){}; 
			
			\draw[edge, bend angle=60,bend left](v1) to (v2);
			\draw[edge, bend angle=60,bend right](v1) to (v2);
			\draw[edge ](v2) ..controls +(.6,.7) and +(-.6,.7)..  (v2);
			\draw[edge] (v1) to +(-.3,0);
			\draw[edge] (v1) to +( .3,0);
			
			\node at (5.5,-.1){$.$};
		\end{tikzpicture}
	\end{center}

	The two-loop mulitiedge, or sunrise, $S$ has $\omega_S=3-2(2-\epsilon)= -1+2\epsilon$. In the first formula of \cref{lem:tropical_tree}, we have three (equivalent) choices for the first edge, and each of them results in a remaining graph that is a 1-loop multiedge $r_2$. Therefore
	\begin{align*}
		\tropicalintegral[S] &= \frac{1}{ 2\epsilon-1}\cdot\left( \frac 2 \epsilon + \frac 2 \epsilon + \frac 2 \epsilon \right) = \frac{6}{\epsilon(2\epsilon - 1)}.
	\end{align*}
	The sunrise has symmetry factor $\frac 1 6$. The second graph, $P$,   is a 1-vertex product of $r_2$ and $r_1$ (\cref{ex:tropical_cycle}) and has symmetry factor $\frac 14$. Together, one obtains  the tropical 2-loop 2-point amplitude	
	\begin{align*}
		\Gamma^{(2)}_2 &= \frac 1 6 \cdot  \frac{6}{\epsilon(2\epsilon - 1)}+ \frac 1 4 \cdot \frac{2}{\epsilon}\cdot \frac{1}{\epsilon-1}  = \frac{4\epsilon-3}{2\epsilon(\epsilon-1)(2\epsilon-1)}.
	\end{align*}
\end{example}

\begin{example}\label{ex:tropical_I3}
	Consider the sunrise as a subgraph in the 1-loop multiedge, 
	\begin{center}
		\begin{tikzpicture}
			\node at (-.5,0){$I_3:=$};
			\node[vertex](v1) at (.4,0){};
			\node[vertex](v2) at (1.2,.5){};
			\node[vertex](v3) at (2.8,.5){};
			\node[vertex](v4) at (3.6,0){};
			
			\draw[edge, bend angle=30,bend right](v1) to (v4);
			\draw[edge, bend angle=15,bend left](v1) to (v2);
			\draw[edge, bend angle=15,bend left](v3) to (v4);
			\draw[edge, bend angle=40,bend left](v2) to (v3);
			\draw[edge ](v2) to (v3);
			\draw[edge, bend angle=40,bend right](v2) to (v3);
			\node at (4.2,-.1){$.$};
		\end{tikzpicture}
	\end{center}
	Using formula 1 of \cref{lem:tropical_tree} and  $\tropicalintegral[S]$ from \cref{ex:tropical_twopoint_twoloop}, one finds
	\begin{align*}
		\tropicalintegral[I_3] &= \frac{1}{3\epsilon}\left( 3  \cdot  \frac{6}{\epsilon(2\epsilon-1)} + 3  \cdot  \frac{1}{2 \epsilon+1} \left( 2 \cdot  \frac{4}{2+\epsilon} + 3 \cdot \frac 2 \epsilon \right)   \right) = \frac{40(1+\epsilon)}{\epsilon(4 \epsilon^3+8 \epsilon^2-\epsilon-2)}.
	\end{align*}
	Note that  this tropical integral has merely a first-order pole at $\epsilon=0$. This will become important later in \cref{ex:I3_tropical_renormalized}.
\end{example}

In \cite{panzer_hepps_2022}, the tropical Feynman integral and its combinatorial representation \cref{def:tropical_combinatorial} were only considered for the case where  $\Graph$ is superficially log-divergent and has no subdivergences. In that case, the tropical integral becomes the \emph{Hepp bound} $\Hepp[\Graph] = \Res_{\epsilon=0} \tropicalintegral[\Graph]$. That is, the Hepp bound is the tropical quantity corresponding to the Feynman period (\cref{def:period}) of the non-tropical theory. crucial difference with respect to \cite{panzer_hepps_2022} is that we consider these quantities for graphs with subdivergences in dimensional regularization, so that $\tropicalintegral[\Graph]$ is a rational function of $\epsilon$, whereas $\Hepp[\Graph]$ is a real number. 

This \enquote{analytic continuation} of the Hepp bound formula can in fact be motivated in a different way: Although the Hepp bound appears to be a drastic simplification compared to the Feynman period, it has been observed empirically that $\Hepp[G]$ and $\period[G]$ are closely correlated numerically \cite{panzer_hepps_2022,balduf_statistics_2023}, see \cref{fig:P_H}. This correlation allows to predict the value of a period to a few per cent accuracy, and has   been used for estimating the primitive beta function \cite{kompaniets_minimally_2017} and for importance sampling of Feynman graphs \cite{balduf_predicting_2024}. Hence, without any reference to the physical relevance of long-range field theory, it is plausible to apply the same combinatorial formula to graphs with subdivergences, in the expectation that it might be a numerically useful approximation, too. Indeed, the plot  \cref{fig:counterterm_correlation} in \cref{sec:correlation} shows that   a close correlation between tropical and non-tropical theory exists even for non-primitive graphs. 

\begin{figure}[htbp]
	\centering
	\includegraphics[width=.5\linewidth]{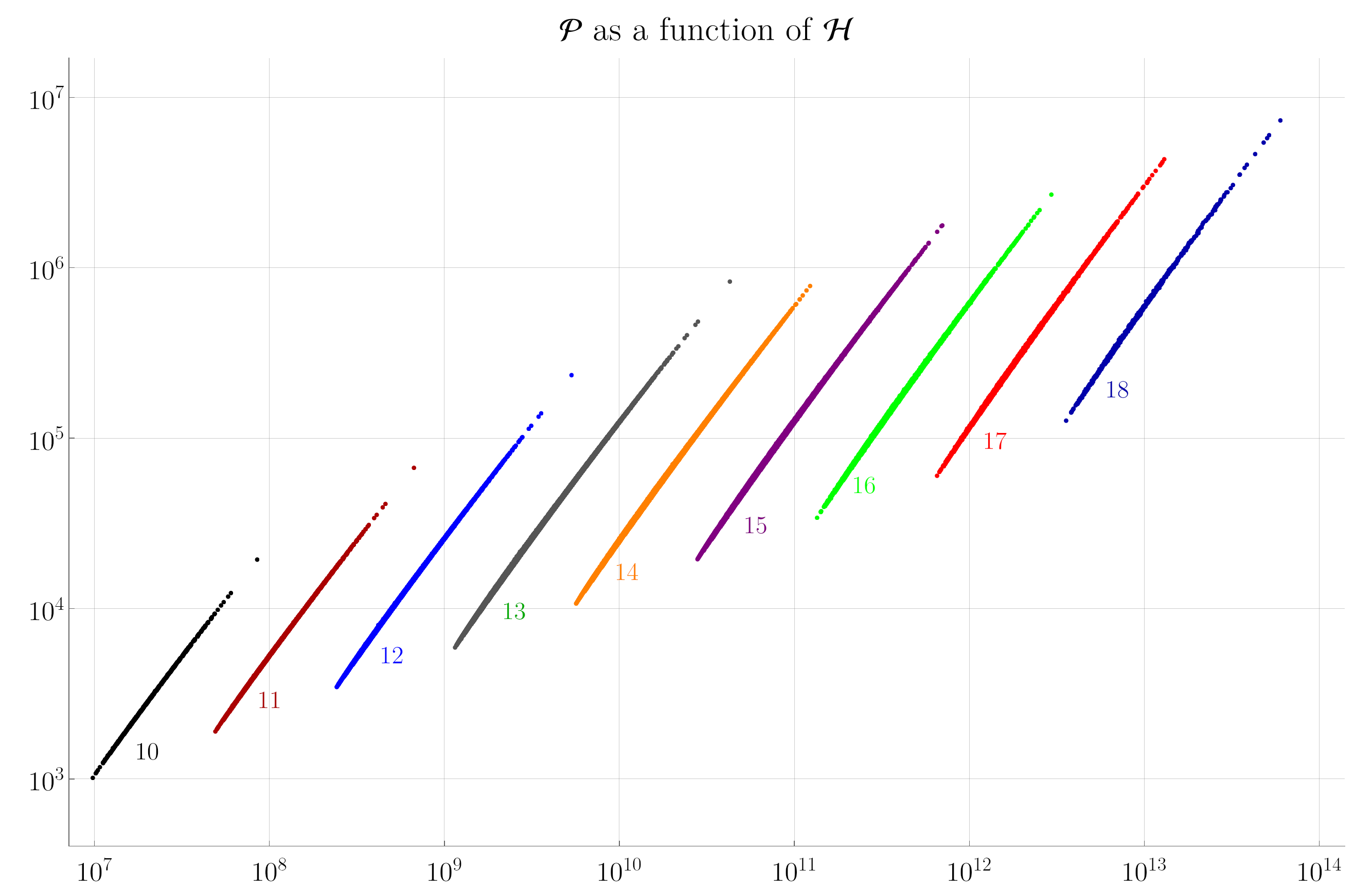}
	\caption{Hepp bound and period of primitive graphs up to 18 loops, double logarithmic plot. Within each loop order, the correlation is very close. Figure taken from \cite{balduf_statistics_2023}.}
	\label{fig:P_H}
\end{figure}

\section{Recurrence relations for bare Green functions} \label{sec:recurrence_relations}

Formulas such as \cref{lem:tropical_tree} indicate that tropical integrals can be computed recursively without enumerating all $\abs{E_\Graph}!$ terms in \cref{def:tropical_combinatorial}. Further simplifications occur if one considers the sum of all  graphs. The generating function $\mathcal G(x,t,\kappa)$ of 1PI graphs, where the variable $x$ counts legs, $t$ counts loops, and $\kappa$ counts mass insertions, is the quantum effective potential, that is, the zero-momentum piece of the quantum effective action \cite{coleman_radiative_1973,branchina_antiferromagnetic_1999}.

As has recently been shown by Borinsky \cite{borinsky_tropicalized_2025}, the quantum effective potential of the tropical field theory satisfies a partial differential equation, the \emph{tropical loop equation}. 
In the present section, we  generalize the tropical loop equation  to the $O(N)$ symmetric $\phi^4$ theory, and  we discuss its relation to zero-dimensional field theory.

\subsection{Recurrence of 1PI graphs}\label{sec:recurrence}

When computing amplitudes, we need to take into account symmetry factors of graphs. The only 1PI graphs at 1 loop order are the cycles (\cref{ex:tropical_cycle}).  The $2n$ external legs give a factor $(2n)!$, which is to be divided by the number of graph  automorphisms. Any cycle, even the boundary cases $n\in \lbrace 1,2 \rangle$, have $2^n$   ways to flip the two external legs adjacent to every vertex, and   $n$  cyclic permutations, and an extra factor $2$ for reversing the direction of the cycle. The 1-loop $(2n)$-point 1PI amplitudes are $\Gamma^{(1)}_0=0$ and $\forall n\geq 1: $
\begin{align}\label{recursion_1loop}
	\Gamma_{2n}^{(1)} &=\frac{(2n)!}{n \cdot 2^{n+1}}\tropicalintegral[r_n]= \frac{(2n)!}{2^{n+1}(\epsilon+n-2)}.
\end{align}
\begin{example}\label{ex:gamma_1loop}
	Using \cref{recursion_1loop} and \cref{ex:tropical_cycle}, one finds
	\begin{align*}
		\Gamma^{(1)}_2&=\frac{1}{2}\frac{1}{\epsilon-1}, \qquad
		\Gamma^{(1)}_4=3\frac{1}{\epsilon}, \qquad
		\Gamma^{(1)}_6=45\frac{1}{\epsilon+1}, \qquad 
		\Gamma^{(1)}_8=1260 \frac{1}{\epsilon+2}.
	\end{align*}
	For the recursive construction in the present section, it is irrelevant whether or not amplitudes have poles at $\epsilon=0$. We will discuss this aspect in detail in \cref{sec:renormalization}.
\end{example}

Tropical 1PI amplitudes at higher loop order are based upon the recurrence in \cref{lem:tropical_tree}, 
\begin{align}\label{tropical_combinatorial_one}
	\tropicalintegral[\Graph] &= \frac{1}{\omega_\Graph}\sum_{e \in  E_\Graph} \tropicalintegral[\Graph \setminus e].
\end{align}
Recall that a 1PI graph $\Graph$ stays  connected if any one edge $e\in E_\Graph$ is removed.  Hence, every 1PI graph has the structure of a cycle of 1PI subgraphs, joined by one or multiple edges $\left \lbrace e,f_1,\ldots \right \rbrace $. This mechanism is shown in \cref{fig:1PI_decomposition}; for a fixed $\Graph$, different choices of the first edge $e$ to remove give rise to different (but potentially isomorphic) such decompositions.

\begin{figure}[htbp]
	\begin{center}
	 
		\begin{tikzpicture} 
			\node  (x) at (0,0){};
			\node at ($(x)+(2.2,-1.8)$){\textbf{(a)}};
			\node [circle,draw=black,fill=lightgray,inner sep=1.5mm](gamma1) at ($(x) + (0:1 ) $) {$\gamma_1$};
			\node [circle,draw=black,fill=lightgray,inner sep=2mm](gamma2) at ($(x) + (120:.8) $) {$\gamma_2$};
			\node [circle,draw=black,fill=lightgray,inner sep=.8mm](gamma3) at ($(x) + (240:.8) $) {$\gamma_3$};
		
			\draw[edge,dashed] (gamma1)  to node[pos=.4,above]{$e_1$} (gamma2);
			\draw[edge ] (gamma2) to node[left]{$f_1$} (gamma3);
		\draw[edge ] (gamma3) to node[pos=.6,below]{$f_2$} (gamma1);
			\draw[edge] (gamma1) .. controls +( 40:1.5) and +(-40:1.5) ..node[pos=.2,above]{$e_2$}  (gamma1);

			\node  (x) at (4,0){};
			\node at ($(x)+(180:1.5)$){$\rightarrow$};
			\node [circle,draw=black,fill=lightgray,inner sep=2.0mm](gamma1) at ($(x) + (0:1 ) $) {$\gamma_1'$};
			\node [circle,draw=black,fill=lightgray,inner sep=2mm](gamma2) at ($(x) + (120:.8) $) {$\gamma_2$};
			\node [circle,draw=black,fill=lightgray,inner sep=.8mm](gamma3) at ($(x) + (240:.8) $) {$\gamma_3$};
			
			\draw[edge ] (gamma2) to node[left]{$f_1$} (gamma3);
			\draw[edge ] (gamma3) to node[pos=.6,below]{$f_2$} (gamma1);

			\node  (x) at (8.5,0){};
			\node at ($(x)+(2,-1.8)$){\textbf{(b)}};
			\node [circle,draw=black,fill=lightgray,inner sep=1.5mm](gamma1) at ($(x) + (0:.9 ) $) {$\gamma_1$};
			\node [circle,draw=black,fill=lightgray,inner sep=2mm](gamma2) at ($(x) + (120:.8) $) {$\gamma_2$};
			\node [circle,draw=black,fill=lightgray,inner sep=.8mm](gamma3) at ($(x) + (240:.8) $) {$\gamma_3$};
			
			\draw[edge] (gamma1)  to node[pos=.4,above]{$e_1$} (gamma2);
			\draw[edge ] (gamma2) to node[left]{$f_1$} (gamma3);
			\draw[edge ] (gamma3) to node[pos=.6,below]{$f_2$} (gamma1);
			\draw[edge,dashed] (gamma1) .. controls +( 40:1.5) and +(-40:1.5) ..node[pos=.2,above]{$e_2$}  (gamma1);

			\node  (x) at (11.8,0){};
			\node at ($(x)+(180:.9)$){$\rightarrow$};
			\node [circle,draw=black,fill=lightgray,inner sep=2.5mm](gamma1) at ($(x) + (0:.2) $) {$\gamma$}; 
			

		\end{tikzpicture}
	\end{center}
	\caption{A graph with 1PI subgraphs $\gamma_j$. \textbf{(a)}: Removing the edge $e_1$ creates a graph with two bridges $f_1,f_2$. The edge $e_2$ plays no special role and becomes part of a 1PI component $\gamma_1'=\gamma_1\cup e_2$. This choice of edge $e$ induces  a cycle with three 1PI components $\left \lbrace \gamma_1,\gamma_2,\gamma_3 \right \rbrace$.   \textbf{(b)} For the same graph, choosing $e_2$ as the first edge to remove means that the remaining graph $\gamma=\gamma_1\cup \gamma_2\cup \gamma_3\cup e_1\cup  e_2\cup f_1\cup f_2$ is still 1PI. Hence, this choice produces a cycle with only a single 1PI component. }
	\label{fig:1PI_decomposition}
\end{figure}
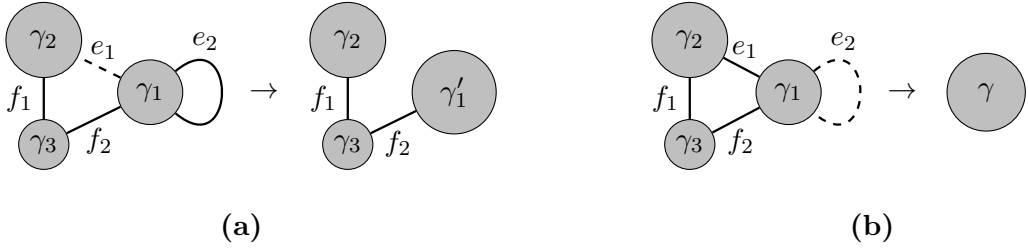

The terms in the sum \cref{tropical_combinatorial_one} can be identified with the decompositions into 1PI components as shown in \cref{fig:1PI_decomposition}. On the other hand, by the third point of \cref{lem:tropical_tree}, the tropical integral of the decomposed graph is the product of its 1PI components $\gamma_j$. In the sum over all graphs, the subgraphs $\gamma_j$ are likewise sums over all graphs with certain number of external legs, and one obtains a recurrence for the full tropical amplitude.

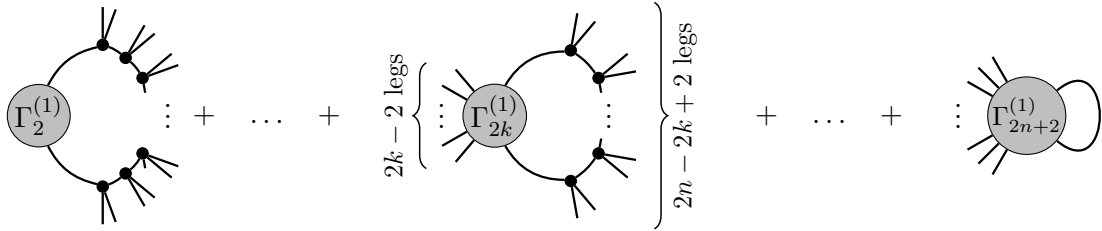
\begin{figure}[htbp]
	\begin{center}
		\begin{tikzpicture} 
			\node  (x) at (0,0){};
			\node [circle,draw=black,fill=lightgray,inner sep=0](gamma) at ($(x) + (-.5, 0) $) {$\Gamma^{(1)}_2$};
			\node [vertex](v1) at ($(x) + (70:1) $) {};
			\node [vertex](v2) at ($(x) + (50:1) $) {};
			\node [vertex](v3) at ($(x) + (30:1) $) {};
			\node [vertex](v4) at ($(x) + (-30:1) $) {};
			\node [vertex](v5) at ($(x) + (-50:1) $) {};
			\node [vertex](v6) at ($(x) + (-70:1) $) {};
			\draw[edge, bend angle=25,bend left] (gamma) to (v1);
			\draw[edge, bend angle=10,bend left] (v1) to (v2);
			\draw[edge, bend angle=10,bend left] (v2) to (v3);
			\draw[edge, bend angle=10,bend left] (v4) to (v5);
			\draw[edge, bend angle=10,bend left] (v5) to (v6);
			\draw[edge, bend angle=25,bend left] (v6) to (gamma);
			\draw[edge] (v1) -- +(90:.5);
			\draw[edge] (v1) -- +(70:.5);
			\draw[edge] (v2) -- +(60:.5);
			\draw[edge] (v2) -- +(40:.5);
			\draw[edge] (v3) -- +(40:.5);
			\draw[edge] (v3) -- +(20:.5);
			\draw[edge] (v3) -- +(-80:.2);
			\node  at ($(x) + (1.2, .1) $) {$\vdots$};
			\draw[edge] (v4) -- +(-80:.2);
			\draw[edge] (v4) -- +(-20:.5);
			\draw[edge] (v4) -- +(-40:.5);
			\draw[edge] (v5) -- +(-40:.5);
			\draw[edge] (v5) -- +(-60:.5);
			\draw[edge] (v6) -- +(-70:.5);
			\draw[edge] (v6) -- +(-90:.5);

			\node at (2.5,0){$+\quad \ldots \quad +$};

			\node  (x) at (6,0){};
			\node [circle,draw=black,fill=lightgray,inner sep=0](gamma) at ($(x) + (-.5, 0) $) {$\Gamma^{(1)}_{2k}$};
			\node [vertex](v1) at ($(x) + (60:1) $) {};
			\node [vertex](v2) at ($(x) + (30:1) $) {};
			\node [vertex](v3) at ($(x) + (-30:1) $) {};
			\node [vertex](v4) at ($(x) + (-60:1) $) {};
			\draw[edge, bend angle=30,bend left] (gamma) to (v1);
			\draw[edge, bend angle=10,bend left] (v1) to (v2);
			\draw[edge, bend angle=10,bend left] (v3) to (v4);
			\draw[edge, bend angle=30,bend left] (v4) to (gamma);
			\draw[edge] (v1) -- +(80:.5);
			\draw[edge] (v1) -- +(50:.5);
			\draw[edge] (v2) -- +(40:.5);
			\draw[edge] (v2) -- +(10:.5);
			\draw[edge] (v2) -- +(-80:.2);
			\node  at ($(x) + (1, .1) $) {$\vdots$};
			\draw[edge] (v3) -- +(80:.2);
			\draw[edge] (v3) -- +(-10:.5);
			\draw[edge] (v3) -- +(-40:.5);
			\draw[edge] (v4) -- +(-50:.5);
			\draw[edge] (v4) -- +(-80:.5);
			
			\draw[edge] (gamma) -- +(130:.8);
			\draw[edge] (gamma) -- +(150:.8);
			\node  at ($(gamma) + (-.7, .1) $) {$\vdots$};
			\draw[edge] (gamma) -- +(210:.8);
			\draw[edge] (gamma) -- +(230:.8);
			
			\draw [line width=.2mm, decorate,decoration={brace,amplitude=5pt}] ($(x)+ (-1.4,-.7)$) -- ($(x)+ (-1.4,.7)$);
			\node[font=\small,rotate=90]  at ($(x)+ (-1.8,0)$) { $2k-2$  legs };
			
			\draw [line width=.2mm, decorate,decoration={brace,amplitude=5pt}] ($(x)+ (1.6,1.5)$) -- ($(x)+ (1.6,-1.5)$);
			\node[font=\small, rotate=90]  at ($(x)+ (2,0)$) { $2n-2k+2$   legs };

			\node at (9.9,0){$+\quad \ldots \quad +$};
			
			\node  (x) at (13,0){};
			\node [circle,draw=black,fill=lightgray,inner sep=0,font=\small](gamma) at ($(x) + (-.5, 0) $) {$\Gamma^{(1)}_{2n+2}$};
			\draw[edge] (gamma) .. controls +( 40:1.5) and +(-40:1.5) .. (gamma);
			\draw[edge] (gamma) -- +(120:.9);
			\draw[edge] (gamma) -- +(135:.9);
			\draw[edge] (gamma) -- +(150:.9);
			\node  at ($(gamma) + (-.9, .1) $) {$\vdots$};
			\draw[edge] (gamma) -- +(210:.9);
			\draw[edge] (gamma) -- +(225:.9);
			\draw[edge] (gamma) -- +(240:.9);
			
		\end{tikzpicture}
	\end{center}
	\caption{A $2n$-point 2-loop amplitude is computed from summing all ways to add one more loop to a $2k$-point 1-loop amplitude $\Gamma^{(1)}_{2k}$. The extreme cases are $k=1$, in which case $\Gamma^{(1)}_2$ is a 1PI propagator correction (left),  and $k=n+1$, in which case the added arc degenerates to a single edge (right).}
	\label{fig:recursion_2loop}
\end{figure}

At $\loopnumber>1$ loops, one obtains symmetry factors for permutations of external legs similar to the 1-loop case, just that the components $\gamma_j$ might have any even number of external legs, compare \cref{fig:recursion_2loop}. For a given $k$, the result obtains a factor $(n+2-k)$ for the number of edges in the cycle, i.e. number of equivalent terms in \cref{tropical_combinatorial_one}, a combinatorial factor of $\frac{(2n)!}{(2!)^{n+1-k}\cdot (2k-2)!}$ for permutations of external legs, and $\frac 12$ for reversion of the cycle. Summing over all values of $k$, the tropical 2-loop amplitude is
\begin{align}\label{recurrence_2_loop}
	\Gamma^{(2)}_{2n} 
	&= \frac{(2n)!}{ 2^{n+2}(2 \epsilon +n-2)} \sum_{k=1}^{n+1} \frac{k (2k-1)(n+2-k) }{ \epsilon+k-2  }, \qquad n \geq 0 . 
\end{align}

\begin{example}\label{ex:gamma_2loop}
	The 2-loop 2-point function has $2n=2$, so $n=1$, and the summands $\frac{1}{2^{-k}(2k-2)!}\times (n+2-k)\times \Gamma^{(1)}_{2k}$ are:
	\begin{align*}
		\Gamma^{(2)}_{2} &= \frac{2}{8(2\epsilon-1)}\Big( \underbrace{2 \times 2 \times \frac{1}{2(\epsilon-1)} }_{\text{ for }k=1 } + \underbrace{2 \times 1 \times \frac 3 \epsilon }_{\text{for }k=2} \Big)   = \frac{4\epsilon-3}{2 \epsilon(\epsilon-1)(2\epsilon-1)}.  
	\end{align*}
	This result coincides with \cref{ex:tropical_twopoint_twoloop}. Note that the  two terms in the sum over $k$ do not   correspond to the two 1PI graphs contributing to $\Gamma^{(2)}_2$ in \cref{ex:tropical_twopoint_twoloop}, in general there are many more Feynman graphs than terms in the recurrence. Likewise one finds 
	\begin{align*}
		\Gamma^{(2)}_0 &=  \frac 1 {2\epsilon-2} \cdot \frac 12 \cdot \Gamma^{(1)}_2 = \frac 1 {8 (\epsilon-1)^2}\\
		\Gamma^{(2)}_4,
		&= \frac{1}{2\epsilon}\left( 9 \Gamma^{(1)}_2  + 6 \Gamma^{(1)}_4+\frac 1 2 \Gamma_6^{(1)}   \right) = \frac{9(5\epsilon^2-2\epsilon-2)}{2\epsilon^2(\epsilon-1)(\epsilon+1)},\\
		\Gamma^{(2)}_6 &= \frac{1}{2\epsilon+1}\left(     180 \Gamma^{(1)}_2  + 135 \Gamma^{(1)}_4+ 15 \Gamma^{(1)}_6+\frac 1 2 \Gamma^{(1)}_8 \right)   =\frac{45(40 \epsilon^3+39 \epsilon^2 -49 \epsilon-18)}{\epsilon(\epsilon-1)(\epsilon+1)(\epsilon+2)(2\epsilon+1)} .
	\end{align*}
\end{example}

At higher loop order, the possible distributions of $2n$ external legs to $k$ subgraphs $\gamma_j$ are equivalent to partitioning $n$ indistinguishable objects into $k$ (potentially empty) disjoint subsets. These possibilities are give by   the incomplete Bell polynomial \cite{comtet_advanced_1974}
\begin{align}\label{Bell_partitions}
	 \frac{k!}{2k}\cdot \frac{1}{(n+k)!}B_{n+k,k}\big( 1!\Gamma_2,~2! \Gamma_4, ~3! \Gamma_6 ~4! \Gamma_8, ~5!\Gamma_{10},~ \ldots \big).
\end{align}
The first factor occurs because the Bell polynomials refer to subsets which can be permuted in all $k!$ ways,  whereas in our case, the subgraphs $\gamma_j$ only enjoy cyclic symmetry. 

\begin{example}
	Consider the 8-point function, $2n=8$, and the partitions into $k=3$ 1PI functions. The Bell polynomial (\cref{Bell_partitions}) is
	\begin{align*}
		\frac{3!}{2\times 3 \times (4+3)!}B_{4+3, 3}\left(\Gamma_{2}, 2\Gamma_{4}, 6 \Gamma_{6}, 24 \Gamma_{8}, \ldots \right) &= \frac 12 \Gamma_{4}^2 \Gamma_{8} + \frac 12 \Gamma_{2} \Gamma_{6}^2 + \Gamma_{2} \Gamma_{4} \Gamma_{8} + \frac 12\Gamma_{2}^2 \Gamma_{10}.
	\end{align*}
	Graphically, these terms correspond to the following arrangements:
	\begin{center}
		\begin{tikzpicture}
			\node  (x) at (0,0){};
			\node at ($(x)+(-1.1,.1)$){$\frac 12$};
			\node [vertex](v1) at ($(x) + (90:.45) $) {};
			\node [vertex](v2) at ($(x) + (210:.45) $) {};
			\node [vertex](v3) at ($(x) + (330:.45) $) {};
			\draw[edge,bend angle =30, bend right] (v1) to (v2);
			\draw[edge,bend angle =30, bend right] (v2) to (v3);
			\draw[edge,bend angle =30, bend right](v3) to (v1);
			\draw[edge] (v1) -- +(40:.4);
			\draw[edge] (v1) -- +(70:.4);
			\draw[edge] (v1) -- +(110:.4);
			\draw[edge] (v1) -- +(140:.4);
			\draw[edge] (v2) -- +(190:.4);
			\draw[edge] (v2) -- +(230:.4);
			\draw[edge] (v3) -- +(310:.4);
			\draw[edge] (v3) -- +(350:.4);
			
			\node  (x) at (3,0){};
			\node at ($(x)+(-1.25,.1)$){$+~~\frac 12$};
			\node [vertex](v1) at ($(x) + (90:.45) $) {};
			\node [vertex](v2) at ($(x) + (210:.45) $) {};
			\node [vertex](v3) at ($(x) + (330:.45) $) {};
			\draw[edge,bend angle =30, bend right] (v1) to (v2);
			\draw[edge,bend angle =30, bend right] (v2) to (v3);
			\draw[edge,bend angle =30, bend right](v3) to (v1);
			\draw[edge] (v1) -- +(40:.4);
			\draw[edge] (v1) -- +(70:.4);
			\draw[edge] (v1) -- +(110:.4);
			\draw[edge] (v1) -- +(140:.4);
			\draw[edge] (v3) -- +(280:.4);
			\draw[edge] (v3) -- +(310:.4);
			\draw[edge] (v3) -- +(350:.4);
			\draw[edge] (v3) -- +(390:.4);
			
			\node  (x) at (6,0){};
			\node at ($(x)+(-1.25,.1)$){$+~~1$};
			\node [vertex](v1) at ($(x) + (90:.45) $) {};
			\node [vertex](v2) at ($(x) + (210:.45) $) {};
			\node [vertex](v3) at ($(x) + (330:.45) $) {};
			\draw[edge,bend angle =30, bend right] (v1) to (v2);
			\draw[edge,bend angle =30, bend right] (v2) to (v3);
			\draw[edge,bend angle =30, bend right](v3) to (v1);
			\draw[edge] (v1) -- +(40:.4);
			\draw[edge] (v1) -- +(60:.4);
			\draw[edge] (v1) -- +(80:.4);
			\draw[edge] (v1) -- +(100:.4);
			\draw[edge] (v1) -- +(120:.4);
			\draw[edge] (v1) -- +(140:.4);
			\draw[edge] (v3) -- +(310:.4);
			\draw[edge] (v3) -- +(350:.4);
			
			\node  (x) at (9,0){};
			\node at ($(x)+(-1.25,.1)$){$+~~\frac 12$};
			\node [vertex](v1) at ($(x) + (90:.45) $) {};
			\node [vertex](v2) at ($(x) + (210:.45) $) {};
			\node [vertex](v3) at ($(x) + (330:.45) $) {};
			\draw[edge,bend angle =30, bend right] (v1) to (v2);
			\draw[edge,bend angle =30, bend right] (v2) to (v3);
			\draw[edge,bend angle =30, bend right](v3) to (v1);
			\draw[edge] (v1) -- +(20:.4);
			\draw[edge] (v1) -- +(40:.4);
			\draw[edge] (v1) -- +(60:.4);
			\draw[edge] (v1) -- +(80:.4);
			\draw[edge] (v1) -- +(100:.4);
			\draw[edge] (v1) -- +(120:.4);
			\draw[edge] (v1) -- +(140:.4);
			\draw[edge] (v1) -- +(160:.4);

		\end{tikzpicture}
	\end{center}
	Indeed, the combinatorial factor is the expected symmetry factor for permutations of vertices. 	
	In general, one term in the Bell polynomial corresponds to more than one graph, as becomes apparent for the choice $k=4$, where  
	\begin{align*}
		\frac{4!}{2\times 4 \times (4+4)!}B_{4+4, 4}=\frac 1 8 \Gamma_{4}^4 + \frac 3 2\Gamma_{2} \Gamma_{4}^2 \Gamma_{6} + \frac 3 4 \Gamma_{2}^2 \Gamma_{6}^2 + \frac 3 2 \Gamma_{2}^2 \Gamma_{4} \Gamma_{8} + \frac 12 \Gamma_{2}^3 \Gamma_{10}. 
	\end{align*}	
	The second, third, and fourth term in this sum correspond to two distinct graphs, each of which  contain the same vertex valences, but arranged differently:
	\begin{center}
		\begin{tikzpicture}
			\node  (x) at (0,0){};
			\node at ($(x)+(-.8,0)$){$\frac 18$};
			\node [vertex](v1) at ($(x) + (45:.4) $) {};
			\node [vertex](v2) at ($(x) + (135:.4) $) {};
			\node [vertex](v3) at ($(x) + (225:.4) $) {};
			\node [vertex](v4) at ($(x) + (315:.4) $) {};
			\draw[edge,bend angle=20,bend right] (v1) to (v2);
			\draw[edge,bend angle=20,bend right] (v2) to (v3);
			\draw[edge,bend angle=20,bend right](v3) to (v4);
			\draw[edge,bend angle=20,bend right](v4) to (v1);
			\draw[edge] (v1) -- +(25:.3);
			\draw[edge] (v1) -- +(65:.3);
			\draw[edge] (v2) -- +(115:.3);
			\draw[edge] (v2) -- +(155:.3);
			\draw[edge] (v3) -- +(205:.3);
			\draw[edge] (v3) -- +(245:.3);
			\draw[edge] (v4) -- +(295:.3);
			\draw[edge] (v4) -- +(335:.3);
			
			\node  (x) at (3,0){};
			\node at ($(x)+(-.8,0)$){$1$};
			\node [vertex](v1) at ($(x) + (45:.4) $) {};
			\node [vertex](v2) at ($(x) + (135:.4) $) {};
			\node [vertex](v3) at ($(x) + (225:.4) $) {};
			\node [vertex](v4) at ($(x) + (315:.4) $) {};
			\draw[edge,bend angle=20,bend right] (v1) to (v2);
			\draw[edge,bend angle=20,bend right] (v2) to (v3);
			\draw[edge,bend angle=20,bend right](v3) to (v4);
			\draw[edge,bend angle=20,bend right](v4) to (v1);
			\draw[edge] (v2) -- +(115:.3);
			\draw[edge] (v2) -- +(155:.3);
			\draw[edge] (v3) -- +(205:.3);
			\draw[edge] (v3) -- +(245:.3);
			\draw[edge] (v4) -- +(285:.3);
			\draw[edge] (v4) -- +(305:.3);
			\draw[edge] (v4) -- +(325:.3);
			\draw[edge] (v4) -- +(345:.3);
			
			\node  (x) at (3,-1.5){};
			\node at ($(x)+(-.8,0)$){$\frac 12$};
			\node [vertex](v1) at ($(x) + (45:.4) $) {};
			\node [vertex](v2) at ($(x) + (135:.4) $) {};
			\node [vertex](v3) at ($(x) + (225:.4) $) {};
			\node [vertex](v4) at ($(x) + (315:.4) $) {};
			\draw[edge,bend angle=20,bend right] (v1) to (v2);
			\draw[edge,bend angle=20,bend right] (v2) to (v3);
			\draw[edge,bend angle=20,bend right](v3) to (v4);
			\draw[edge,bend angle=20,bend right](v4) to (v1);
			\draw[edge] (v1) -- +(25:.3);
			\draw[edge] (v1) -- +(65:.3);
			\draw[edge] (v3) -- +(205:.3);
			\draw[edge] (v3) -- +(245:.3);
			\draw[edge] (v4) -- +(285:.3);
			\draw[edge] (v4) -- +(305:.3);
			\draw[edge] (v4) -- +(325:.3);
			\draw[edge] (v4) -- +(345:.3);
			
			\node  (x) at (6,0){};
			\node at ($(x)+(-.8,0)$){$\frac 1 4$};
			\node [vertex](v1) at ($(x) + (45:.4) $) {};
			\node [vertex](v2) at ($(x) + (135:.4) $) {};
			\node [vertex](v3) at ($(x) + (225:.4) $) {};
			\node [vertex](v4) at ($(x) + (315:.4) $) {};
			\draw[edge,bend angle=20,bend right] (v1) to (v2);
			\draw[edge,bend angle=20,bend right] (v2) to (v3);
			\draw[edge,bend angle=20,bend right](v3) to (v4);
			\draw[edge,bend angle=20,bend right](v4) to (v1);
			\draw[edge] (v2) -- +(105:.3);
			\draw[edge] (v2) -- +(125:.3);
			\draw[edge] (v2) -- +(145:.3);
			\draw[edge] (v2) -- +(165:.3);
			\draw[edge] (v4) -- +(285:.3);
			\draw[edge] (v4) -- +(305:.3);
			\draw[edge] (v4) -- +(325:.3);
			\draw[edge] (v4) -- +(345:.3);
			
			\node  (x) at (6,-1.5){};
			\node at ($(x)+(-.8,0)$){$\frac 12$};
			\node [vertex](v1) at ($(x) + (45:.4) $) {};
			\node [vertex](v2) at ($(x) + (135:.4) $) {};
			\node [vertex](v3) at ($(x) + (225:.4) $) {};
			\node [vertex](v4) at ($(x) + (315:.4) $) {};
			\draw[edge,bend angle=20,bend right] (v1) to (v2);
			\draw[edge,bend angle=20,bend right] (v2) to (v3);
			\draw[edge,bend angle=20,bend right](v3) to (v4);
			\draw[edge,bend angle=20,bend right](v4) to (v1);
			\draw[edge] (v1) -- +(15:.3);
			\draw[edge] (v1) -- +(35:.3);
			\draw[edge] (v1) -- +(55:.3);
			\draw[edge] (v1) -- +(75:.3);
			\draw[edge] (v4) -- +(285:.3);
			\draw[edge] (v4) -- +(305:.3);
			\draw[edge] (v4) -- +(325:.3);
			\draw[edge] (v4) -- +(345:.3);
			
			\node  (x) at (9,0){};
			\node at ($(x)+(-.8,0)$){$\frac 1 2$};
			\node [vertex](v1) at ($(x) + (45:.4) $) {};
			\node [vertex](v2) at ($(x) + (135:.4) $) {};
			\node [vertex](v3) at ($(x) + (225:.4) $) {};
			\node [vertex](v4) at ($(x) + (315:.4) $) {};
			\draw[edge,bend angle=20,bend right] (v1) to (v2);
			\draw[edge,bend angle=20,bend right] (v2) to (v3);
			\draw[edge,bend angle=20,bend right](v3) to (v4);
			\draw[edge,bend angle=20,bend right](v4) to (v1);
			\draw[edge] (v2) -- +(115:.3);
			\draw[edge] (v2) -- +(155:.3);
			\draw[edge] (v4) -- +(265:.3);
			\draw[edge] (v4) -- +(285:.3);
			\draw[edge] (v4) -- +(305:.3);
			\draw[edge] (v4) -- +(325:.3);
			\draw[edge] (v4) -- +(345:.3);
			\draw[edge] (v4) -- +(365:.3);
			
			\node  (x) at (9,-1.5){};
			\node at ($(x)+(-.8,0)$){$1$};
			\node [vertex](v1) at ($(x) + (45:.4) $) {};
			\node [vertex](v2) at ($(x) + (135:.4) $) {};
			\node [vertex](v3) at ($(x) + (225:.4) $) {};
			\node [vertex](v4) at ($(x) + (315:.4) $) {};
			\draw[edge,bend angle=20,bend right] (v1) to (v2);
			\draw[edge,bend angle=20,bend right] (v2) to (v3);
			\draw[edge,bend angle=20,bend right](v3) to (v4);
			\draw[edge,bend angle=20,bend right](v4) to (v1);
			\draw[edge] (v1) -- +(25:.3);
			\draw[edge] (v1) -- +(65:.3);
			\draw[edge] (v4) -- +(265:.3);
			\draw[edge] (v4) -- +(285:.3);
			\draw[edge] (v4) -- +(305:.3);
			\draw[edge] (v4) -- +(325:.3);
			\draw[edge] (v4) -- +(345:.3);
			\draw[edge] (v4) -- +(365:.3);
			
			\node  (x) at (12,0){};
			\node at ($(x)+(-.8,0)$){$\frac 12$};
			\node [vertex](v1) at ($(x) + (45:.4) $) {};
			\node [vertex](v2) at ($(x) + (135:.4) $) {};
			\node [vertex](v3) at ($(x) + (225:.4) $) {};
			\node [vertex](v4) at ($(x) + (315:.4) $) {};
			\draw[edge,bend angle=20,bend right] (v1) to (v2);
			\draw[edge,bend angle=20,bend right] (v2) to (v3);
			\draw[edge,bend angle=20,bend right](v3) to (v4);
			\draw[edge,bend angle=20,bend right](v4) to (v1);
			\draw[edge] (v4) -- +(245:.3);
			\draw[edge] (v4) -- +(265:.3);
			\draw[edge] (v4) -- +(285:.3);
			\draw[edge] (v4) -- +(305:.3);
			\draw[edge] (v4) -- +(325:.3);
			\draw[edge] (v4) -- +(345:.3);
			\draw[edge] (v4) -- +(365:.3);
			\draw[edge] (v4) -- +(385:.3);
			
		\end{tikzpicture}
	\end{center}
	
\end{example} 

To turn \cref{Bell_partitions} into a tropical amplitude, we  multiply by $(2n)!$ and divide each $\Gamma_{2j+2}$ by $(2j)!$ for the permutations of its external legs.  As in \cref{recurrence_2_loop}, there is a factor $k$ for the $k$ edges in the cycle that contribute to \cref{tropical_combinatorial_one}, and an overall factor $\frac 1 {\omega_\Graph}=\frac 1{\loopnumber \epsilon-2+n}$.

\begin{proposition}\label{thm:gamma_recurrence}  For $\loopnumber>1$ and $n \geq 0$, using $\Gamma_j \defas \delta_{j4}+\sum_{l\geq 1} \Gamma^{(l)}_j$, the tropical $\loopnumber$-loop $2n$-point 1PI amplitude is given by
	\begin{align*}
		\Gamma^{(\loopnumber)}_{2n} &= \frac{(2n)!}{2(\loopnumber \epsilon -2+n)}\left[ t^{\loopnumber-1} \right]\sum_{k=1}^{\loopnumber-1+n} \frac{k!}{(n+k)!}   B_{n+k,k} \left( \Gamma_2,~\Gamma_4,  ~\frac{\Gamma_6}{4}, ~\frac{\Gamma_8}{30}, ~\ldots , ~\frac{(n+1)! \Gamma_{2n+2}}{(2n)!}  \right)  .
	\end{align*}
\end{proposition} 
\Cref{thm:gamma_recurrence} reproduces \cref{recurrence_2_loop} for $\loopnumber=2$. Furthermore, by explicit calculation, one finds that it is valid even for $\loopnumber=1$ where it reproduces \cref{recursion_1loop}.

\begin{example}\label{ex:gamma_3loop}
	At $\loopnumber=3$ loops, the formula of \cref{thm:gamma_recurrence} with $n=1$ produces
	\begin{align*}
		\Gamma^{(3)}_2 &= \scalemath{.9}{ \frac{2}{3\epsilon-1}\sum_{k=1}^3 \frac{k!}{2(1+k)!}\left[ t^2 \right] B_{1+k,k} \left( \Gamma_2,   \ldots  \right) = \frac{1}{3\epsilon-1}\left( \frac 12 \Gamma^{(2)}_4 + 1 \Gamma^{(1)}_4 \Gamma^{(1)}_2 + 1 \Gamma^{(2)}_2 + \frac 3 2 \Gamma^{(1)}_2 \Gamma^{(1)}_2 \right) }\\
		&= \frac{226 \epsilon^4 -363\epsilon^3 +83 \epsilon^2 +96 \epsilon -36} {8\epsilon^2 (\epsilon-1)^2(\epsilon+1)(2\epsilon-1)(3\epsilon-1)}.
	\end{align*}	
	For the 3-loop 4-point function, one sets  $n=2, \loopnumber=3$ in \cref{thm:gamma_recurrence} and finds
	\begin{align*}
		\Gamma^{(3)}_4 
		&= \frac{1}{3\epsilon }\left( 18 \left( \Gamma^{(1)}_2 \right) ^2 + 9 \Gamma^{(2)}_2 + 18 \Gamma_2^{(1)} \Gamma_4^{(1)} + 3 \left( \Gamma_4^{(1)} \right) ^2 + 6 \Gamma_4^{(2)} + \Gamma_2^{(1)} \Gamma_6^{(1)} + \frac 12 \Gamma_6^{(2)} \right)  \\
		&= \frac{3(300 \epsilon^6 -10 \epsilon^5 -701 \epsilon^4 + 331 \epsilon^3 +146 \epsilon^2 -39 \epsilon-18)}{\epsilon^3(\epsilon-1)^2(\epsilon+1)(\epsilon+2)(2\epsilon-1)(2\epsilon+1)}.
	\end{align*}
\end{example}

The next step is to generalize  the recursion formula of \cref{thm:gamma_recurrence}   to include a tree-level 2-valent vertex with Feynman rule $\kappa$, i.e. a mass term.  Graphically, this amounts to replacing every edge by a geometric sum in $\kappa$. 
\begin{example} \label{ex:1loop_massive}
	At one loop, the tropical amplitude for fixed number of legs is a sum over all cycles.According to \cref{ex:tropical_cycle}, the two point function becomes	
	\begin{align*}
		\Gamma^{(1)}_2&=	\tikz[baseline=-2mm]{
			\node  (x) at (0,0){};
			\node [vertex](v1) at ($(x) + (270:.4) $) {}; 
			\draw[edge] (v1) .. controls +( 130:1.2) and +(50:1.2) .. (v1);
			\draw[edge] (v1) -- +(0:.5);
			\draw[edge] (v1) -- +(180:.5);
			\node  (x) at (2,0){};
			\node at ($(x)+(-1,0)$){$+ ~\kappa $};
			\node [vertex](v1) at ($(x) + (270:.4) $) {};
			\node [vertex](v2) at ($(x) + (90:.3) $) {};
			\draw[edge,bend angle=50,bend right, in=-100] (v1) to (v2);
			\draw[edge,bend angle=50,bend right,out=-80] (v2) to (v1);
			\draw[edge] (v1) -- +(0:.5);
			\draw[edge] (v1) -- +(180:.5);
			\node  (x) at (4,0){};
			\node at ($(x)+(-1,0)$){$+ ~\kappa^2 $};
			\node [vertex](v1) at ($(x) + (270:.4) $) {};
			\node [vertex](v2) at ($(x) + (30:.3) $) {};
			\node [vertex](v3) at ($(x) + (150:.3) $) {};
			\draw[edge,bend angle=35,bend right ] (v1) to (v2);
			\draw[edge,bend angle=50,bend right ] (v2) to (v3);
			\draw[edge,bend angle=35,bend right ] (v3) to (v1);
			\draw[edge] (v1) -- +(0:.5);
			\draw[edge] (v1) -- +(180:.5);
			\node  (x) at (6,0){};
			\node at ($(x)+(-1,0)$){$+ ~\kappa^3 $};
			\node [vertex](v1) at ($(x) + (270:.4) $) {};
			\node [vertex](v2) at ($(x) + (10:.3) $) {};
			\node [vertex](v3) at ($(x) + (90:.3) $) {};
			\node [vertex](v4) at ($(x) + (170:.3) $) {};
			\draw[edge,bend angle=25,bend right ] (v1) to (v2);
			\draw[edge,bend angle=30,bend right ] (v2) to (v3);
			\draw[edge,bend angle=30,bend right ] (v3) to (v4);
			\draw[edge,bend angle=25,bend right ] (v4) to (v1);
			\draw[edge] (v1) -- +(0:.5);
			\draw[edge] (v1) -- +(180:.5);
			\node  (x) at (8,0){};
			\node  at ($(x)+(-.9,-.1)$){$\displaystyle +~ \ldots $};
		}\\
	 &= \sum_{j=0}^\infty \frac{(j+1) \kappa^j}{2(\epsilon-1+j)}=\frac{1}{2-2\kappa} - \frac{ (\epsilon-2) ~_2F_1\left(   1,\epsilon-1;\epsilon;\kappa\right)}{2(\epsilon-1)}.
	\end{align*}
		\noindent 
		When $\kappa \rightarrow 0$, this reproduces the massless result $\frac{1}{2(\epsilon-1)}$ of \cref{ex:gamma_1loop}. However, when $\kappa \neq 0$, the function has a simple pole at $\epsilon=0$. This pole is a UV-singularity of the mass, to be discussed  in \cref{sec:longrange_divergences}. 		
		For the four-point amplitude, one has to sum over all ways to distribute arbitrarily many $\kappa$-vertices over the two internal edges, this yields
		\begin{align*}
		\Gamma^{(1)}_4 &= \frac 3 2 \sum_{j=0}^\infty \kappa^j \sum_{r=0}^j \frac{j+2}{\epsilon+j}= \frac 3 2 \sum_{j=0}^\infty  \frac{(j+1)(j+2) \kappa^j }{\epsilon+j} \\
		&= \scalemath{.9}{ \frac{4+\kappa-2\kappa^2+\epsilon(\kappa^2+\kappa-2)}{2\kappa(\kappa-1)^2} + \frac{\epsilon-2}{2}\left(   \frac{2}{\kappa}~_2F_1\left( \epsilon-1,1;\epsilon;\kappa \right) +   \frac{\epsilon-1}{\epsilon} ~_2F_1\left( \epsilon,1;\epsilon+1;\kappa \right)   \right)  .}
		\end{align*}
		Despite the denominator, $\lim_{\kappa\rightarrow 0}\Gamma^{(1)}_4=\frac 3 \epsilon$ reproduces the massless case.	
\end{example}

\begin{proposition}\label{thm:gamma_recurrence2}
	Let $\Gamma^{(\loopnumber)}_{2n,s}=[\kappa^s]\Gamma^{(\loopnumber)}_{2n}$ be the $2n$-point function with exactly $s$ mass vertices, then 
	\begin{align*}
		&\Gamma^{(\loopnumber)}_{2n,s}  = \frac{(2n)!}{2(\loopnumber \epsilon +n +s-2)}\cdot \\
		&\quad \cdot \left[ t^{\loopnumber-1} \cdot \kappa^s \right]\left( \frac{1}{n!}\sum_{r=0}^{n} \frac{r!}{\left( 1-\Gamma_2 \right) ^{r+1}}  B_{n,r} \left(  \frac{\Gamma_4}2,  ~\frac{\Gamma_6}{12}, ~\frac{\Gamma_8}{120}, ~\ldots , ~\frac{n! \Gamma_{2n+2}}{(2n)!}  \right) -1\right) .\nonumber
	\end{align*}
\end{proposition}
\begin{proof}
	With the usual definitions $0!=1$, $B_{0,0}=1$, and $B_{0,k\geq 1}=0$,    \cite[Sec.~3.3~eq.~3L]{comtet_advanced_1974} 
	\begin{align}\label{Bell_n_k_shift}
		B_{n+k,k} \left( x_1, x_2, \ldots  \right)  &= \sum_{s=0}^k    \frac{(n+k)!}{n! s!}x_1^s B_{n , k-s} \left( \frac {x_2}2, \frac{x_3}3,\frac{x_4}{4}, \ldots \right) , \qquad n\geq 0.
	\end{align}
	With $x_j=\frac{j!\Gamma_{2j}}{(2j-2)!}$ and the shift of indices $k-s:=r$, the sum in \cref{thm:gamma_recurrence} becomes
	\begin{align}\label{Bell_sum_shifted}
		\sum_{k=1}^{\loopnumber-1+n} \frac{k!  }{(n+k)!}  B_{n+k,k} \left( x_1, \ldots \right)&= \frac{1}{n!}\sum_{k=1}^{\loopnumber-1+n} \sum_{r=0}^k    \frac{k!}{ (k-r)!}x_1^{k-r} B_{n,r} \left( \frac {x_2}2, \frac{x_3}3,  \ldots \right).
	\end{align}
	The upper limit $n+\loopnumber-1$ of the summation index $k$   needs to be set to infinity when 2-valent vertices are present. 
	As the Bell polynomial $B_{n,r}$ vanishes for $r>n$, the inner sum in \cref{Bell_sum_shifted} can end at $n$  and we exchange the sums according to
	\begin{align*}
		\sum_{k=1}^{\infty} \sum_{r=0}^k   &=  \sum_{r=0}^{n}  \sum_{k=r}^{\infty}  -\text{ summand at} (r=0,k=0).
	\end{align*}
	The only term   that contributes at $(r=0)$ is $n=0$ since $B_{n,0}=0$ when $n>0$.   Using $B_{0,0}=1$, and shifting the inner summation index $k-r\rightarrow k$, \cref{Bell_sum_shifted} becomes
	\begin{align*}
		\frac{1}{n!}\sum_{r=0}^{n}   B_{n,r} \left( \frac {x_2}2, \frac{x_3}3,  \ldots \right)\sum_{k=0}^{ \infty }    \frac{(k+r)!}{ k!}x_1^{k} ~- 1.
	\end{align*} 
	In this form, the index $r$ counts the number of those 1PI graphs $\gamma_j$ in the cycle which contribute to external legs, while the sum over $k$ represents the insertion of 2-point functions,
	\begin{align}\label{full_propagator_sum}
		\sum_{k=0}^{\infty }    \frac{(k+r)!}{ k!}\Gamma_2^{k}= \frac{r!}{\left( 1-\Gamma_2 \right) ^{r+1}}.
	\end{align}

	In the recurrence formula of \cref{thm:gamma_recurrence}, we used that every massless $\loopnumber$-loop $2n$-point graph has superficial degree of convergence (\cref{def:sdd}) of $\omega=\loopnumber \epsilon -2+n$, which is independent of the choice of 1PI subgraphs $\gamma_j$ that furnish the ring. This is no longer true if 2-valent vertices are allowed, because each such vertex increases $\omega$ by one unit without changing either the loop order or the external  legs.  Consequently, we need a new index $s$ to keep track of the number of  2-valent vertices in an amplitude, that is, its order in $\kappa$.
\end{proof}

Recall that for an arbitrary power series $g(x) = \sum_{n=1}^\infty g_n x^n$, and for $g_0\neq 0$ and a negative constant $m\in \mathbb R$, the Bell polynomials are the expansion coefficients according to
\begin{align*}
	\big( g_0 + g(x) \big) ^m 
	&=  \sum_{n=0}^\infty   \frac{x^n}{n!} \sum_{j=0}^n \frac{(-1)^j\Gamma(-m+j)}{\Gamma(-m)} g_0^{m-j} B_{n,j} \big( 1! g_1, 2! g_2, \ldots, (n+1-j)! g_{n+1-j} \big). 
\end{align*}
With the choice $g_j=\Gamma_{2j+2}/(2j)!$ for $j\geq 1$, and $m=-1$, we reproduce \cref{thm:gamma_recurrence2}. Comparing powers of $g_0$, we find $g_0=\Gamma_2-1$. Therefore, the sum in \cref{thm:gamma_recurrence2} is
\begin{align}\label{PDE_derivation_rescaling}
	&\frac{1}{n!}\sum_{r=1}^{n} \frac{r!}{\left( 1-\Gamma_2 \right) ^{r+1}}  B_{n,r} \left(  \frac{\Gamma_4}2,  ~\frac{\Gamma_6}{12}, ~\frac{\Gamma_8}{120}, ~\ldots , ~\frac{n! \Gamma_{2n+2}}{(2n)!}  \right) \nonumber \\
	&= -\left[ x^n \right]  \left(  \Gamma_2-1 + \sum_{j=1}^\infty x^j \frac{\Gamma_{2j+2}}{(2j)!} \right) ^{-1}= \left[ x^{2n} \right]  \left(  1 -\sum_{j=0}^\infty x^{2j} \frac{\Gamma_{2j+2}}{(2j)!} \right) ^{-1}.
\end{align}
In the last step, we have replaced $x \rightarrow\sqrt x$ in order for the index $n$ to match the power of $x$. The recurrence relation of \cref{thm:gamma_recurrence2}, valid for $\loopnumber\geq 1$, now becomes
\begin{align}\label{gamma_recursion3}
	\frac{\Gamma^{(\loopnumber)}_{2n,s} }{ (2n)! } 
	&=\frac{1}{2(\loopnumber \epsilon +n +s-2) }\left[ x^{2n} \cdot t^{\loopnumber-1}\cdot \kappa^s \right]\left(  \frac 1 { \left(  1 -\sum_{j=0}^\infty x^{2j} \frac{\Gamma_{2j+2}}{(2j)!} \right)}-1 \right)    .
\end{align}
This suggests to define the generating function 
\begin{align}\label{def:G_generating_function}
	\mathcal G(x,t,\kappa):=\sum_{n=0}^\infty \sum_{\loopnumber=0}^\infty \sum_{s=0}^\infty  \Gamma^{(\loopnumber)}_{n ,s} \kappa^s t^\loopnumber \frac{x^{n}}{n!} .
\end{align}

\begin{example}\label{ex:G_generating_firstorder}
		At tree level, we have   $\mathcal G(x,t,\kappa)=\kappa \frac{x^2}{2} +\frac{x^4}{4!}+\asyO{t}$ by definition. The $\kappa$-dependent 1-loop coefficients have been computed in \cref{ex:1loop_massive}, if we set $\kappa=0$, the 1-loop graphs reduce to \cref{recursion_1loop} and one finds their sum in closed form, 
	\begin{align*}
		\mathcal G(x,t,0) &= \frac{x^4}{24}+ \frac{x^2 ~_2F_1 \left( 1, \epsilon-1; \epsilon; \frac{x^2}{2} \right)   }{4(\epsilon-1)}t  +\asyO{t^2} .
	\end{align*}
\end{example}

The power series  in the right hand side of \cref{gamma_recursion3} is $\partial_x^2\mathcal G$. Since $\loopnumber>0$, we can use  $[t^{\loopnumber-1}]f(t) = [t^\loopnumber](t\cdot f(t))$, and obtain
\begin{align}\label{gamma_recursion4}
	(2\loopnumber \epsilon +2n+2s-4) \frac{\Gamma^{(\loopnumber)}_{2n,s} }{ (2n)! } 
	&= \left[ x^{2n} \cdot t^{\loopnumber }\cdot \kappa^s \right] \left( \frac{t}{ \left(  1 -\partial_x^2 \mathcal G(x,t,\kappa) \right)}  -t\right)   .
\end{align}
The resulting \cref{gamma_recursion4} is even valid for $\loopnumber=0$ if we choose boundary conditions such that the left hand side vanishes at $\loopnumber=0$. This requires $( 2n +2s -4) \Gamma^{(0)}_{2n,s} =0$, hence  exactly three treelevel amplitudes may be chosen non-zero. As above, we set $\Gamma^{(0)}_{2,1}=1$ and $\Gamma^{(0)}_{4,0}=1$. Moreover, a tree level vacuum term $\Gamma^{(0)}_{0,2}\neq 0$ is allowed, but we will not use it.

Finally, we sum \cref{gamma_recursion4} over $\loopnumber$, $s$, and $n$, and identify the terms on the left hand side as first-order differential operators acting on $\mathcal G$ according to 
\begin{align}\label{factors_derivatives}
	2\loopnumber \epsilon t^\loopnumber &= 2\epsilon ~t \partial_t t^\loopnumber, \qquad \qquad 2n x^{2n} = x \partial_x x^{2n}, \qquad \qquad 2s \kappa^s = 2 \kappa  \partial_\kappa \kappa^s.
\end{align}
With this identification, \cref{gamma_recursion4} becomes a partial differential equation:
\begin{theorem}\label{thm:G_PDE} \cite{borinsky_tropicalized_2025} 
	The generating function of \cref{def:G_generating_function} satisfies the \emph{tropical loop equation}, which is the partial differential equation (PDE)
	\begin{align*}
		\big( 2 \epsilon ~t\partial_t + x \partial_x + 2 \kappa \partial_\kappa  - 4 \big) \mathcal G(x,t, \kappa ) &=  t \cdot \left( \frac{1 }{1-\partial_x^2 \mathcal G(x,t,\kappa)}-1\right) .
	\end{align*}
\end{theorem}

\begin{example}\label{ex:G_firstorder}
 	The  rational function on the right hand side of \cref{thm:G_PDE} vanishes at tree level, i.e. order $[t^0]$. The tropical loop equation simply becomes
	\begin{align*}
		x\partial_x \mathcal G +2 \kappa \partial_\kappa - 4 \mathcal G&= 0  \qquad \Rightarrow \quad  \mathcal G(x,t)  =x^4 \cdot f \left( \frac{\kappa}{x^2} \right) +\asyO{t}.
	\end{align*}
	The function $f$ is arbitrary. To reproduce our physical boundary condition $\mathcal G(x,t,\kappa)=\frac{\kappa x^2}{2!} + \frac{x^4}{4!}+\mathcal G^{(1)}(x, \kappa) + \asyO{t^2}$, we impose $f(u)= \frac{1}{4!}+\frac 12 u$.
	 Using this as an input,  the right hand side of \cref{thm:G_PDE} at 1-loop order involves $ 1-\partial_x^2 \mathcal G=  1-\kappa-\frac{x^2}{2}  $. Hence, the 1-loop order $\mathcal G^{(1)}$ satisfies the differential equation 
	\begin{align*}
	\big(	2 \epsilon \mathcal  + x \partial_x \mathcal  +2 \kappa \partial_\kappa   - 4 \big)\mathcal G^{(1)}(x,\kappa)  &= \frac{2\kappa+x^2 }{ 2 - 2 \kappa - x^2}.
	\end{align*}
	The general solution to this differential equation has two terms, one formal power series in $x$ and one  term proportional to $x^{4-2\epsilon}$. Leaving out the latter, one obtains an expression whose power series expansion reproduces \cref{ex:1loop_massive,ex:G_generating_firstorder},
	\begin{align*}
	\mathcal G(x,t,\kappa)=\frac{\kappa x^2}{2!} +\frac{x^4}{4!} - \frac{  \left( x^2+2\kappa \right)   ~_2F_1 \left( 1, \epsilon-1 ; \epsilon  ; \kappa+ \frac{ x^2}2\right)    }{4(\epsilon-1)  }t + \asyO{t^2}.
	\end{align*}	
\end{example}

\begin{remark}
	The vacuum amplitudes $\Gamma^{(\loopnumber)}_{0,s}$ vanish in the derivative $\partial_x^2\mathcal G$, therefore they do not enter the right hand side of \cref{thm:G_PDE}. In order to solve \cref{thm:G_PDE} recursively in terms of power series, it is useful to exclude the vacuum amplitudes $\Gamma_0$ from the recurrence by setting
	$\mathcal H(x,t,\kappa) \defas \partial_x^2 \mathcal G(x,t,\kappa)$. This function satisfies 
	\begin{align*}
		\left( 2 \epsilon ~t\partial_t + x \partial_x + 2 \kappa \partial_\kappa  -2 \right) \mathcal H(x,t, \kappa ) &= \partial_x^2 \frac{t}{1-\mathcal H(x,t,\kappa)}.
	\end{align*}
\end{remark}

\begin{remark}
	By expanding $\mathcal G = \sum_n \Gamma_n \frac{x^n}{n!}$, one can turn \cref{thm:G_PDE} into an infinite system of equations, the Dyson-Schwinger equations of tropical field theory, starting from 
	\begin{align*}
			\left( 2 \epsilon t \partial_t + 2 \kappa \partial_\kappa -4 \right) \Gamma_0  &= \frac{t \Gamma_2 }{1-\Gamma_2 }, \\
		(2 \epsilon t \partial_t +2 \kappa \partial_\kappa -2)\Gamma_2 &= \frac{t \Gamma_4}{(1-\Gamma_2 )^2} ,\nonumber \\
		\left( 	2 \epsilon  t \partial_t +2 \kappa \partial_\kappa  \right)  \Gamma_4  &=  \frac{  6t \Gamma_4 ^2 }{ (1-\Gamma_2 )^3}  + \frac{t \Gamma_6 }{\left(1- \Gamma_2  \right) ^2},   \\
		(2 \epsilon t \partial_t +2\kappa \partial_\kappa +2)\Gamma_6  &= \frac{90 t \Gamma_4^3}{(1-\Gamma_2 )^4} + \frac{30t \Gamma_4  \Gamma_6 }{(1-\Gamma_2)^3} + \frac{t \Gamma_8 }{ (1-\Gamma_2)^2}, \quad \ldots.  
	\end{align*}
\end{remark}

\subsection[O(N)-Symmetry]{$O(N)$ symmetry}\label{sec:PDE_ON}

In the $O(N)$ symmetric $\phi^4$ theory, the field variable is promoted to a $N$-component vector, and the interaction term $(\vec \phi^2)^2$ implies that every 4-valent vertex can graphically be viewed as a sum over three ways to match its for lets into pairs, see \cite{balduf_primitive_2024} for details. If all vertices of a graph $G$ are decomposed this way, the resulting object is a collection of open lines connecting the external legs of $G$ into pairs, and zero or more cycles. Each of the cycles represents a trace of vector indices and evaluates to a factor $N$. Summing over all $3^{\abs{V_\Graph}}$ decompositions of the vertices of $G$, one obtains a polynomial in $N$, the $O(N)$ symmetry factor $T_G(N)$. This polynomial constitutes an analytic continuation  to all $N\in \mathbb R$, in particular, the choice $N=-2$ has recently been of interest as the \emph{Martin invariant} \cite{PanzerYeats:c2Martin}.

\begin{example}\label{ex:gamma_1loop_N}	
	A cycle  graph $r_n$ consisting of 4-valent vertices (\cref{ex:tropical_cycle,ex:tropical_cycle_combinatorial}) allows for $3^n$ decompositions of the vertices. Exactly one decomposition produces a cycle, all other ones do not. Thereby
	\begin{align*} 
		T_{r_n}( N) &= \frac{N-1+3^n}{3^n}.
	\end{align*}
	In the $O(N)$ symmetric theory, 1-loop tropical Green functions  \cref{recursion_1loop} becomes
	\begin{align*} 
		\Gamma^{(1)}_{2n} &= \frac{(2n)! \left( N-1+3^n \right)  }{2\cdot 6^n \left( \epsilon+n-2 \right)  }.
	\end{align*}
	\begin{align*}
		\text{For example }\quad \Gamma^{(1)}_2&=\frac{1}{6}\frac{N+2}{\epsilon-1}, \qquad
		\Gamma^{(1)}_4=\frac{1}{3}\frac{N+8}{\epsilon}, \qquad
		\Gamma^{(1)}_6=\frac 5 3\frac{N+26}{\epsilon+1}, \qquad  \ldots 
	\end{align*}
\end{example}

The $O(N)$ symmetry factors $T_G(N)$ have a combinatorial behaviour that is strikingly similar to that of the tropical integral discussed in \cref{sec:recurrence} in the following sense: For an individual graph $G$ and a subgraph $\gamma \subset G$, the polynomial $T_G(N)$  depends non-trivially on how $\gamma$ is oriented inside $G$.   However, if one sums over all possible ways of inserting, the symmetry factor factorizes in terms of $T_\gamma(N)$.

\begin{lemma}[{e.g. \cite[Lem.~2]{balduf_primitive_2024}}] \label{lem:factorization_T}
		 Let $G_1$ be a connected graph with  $(2p)$ external edges, and let $v_2$ be a $(2p)$-valent vertex in a graph $G_2$ (which may or may not have external edges). Let $G_1 \circ_{v_2} G_2$  be the sum of graphs which arise from inserting $G_1$ in place of $v_2$ in all possible orientations. Each term in this sum has the same number of external edges as $G_2$ has, and the $O(N)$ symmetry factors factorize:  
		\begin{align*}
			T_{G_1 \circ_{v_2} G_2} (N) &= T_{G_1}(N) \cdot T_{G_2} (N).
		\end{align*}
\end{lemma}

\begin{proposition}\label{thm:gamma_recurrence_N}
	The $N$-dependent tropical $\loopnumber$-loop $2n$-point 1PI amplitude , for $\loopnumber>1$ and $n \geq 0$, is given by
	\begin{align*}
		\Gamma^{(\loopnumber)}_{2n} &= \frac{(2n)!}{2(\loopnumber \epsilon -2+n)}\left[ t^{\loopnumber-1} \right]\sum_{k=1}^{\loopnumber-1+n} \frac{k!}{(n+k)!}   \Bigg( B_{n+k,k} \left( \Gamma_2,~\Gamma_4,  ~\frac{\Gamma_6}{4}, ~\frac{\Gamma_8}{30}, ~\ldots , ~\frac{(n+1)! \Gamma_{2n+2}}{(2n)!}  \right)  \\
		&\qquad  + (N-1) B_{n+k,k} \left( \Gamma_2,~\frac{\Gamma_4}3,  ~\frac{\Gamma_6}{20}, ~\frac{\Gamma_8} {210}, ~\ldots , ~\frac{(n+1)! \Gamma_{2n+2}}{(2n+1)!}  \right) \Bigg) .
	\end{align*}
\end{proposition} 
\begin{proof}
	Firstly, in the $O(N)$ symmetric theory, a $v$-valent vertex (where $v$ is even) is to be decomposed into the $(v-1)!!$ matchings of its legs. Fix any two out of the $v$ adjacent edges. Precisely $(v-3)!!$ of the matchings will connect the two fixed edges. Hence, of all matchings, the proportion that connects the two edges is $\frac{(v-3)!!}{(v-1)!!}= \frac 1 {v-1}$. 
	
	In order to prove the claim, one repeats the steps that lead to \cref{thm:gamma_recurrence}, starting from the 1-loop amplitudes given in \cref{ex:gamma_1loop_N}. Then $\Gamma^{(\loopnumber)}_{2n}$ is given by a sum over all cycles of amplitudes of lower loop order. 
	The crucial step is to understand the behaviour of $O(N)$ symmetry factors upon closing a cycle. Consider a cycle with $2n$ external legs, built from $k$ 1PI subgraphs $\gamma_j$. According to \cref{lem:factorization_T}, in the sum over all graphs the individual polynomials $T_{\gamma_j}(N)$ factorize, hence it is sufficient to think of the $\gamma_j$ as $v_j$-valent vertices, disregarding their concrete structure. We need to determine $T(N)$ for a cycle graph consisting of $k$ vertices with valences $\langle v_1, v_2, \ldots, v_k \rangle$.
	
	Clearly, when all vertices are decomposed into matchings, a cycle graph can result in either zero or one, but not more, cycles in the matchings, analogously to \cref{ex:gamma_1loop_N}. The cycle matching can only arise when the decompositions of every vertex match the two edges adjacent to the graph. The symmetry factor $T(N)$ is normalized to unity, hence we require the proportion, not the absolute number, of suitable matchings.  As shown in the beginning, the proportion at any one vertex is $\frac 1 {v-1}$. Consequently, the symmetry factor of the whole cycle is 
	\begin{align*}
		1+ (N-1)\prod_{j=1}^k \frac{1}{v_j-1}.
	\end{align*}

	We conclude that in the recurrence step that leads to \cref{thm:gamma_recurrence}, two terms arise instead of one. The first term is identical to the one in \cref{thm:gamma_recurrence}. The second term carries an overall factor $(N-1)$, and a factor $\frac 1 {v_j-1}$ for each $v_j$-valent subgraph $\gamma_j$. Concretely, these terms are 
	\begin{align*} 
		&(2n)! \frac{k!}{2(n+k)!}B_{n+k,k} \left( \Gamma_2,~\Gamma_4,  ~\frac{\Gamma_6}{4}, ~\frac{\Gamma_8}{30}, ~\ldots , ~\frac{(n+1)! \Gamma_{2n+2}}{(2n)!}\right)  \nonumber  \\
		&\qquad + 	\left( N-1 \right) \cdot (2n)! \frac{k!}{2(n+k)!}B_{n+k,k} \left( \Gamma_2, ~\frac{\Gamma_4}3,  ~\frac{\Gamma_6}{20}, ~\frac{\Gamma_8}{210}, ~\ldots , ~\frac{(n+1)! \Gamma_{2n+2}}{(2n+1)!}\right) .
	\end{align*}
	In this formula, the individual $\Gamma_j$ are non-trivial functions of $N$, and in particular on the level of amplitudes, the $N$-dependence is not merely a factor $T(N)$ multiplying the ordinary amplitudes. To obtain the claimed formula, sum over $k$ and include the overall factor $\frac 1 {\omega_\Graph}$  analogously to \cref{thm:gamma_recurrence}.
\end{proof}

\begin{example}\label{ex:gamma_2loop_2point_N}
	Consider the 2-point 2-loop amplitude. \Cref{thm:gamma_recurrence_N} contains two summands, $k=1$ and $k=2$, where $B_{2,1}(\Gamma_2,\Gamma_4, \ldots)=\Gamma_4$ and $B_{3,2}(\Gamma_2,\Gamma_4,\ldots)= 3\Gamma_2 \Gamma_4$. Since $\Gamma_2$ has no tree level part, taking the order $[t^1]$ of $\Gamma_2\Gamma_4$ means $ \Gamma^{(1)}_2 \cdot \Gamma^{(0)}_4$, where the tree level vertex is $\Gamma^{(0)}_4=1$. Hence,  \cref{thm:gamma_recurrence_N} becomes
	\begin{align*}
		\Gamma^{(2)}_{2} &= \frac{2!}{2(2\epsilon-1)} \Bigg( \underbrace{\frac{1!}{2!}\Big( \Gamma_4^{(1)} + (N-1) \frac{\Gamma^{(1)}_4}{3} \Big)}_{k=1} + \underbrace{\frac{2!}{3!} \left( 3 \Gamma^{(1)}_2 + (N-1) \Gamma^{(1)}_2 \right)  }_{k=2} \Bigg)  \\
		&= \frac{1}{(2\epsilon-1)} \Bigg(  \frac{(N+8)(N+2)}{18\epsilon }  +  \frac{(N+2)^2}{18(\epsilon-1)}   \Bigg)  = \frac{(N+2)(-N-8 + 2  \epsilon (N +5) )}{18 \epsilon (\epsilon-1)(2\epsilon-1)}.
	\end{align*}
	This result is consistent with a sum of the two tropical Feynman integrals of  \cref{ex:tropical_twopoint_twoloop}, taking into account their $O(N)$ symmetry factors,  
	\begin{align*}
		\frac{\tropicalintegral[S] T_S(N)}{\abs{\Aut(S)}} &= \frac{N+2}{3 \epsilon(2\epsilon-1)}, \qquad \quad 
		\frac{ \tropicalintegral[P]T_P(N) }{\abs{\Aut(P)}} = \frac{(N+2)^2}{18\epsilon(\epsilon-1)}.
	\end{align*} 
	As anticipated, $T_G(N)$  is a multiplicative factor for any fixed graph $G$, the $N$-dependence of the full amplitude is non-trivially mixed with its $\epsilon$-dependence.
\end{example}

\begin{theorem}\label{thm:G_PDE_N}
	The generating function of bare 1PI amplitudes of $O(N)$-symmetric tropical $\phi^4$ theory satisfies the differential equation
	\begin{align*}
		\Big( 2 \epsilon ~t\partial_t + x \partial_x + 2 \kappa \partial_\kappa  - 4 \Big) \mathcal G(x,t, \kappa ) &=  t \cdot \left( \frac{1 }{1-\partial_x^2 \mathcal G(x,t,\kappa)} + \frac{ (N-1)x}{x-\partial_x \mathcal G(x,t,\kappa) } -N \right) .
	\end{align*}
\end{theorem}
\begin{proof}
	One starts from  \cref{thm:gamma_recurrence_N} and proceeds analogously to the derivation of \cref{thm:G_PDE} above. At each step, one obtains, in additional to the one from the $N=1$ case, a second term which is proportional to $(N-1)$ and contains slightly different combinatorial factors.  Concretely in   \cref{Bell_sum_shifted}, the arguments of that second term are $x_j=\frac{j! \Gamma_{2j}}{(2j-1)!}$,  and \cref{PDE_derivation_rescaling} takes the form
	\begin{align}\label{PDE_derivation_rescaling_N}
		&\frac{1}{n!}\sum_{r=1}^{n} \frac{r!}{\left( 1-\Gamma_2 \right) ^{r+1}}  B_{n,r} \left(  \frac{\Gamma_4}6,  ~\frac{\Gamma_6}{60}, ~\frac{\Gamma_8}{840}, ~\ldots , ~\frac{n! \Gamma_{2n+2}}{(2n+1)!}  \right) \nonumber \\
		&= -\left[ x^n \right]  \left(  \Gamma_2-1 + \sum_{j=1}^\infty x^j \frac{\Gamma_{2j+2}}{(2j+1)!} \right) ^{-1}= \left[ x^{2n} \right]  \left(  1 -\sum_{j=0}^\infty x^{2j} \frac{\Gamma_{2j+2}}{(2j+1)!} \right) ^{-1}.
	\end{align}
	Taking both terms together, \cref{gamma_recursion3} is replaced by  
	\begin{align*} 
		\frac{\Gamma^{(\loopnumber)}_{2n,s} }{ (2n)! } 
		&=\frac{\left[ x^{2n} \cdot t^{\loopnumber-1}\cdot \kappa^s \right]}{2(\loopnumber \epsilon +n +s-2) }\left(  \frac 1 { \left(  1 -\sum_{j=0}^\infty x^{2j} \frac{\Gamma_{2j+2}}{(2j)!} \right)} +  \frac {N-1} { \left(  1 -\sum_{j=0}^\infty x^{2j} \frac{\Gamma_{2j+2}}{(2j+1)!} \right)}-N  \right)    .
	\end{align*}
	The last sum can be expressed in terms of the generating function \cref{def:G_generating_function} according to
	\begin{align*}
		\sum_{j=0}^\infty x^{2j} \frac{\Gamma_{2j+2}}{(2j+1)!}= \frac 1 x \partial_x \sum_{j=0}^\infty x^{2j} \frac{\Gamma_{2j}}{(2j)!}.
	\end{align*} 
\end{proof}

Observe that $N$, similarly to $\epsilon$, enters the tropical loop equation merely as a parameter, but there are no derivatives with respect to $N$.

\begin{remark}\label{ren:propagator_vanishing}
	The amplitude $\Gamma_2$ with two external legs amounts to $n=1$. At $N=-2$, we have $T(r_1,-2)=0$ by \cref{ex:gamma_1loop_N}, hence, the tadpole graph vanishes at $N=-2$. At higher loop order, a cycle with two external legs consists of exactly one factor $\Gamma_4$, and arbitrarily many $\Gamma_2$. Inserting $N=-2$ in \cref{PDE_derivation_rescaling_N} and taking the coefficient $[x^2]$, we see that $\frac{\Gamma_4}2+ (-3)\frac{\Gamma_4}{6}=0$. Thereby, at $N=-2$ the tropical two-point function vanishes at all orders. By \cref{thm:G_PDE_N}, the vacuum function   vanishes as well when $N=-2$.   These findings are a structural properties of the $O(N)$ symmetric theory, they  have long been known in the non-tropical theory \cite{balian_critical_1973}, and follow at once from an exact enumeration of all graphs \cite{balduf_primitive_2024}.
\end{remark}

\subsection{Functions for other classes of graphs} \label{sec:PDE_other}

The quantum effective potential $\mathcal G(x,t,\kappa)$ of \cref{def:G_generating_function} is the generating function of 1PI amplitudes at zero external momentum, where  $t^\loopnumber$ counts the loop number. It is well known that various other generating functions can be obtained from $\mathcal G$ through algebraic operations \cite{jona-lasinio_relativistic_1964,kleinert_higher_1982,oraifeartaigh_constraint_1986}, which are equivalent to the relations between various thermodynamic potentials in statistical physics. 
In the present section, we state the transformations, and the partial differential equations that the transformed generating functions satisfy. We will not spell out the proofs, they amount to elementary calculations upon inserting the (inverse) transformations into the tropical loop equation \cref{thm:G_PDE_N}.  

\medskip

Firstly, we rewrite $\mathcal G$ as a function of the number of vertices  instead of number of loops, this is realized  by a change of variables, where $g$ counts vertices and $\phi$ counts legs,
\begin{align}\label{def:Gbar}
	\bar{\mathcal G}(\phi,g,\kappa) &= \frac 1 g \mathcal G \left( \phi \sqrt{g},g,\kappa \right) , \qquad \Leftrightarrow \quad  \mathcal G(x,t,\kappa) = t \bar{\mathcal G}\left( \frac{x}{\sqrt t},t ,\kappa\right) =g\bar{\mathcal G}(\phi,g,\kappa). 
\end{align}

\begin{lemma}\label{lem:G_PDE_g}
	The generating function $\bar {\mathcal G}(\phi,g,\kappa)$ of \cref{def:Gbar} satisfies
	\begin{align*}
	 \Big( 2 \epsilon g \partial_g      +  (1-\epsilon) \phi\partial_\phi    +2\kappa \partial_\kappa  -2(2-\epsilon) \Big) \bar{\mathcal G}   
	&= \frac{1 }{1-\partial_\phi^2 \bar{\mathcal G} } + \frac{ (N-1)\phi }{\phi  -\partial_\phi  \bar{\mathcal G}  } -N.
	\end{align*}	
\end{lemma}

\begin{example} \label{ex:G_transformed}
	Let $\kappa=0$, then the 1PI generating function starts with
	\begin{align*}
		&\bar{\mathcal G}  =   \left( \frac{N(N+2)}{24(\epsilon -1)^2} + \frac{(N+2)\phi^2}{12( \epsilon-1)} + \frac{\phi^4}{24} \right) g \\
		& + \scalemath{.82}{\left( \frac{N(N+2)(6 (N+4)\epsilon^2-(38+7N) \epsilon + 2(N+8))}{72 \epsilon ( \epsilon-1)^2 ( 2\epsilon-1)( 3\epsilon-2)} + \frac{(N+2)( 2(N+5)\epsilon-(N+8) )\phi^2}{36\epsilon (\epsilon-1)(2\epsilon-1)} + \frac{(N+8)\phi^4}{ 72\epsilon} \right)g^2 + \ldots } .
	\end{align*}
	Notice that in our definitions, $\bar {\mathcal G}$ does not contain a treelevel propagator term $\frac{\phi^2}{2}$, and the treelevel vertex appears at order $[g^1]$ as it should. 
	The term $[g^2]$ consists of $\Gamma^{(1)}_4$ from \cref{ex:gamma_1loop_N}, $\Gamma^{(2)}_2$ from \cref{ex:gamma_2loop_2point_N}, and the 3-loop vacuum amplitude $\Gamma^{(3)}_0$. No higher order terms in $\phi$ have been left out; the first six-point amplitude appears at order $[g^3]$.
\end{example}

The generating function of connected amplitudes $\bar{\mathcal W}$ is related to the effective action $\bar{\mathcal G}$ by a Legendre transform. In our conventions, we need to manually include  the propagator term $-\frac{\phi^2}{2!}$ which was left out in $\bar{\mathcal G}$. This term is unrelated to whether $\kappa \neq 0$ since the parameter $\kappa$ counts the number of  mass vertices, not  of ordinary propagators. Define 
\begin{align}\label{def:W_g}
 	j(\phi,g,\kappa )&\defas \partial_\phi\Big(  \bar {\mathcal G}(\phi,g,\kappa) - \frac{\phi^2}{2} \Big) \nonumber  \\
	\bar{\mathcal W}(j,g,\kappa) &\defas \bar{\mathcal G} \big( \phi(j,g, 	\kappa), g,\kappa \big)  - \frac{\phi(j,g,\kappa)^2}{2} - j \cdot \phi(j,g,\kappa) .
\end{align}

\begin{lemma}\label{lem:W_PDE_g}
	The function $\bar {\mathcal W}$ of \cref{def:W_g} satisfies  
	\begin{align*}
		\Big(  2\epsilon g \partial_g   + \left( (3-\epsilon )  j -\frac{N-1}{j} \right) \partial_j  +2 \kappa \partial_\kappa   + 2(\epsilon-2) \Big)   \bar{\mathcal W}   &=\left( \partial_j \bar{\mathcal W} \right) ^2 + \partial_j^2 \bar{\mathcal W}   -N.
	\end{align*} 
\end{lemma}

\begin{example}  
	As for  $\bar{\mathcal G}$ in \cref{ex:G_transformed}, the coefficients of  $\bar{\mathcal W}$ at fixed order of $g$ are  polynomials, and not infinite power series, in $\phi$ because a connected graph with fixed number of vertices have a bounded number of external legs. At $\kappa=0$ and $N=1$,  
	\begin{align*}
		&\bar{\mathcal W}  =  \frac{j^2}{2} + \left( \frac{1}{8( \epsilon -1)^2} + \frac{j^2}{4( \epsilon-1)} + \frac{j^4}{24} \right) g \\
		&\qquad + \left( \frac{6-15 \epsilon +10 \epsilon^2}{8 \epsilon ( \epsilon-1)^2 ( 2\epsilon-1)( 3\epsilon-2)} + \frac{(6-15 \epsilon+10\epsilon^2) j^2}{8\epsilon (\epsilon-1)^2(2\epsilon-1)} + \frac{(5\epsilon -3) j^4}{24 \epsilon(\epsilon-1)} + \frac{j^6}{72}\right)g^2 + \ldots .
	\end{align*}
\end{example}

\medskip

The path integral $\bar{\mathcal Z}$ is a generating function of \emph{all} graphs, 
\begin{align}\label{def:Z_generating_function}
	\bar{\mathcal Z}(j,g,\kappa) &:= \exp \left( \bar{\mathcal W}(j,g,\kappa ) \right) , \qquad  \bar{\mathcal W}(j,g,\kappa) = \ln\left( \bar {\mathcal Z}(j,g,\kappa ) \right)  .
\end{align}

\begin{example}
	The leading coefficient $[g^0]\bar{\mathcal Z}(j,g)$ is $e^{\frac{j^2}{2}}=1+\frac{j^2}{2} + \frac{j^4}{8} + \frac{j^6}{48}+\ldots$, it  enumerates the ways to pair up vertices by propagators. One has
	\begin{align*}
		&\bar{\mathcal Z} = e^{\frac{j^2}{2}} + \left( \frac{1}{8(\epsilon-1)^2} + \frac{(4\epsilon-3)j^2}{16(\epsilon-1)^2} + \frac{(8 \epsilon^2 + 8 \epsilon-13)j^4}{192(\epsilon-1)^2}+\asyO{j^6} \right) g\\
		& \scalemath{.9}{+ \left( \frac{160\epsilon^4 - 554 \epsilon^3 + 729 \epsilon^2 - 430 \epsilon + 96}{128 \epsilon(\epsilon-1)^4 (2 \epsilon-1) (3 \epsilon-2)} + \frac{(480 \epsilon^4 - 1656 \epsilon^3 + 2171 \epsilon^2 - 1282 \epsilon + 288)j^2}{256\epsilon (\epsilon-1)^4(3 \epsilon-2)} \right) g^2 + \asyO{g^3}.}
	\end{align*}
\end{example}

\begin{proposition}\label{thm:Z_PDE_g}
	The  path integral $\bar {\mathcal Z}(j,g,\kappa )$ (\cref{def:Z_generating_function}) of tropical $\phi^4$ theory satisfies 
	\begin{align*}
		2\epsilon g   \partial_g \bar{\mathcal Z}   + \left( (3-\epsilon )  j -\frac{N-1}{j} \right) \partial_j \bar{\mathcal Z}  +2\kappa \partial_\kappa \bar{\mathcal Z} + 2(\epsilon-2) \bar{\mathcal Z}  \ln  \bar{\mathcal Z}  &=  \partial_j^2 \bar{\mathcal Z}-N \bar{\mathcal Z}.
	\end{align*}
\end{proposition}

\medskip 
As discussed in \cref{sec:tropical_limit}, tropical field theory is a scaling limit towards zero spacetime dimension. However, tropical field theory is distinct from what is conventionally called \emph{zero-dimensional $\phi^4$ theory}. The latter  is a model where the \enquote{field} $\phi$ is a $N$-component real-valued vector that does not depend on spacetime at all. The action does not involve any spacial integral and is given by\footnote{Recall that all generating functions in this section use the variables $t,g,j,\ldots$ as \enquote{counting parameters}, such that their coefficients are positive. The physical series are alternating series of $g$ or $t$. Hence, the action \cref{zero_dim_action} has wrong signs, or equivalently, the \enquote{physical} value of $g$ is negative.}
\begin{align}\label{zero_dim_action}
	S &\defas \frac{  \phi^2}{2} -\kappa\frac{ \phi^2}2 - \frac{g (  \phi^2)^2}{4!} -\eta  \phi^2.
\end{align}
Here, $\eta$ represents an $O(N)$ invariant source term discussed in detail   in \cite{balduf_primitive_2024}. 
The zero-dimensional path integral is an ordinary $N$-dimensional integral over the $N$ field components, 
\begin{align}\label{zero_dim_path_integral}
	\bar{\mathcal Z} &= \int \frac{\d^N \phi}{(2\pi)^{\frac N2}}e^{-S}=\int \frac{\d^N \phi}{(2\pi)^{\frac N2}} e^{  -\frac{\phi^2}{2} + \frac{g \phi^4}{4!} +\kappa \frac{\phi^2}{2} + \eta \phi^2  }.
\end{align}
In the zero-dimensional theory, 
each Feynman integral has the value unity, such that the generating functions $\mathcal Z,\mathcal W, \ldots$ simply enumerate the corresponding graphs, weighted by symmetry factors \cite{bessis_quantum_1980}. Their exact large-order asymptotics and analytical properties have been analyzed in  \cite{borinsky_renormalized_2017,balduf_primitive_2024,benedetti_smalln_2024,aniceto_primer_2019}.

\begin{proposition}\label{thm:zero_dim_PDE_series}
	 \Cref{zero_dim_path_integral}  has the asymptotic series expansion
	\begin{align*} 
		\bar{\mathcal Z} &=  \sum_{n=0}^{\infty } \sum_{k=0}^{\infty}\sum_{s=0}^\infty \frac{  \Gamma \left(  2n+k+s+\frac N2 \right)   } 
		{6^n    2^k 2^s   \Gamma \left( k+\frac N2 \right)   } \frac{j^{2k} \kappa^s g^n}{k! s! n!} ,
	\end{align*}
	and it satisfies the partial differential equation
		\begin{align*}
		4 g   \partial_g \mathcal Z   + \left(   j -\frac{N-1}{j} \right) \partial_j \bar{\mathcal Z}  +2\kappa \partial_\kappa \bar{\mathcal Z}    &=  \partial_j^2 \bar{\mathcal Z}-N \bar{\mathcal Z},
	\end{align*}
	which coincides with the case $\epsilon=2$ of the tropical PDE in \cref{thm:Z_PDE_g}.
\end{proposition}
 
\begin{proof}
	The computation to obtain the series expansion  is standard, see e.g. \cite{balduf_primitive_2024,benedetti_smalln_2024}. The crucial step is to use a  Hubbard-Stratonovich transformation \cite{stratonovich_method_1958,hubbard_calculation_1959,byczuk_generalized_2023},
	\begin{align*}
		e^{g \frac{\phi^4}{4!}} &=  \int \frac{\d \sigma}{\sqrt{2\pi }} e^{     - \frac{\sigma^2}{2} + \frac{ \sqrt { g} \phi^2}{\sqrt{12}} \sigma   }.
	\end{align*}	
	With this, $\mathcal Z$ becomes a Gaussian integral over $\phi$, which can be solved in closed form,
	\begin{align*} 
		\bar{\mathcal Z} &=    \int \frac{\d \sigma}{\sqrt{2\pi }} \int \frac{\d^N \phi}{(2\pi)^{\frac N2}} e^{  -\frac{\phi^2}{2}\left(1-\kappa - 2 \eta \right)    }   e^{     - \frac{\sigma^2}{2} + \frac{ \sqrt { g} \phi^2}{\sqrt{12}} \sigma   }= \int \frac{\d \sigma}{\sqrt{2\pi }} \frac{ e^{    - \frac{\sigma^2}{2} } }{\left( 1-\kappa-2\eta-\frac{\sqrt g \sigma}{\sqrt 3} \right) ^{\frac N2}}.
	\end{align*}
	Exapnding the integral and integrating term by term yields the claimed series expansion. 	Taking the derivative of individual terms in the series shows that it satisfies the PDE. 	
\end{proof}

It has long been known that the path integral of zero-dimensional QFT satisfies   differential equations; a second-order ODE for $\bar{\mathcal Z}(g,j=0, \kappa=0,N)$ was analyzed in \cite[Sec~2]{aniceto_primer_2019} and generalized to $N\neq 1$ in \cite[Prop 1.7]{benedetti_smalln_2024}, and  \cite[Exercise~2.I.4]{cvitanovic_field_1983} considers a PDE that contains second derivatives in the coupling. However, we are not aware of a reference for the PDE of \cref{thm:zero_dim_PDE_series} in the existing literature.

\begin{remark}
	 Recall that the conventional Dyson-Schwinger equation expresses invariance of the path integral under a shift of the field variable $\phi \mapsto \phi +c$, schematically 
	\begin{align*}
		0&=\frac{\delta S}{\delta \phi} \Big|_{\phi \rightarrow \frac{\delta}{\delta J}} \mathcal Z(j)= \int \frac{\d^N \phi}{(2\pi)^{\frac N2}}    \frac{\delta S}{\delta \phi }   e^{  -\frac{\phi^2}{2} + \frac{g \phi^4}{4!} +\kappa \frac{\phi^2}{2} + j \phi  } .
	\end{align*} 
	More general versions of Dyson-Schwinger equations can be obtained by considering other types of field transformation, see \cite[(3.32)]{cvitanovic_field_1983}. If one acts with the PDE of \cref{thm:zero_dim_PDE_series} on the zero-dimensional path integral, one finds
	\begin{align*}
		0 &= \int \frac{\d^N \phi}{(2\pi)^{\frac N2}} \left( 4 \frac{\phi^4}{4!} + j \phi + \frac{N-1}{j}\phi +2\kappa \frac{\phi^2}{2} - \phi^2 +N \right)   e^{  -\frac{\phi^2}{2} + \frac{g \phi^4}{4!} +\kappa \frac{\phi^2}{2} + j \phi  }\\
		&= \int \frac{\d^N \phi}{(2\pi)^{\frac N2}} \left( \phi  \frac{ \delta S} {\delta \phi }  +N \right)   e^{  -\frac{\phi^2}{2} + \frac{g \phi^4}{4!} +\kappa \frac{\phi^2}{2} + j \phi  }.
	\end{align*}
	This shows that \cref{thm:zero_dim_PDE_series} is a Dyson-Schwinger equation that corresponds to a constant rescaling $\phi \mapsto c \cdot \phi$ of the field variable. 
\end{remark}

\section{Renormalization}\label{sec:renormalization}

The fact \emph{that} tropical field theory is renormalizable follows at once from the fact that the tropical approximations of Feynman integrals (\cref{sec:combinatorial_formulas}) have been used as bounds to   prove the renormalizability and growth rates of the non-tropical theory in \cite{hepp_proof_1966,decalan_local_1981}. Instead of reproducing these proofs, we will discuss in detail the special role that is played by 2-point functions, mass terms, and the IR regulator, in our particular setup.

\subsection{Divergences in the long-range theory} \label{sec:longrange_divergences}

Since readers with a high energy physics background might not be familiar with long-range field theory, we briefly review the structures of divergences of this theory, see e.g. \cite{fisher_critical_1972}.  
The Weinberg power counting theorem \cite{weinberg_highenergy_1960}  asserts that a Feynman integral is free of UV divergences if it both has no subdivergences, and its superficial degree of convergence $\bar \omega_\Graph$ (\cref{def:sdd}) is strictly positive. The theorem is valid even for propagators that decay with non-integer powers at high energy,  see also \cite[Sec.~20.2]{ticciati_quantum_1999}. If Feynman integrals are regularized in a way that is independent of masses and momenta, such as in dimensional regularization, then taking derivatives with respect to masses or external momenta commutes with the regularized integration.  In the short-range theory, computations like $\partial_{m^2} (p^2+m^2)^{-1} = -(p^2+m^2)^{-2}$  (\cite[eq. (11.33)]{kleinert_critical_2001}) show that the action of these derivatives on a propagator effectively increases the degree of convergence by two units. After finitely many derivatives, the integral is convergent, hence its singular parts vanish. On the other hand, a function vanishes after finitely many derivations if and only if it is a polynomial. This   argument   demonstrates that the singular parts of Feynman integrals are polynomials of masses and momenta, see e.g. \cite{caswell_simple_1982,kleinert_critical_2001}. In    $\phi^4$ theory, a propagator-type graph has $\abs{E_\Graph}=2 \loopnumber-1$ and a vertex-type graph has $\abs{E_\Graph}=2\loopnumber$. In dimensional regularization with $D=4-2\epsilon$, propagator graphs have $\omega_\Graph = -1+\loopnumber \epsilon$ and vertex graphs have $\omega_\Graph=\loopnumber \epsilon$. Singularities are pole terms in $\epsilon$, for a propagator-type graph their residue is    proportional to (i.e. a first-order polynomial in) $p^2$ and $m^2$, and for vertex-type graphs their residue is a pure number independent of $p^2$ or $m^2$. In particular, the residues can not contain  ratios like $p^2/m^2$.

When the dimension $D$ is kept at its integer physical value, assigning a non-integer power $\xi$ to propagators (\cref{longrange_propagator_positionspace}) amounts to analytic regularization, where divergences are poles of the form $\frac{1}{\xi-n}$ for integer $n$. The dimensionally regularized long-range theory effectively employs dimensional- and analytic regularization simultaneously \cite{honkonen_crossover_1989}, singular   terms are then of the form $\frac{1}{a(\xi-1)+b\epsilon}$, where $a,b$ are combinatorial coefficients.

Despite arising from such analytically regularized theory, the combinatorial formula (\cref{def:tropical_combinatorial}) for the tropicalized Feynman integral clearly shows that $\tropicalintegral[\Graph]$ can have pole terms in $\epsilon$ arising from vertex-type subgraphs. This is consistent with analytic regularization due to the peculiar choice of relation $D=\xi(4-2\epsilon)$ we imposed in \cref{dimension_xi}, which implies
\begin{align}\label{tropical_omega}
\bar \omega_\Graph &= \xi  \loopnumber\epsilon \quad \text{for vertex graphs},\qquad \text{and} \quad \bar \omega_\Graph=\xi (\loopnumber \epsilon - 1) \quad \text{for propagator graphs}.
\end{align}
In the tropical limit $\xi \rightarrow 0$, according to \cref{tropical_limit_gamma}, the first expression yields a pole term $\frac{1}{\loopnumber \epsilon}$, while propagator-type graphs give finite factors $\frac{1}{\loopnumber \epsilon-1}$. Phrased differently, we have set up a version of analytic regularization precisely in such a way that log-divergent graphs retain their poles in the dimensional regularization parameter and are divergent as $\epsilon \rightarrow 0$ even if $\xi \neq 1$. In fact, the reason for imposing \cref{dimension_xi} in the first place was just this: We wanted to guarantee that vertex-type graphs are marginal, i.e. log-divergent, for all values of $\xi$. 

It remains to understand the  status of propagator-type graphs. By power counting (\cref{tropical_omega}), their degree of convergence is negative for small enough $\epsilon$, therefore their original Feynman integral does not converge at $\epsilon=0$. However, in the combinatorial formula \cref{def:tropical_combinatorial}, they do not give rise to poles at $\epsilon=0$ in the tropical limit $\xi\rightarrow0$. The situation becomes more transparent if one analyzes possible parameter-dependence of the pole terms  for a massless long-range theory with $0<\xi<1$ and IR regulator $\mu$. The propagator is $(p^{2}+\mu^{2})^{-\xi}$, which has mass dimension $-2\xi$, therefore one could suspect pole terms to be proportional to $\mu^{2\xi}$ or $p^{2\xi}$. If that were the case, then for such a pole term, the derivative
\begin{align}\label{pole_derivative}
	\frac{\partial}{\partial  \mu^{2 \xi}}	\frac{\mu^{2\xi }} {\epsilon}&= \frac{\mu^{2-2\xi}}{ \xi} \frac{\partial}{\partial   \mu^{2  }}	\frac{\mu^{2\xi } }{\epsilon} = \frac 1 \epsilon 
\end{align}
should be non-zero. On the other hand, computing this derivative on a propagator yields
\begin{align}\label{longrange_regulator_derivative}
	\frac{\partial}{\partial    \mu^{2\xi} } \frac{1}{((k-p)^2+\mu^2)^\xi }&=  \frac{-\mu^{2-2\xi} }{((k-p)^2+\mu^2)^{\xi+1}} .
\end{align}
The   dimensionful factor $\mu^{2-2\xi}$ on the right hand side is independent of integration variables and can therefore be pulled out of the integral. Consequently, while the derivative reduces the mass dimension only by $2\xi < 2$, it increases the  superficial degree of convergence of the integral by two units, not just by $2\xi$.   Since the original integral had $\bar \omega=\xi \loopnumber \epsilon -\xi$ (\cref{tropical_omega}), the derived integral has a strictly positive degree of convergence $\bar \omega=\xi \loopnumber \epsilon + (1-\xi)$, and therefore it converges. In view of \cref{pole_derivative}, we conclude that prior to taking the derivative, there can not have been a pole proportional to $\mu^{2\xi}$. An analogous computation shows that the derivative $\partial_{p^{2\xi}}$, where $p$ is an external momentum, also increases the degree of convergence by 2, not just $2\xi$.  Hence,  the massless long-range  theory does not contain any divergences proportional to $p^{2\xi}$ or $\mu^{2\xi}$. For dimensional reasons, propagator-type graphs can not have pole terms which are independent of $p$ and $\mu$, either.
This establishes that as long as $\xi<1$, the 2-point function does not have any pole terms at $\epsilon=0$.

\begin{example}\label{ex:sunrise_longrange}
	Consider the massless, not IR-regulated,  long range theory. 
	The  2-loop multiedge (sunrise) with propagator powers $\xi$ and non-vanishing external momentum $p^2$  has the well-known Feynman integral
	\begin{align*}
		\feynmanintegral[G]&= p^{-2\xi (-1+2 \epsilon)} \frac{\Gamma(\xi(2 \epsilon-1))}{\Gamma(\xi)^3} \frac{\Gamma \left( \xi(1- \epsilon) \right) ^3}{\Gamma \left( 3\xi(1 -  \epsilon)\right)  }.
	\end{align*}
	For $0<\xi<1$, this expression is finite at $\epsilon=0$, it represents the analytically regularized integral. The limit $\xi\rightarrow 0$ poses no problem an gives
	\begin{align*}
		\lim_{\xi \rightarrow 0} \feynmanintegral[\Graph]=\frac{3}{(2\epsilon-1) (\epsilon-1)^2}= -3-12 \epsilon- 33 \epsilon^2 + \asyO{\epsilon^3}.
	\end{align*} 
	This outcome is consistent with our conclusion that a propagator-type graph with $\xi \neq 1$ does not give rise to superficial singularities at $\epsilon=0$.  The vertex-type subdivergence of this graph will be discussed below, compare \cref{ex:sunrise_renormalized}.

\end{example}

Despite containing no pole at $\epsilon=0$, the finite terms of a propagator integral are proportional to $p^{2\xi}$ or to $\mu^{2\xi}$, where $\mu$ is a mass scale of the theory, in our case the IR regulator. Let $\delta$ be the sum of all terms proportional to $\mu^{2\xi}$ of 1PI graphs, for $\xi=1$ this would be a quantum correction to the mass, hence one might expect that these terms can be absorbed into a redefinition of the constant $\mu$. To that end, we examine the effect of replacing  $\mu^2 \mapsto \mu^2(1-\delta)$. The IR-regulated propagator becomes, upon using the binomial theorem:
\begin{align}\label{regularized_propagator_massless_sum}
\frac{1}{(p^2 +\mu^2(1-\delta))^\xi} &= \frac{1}{\left( p^2+\mu^2 \right) ^\xi} \frac{1}{\left( 1-\frac{\mu^2 \delta}{p^2+\mu^2} \right) ^\xi}=\frac{1}{\left( p^2+\mu^2 \right) ^\xi} \sum_{k=0}^\infty \frac{\Gamma(\xi+k)}{\Gamma(\xi)k!} (\mu^2 \delta)^k \left( \frac{1}{p^2+\mu^2} \right) ^k.
\end{align}
For $\xi=1$, the combinatorial coefficient collapses to $\frac{\Gamma(\xi+k)}{\Gamma(\xi)k!}=1$ and the sum has a straightforward interpretation: The $k$\textsuperscript{th} term represents the insertion of $k$ mass corrections $\delta$, joined by propagators. Consequently, in that case the replacement $\mu^2 \rightarrow \mu^2(1-\delta)$ absorbs the insertion of arbitrary many 1PI mass corrections.  However, when $\xi<1$, the combinatorial factors $\frac{\Gamma(\xi+k)}{\Gamma(\xi)k!}$ spoil this interpretation, and the factors $\frac{1}{p^2+\mu^2}$ in the sum are not even the  propagators in the theory since they are missing the exponent $\xi$. Consequently, the mass correction terms in the long range theory can not be interpreted as quantum corrections to the IR regulator $\mu$: As mentioned in \cref{sec:longrange}, $\mu$ is not a physical mass term.

For comparison, we consider the propagator of a \emph{massive} long range theory with mass $m$, which has a different functional form. Redefining $m^{2\xi}\mapsto m^{2\xi} (1-\delta) $ gives:
\begin{align}\label{massive_propagator_sum}
\frac{1}{p^{2\xi}+m^{2\xi}(1-\delta)} = \frac{1}{p^{2 \xi} + m^{2\xi}} \frac{1}{1- \frac{\mu^{2 \xi} \delta}{p^{2\xi}+m^{2\xi}}} = \frac{1}{p^{2 \xi} + m^{2\xi}} \sum_{k=0}^\infty (m^{2\xi}\delta)^k \left( \frac{1}{p^{2\xi}+m^{2\xi}}\right)^k.
\end{align}
Now, regardless of the value of $\xi$, the terms on the right hand side represent a propagator with insertion of exactly $k$ mass correction vertices $\delta$, joined by propagators. Conversely: If the theory has quantum corrections proportional to $m^{2\xi}$, then the geometric sum of these corrections inevitably gives rise to a propagator of the form \cref{massive_propagator_sum}, and not \cref{regularized_propagator_massless_sum}. 
However, we \emph{demand} to be working in a massless theory. In view of the renormalization group discussion of \cref{sec:longrange}, this means that we tacitly introduce a mass term of the form $m^{2\xi}$ in \cref{massive_propagator_sum}, which is subsequently renormalized to zero. This is a kinematic (MOM) renormalization condition for $m$, regardless whether we work in the MS scheme for the resulting massless theory. 

\smallskip

Unlike the regularized massless theory (\cref{sec:longrange_divergences}), the massive long range theory \emph{does} have UV poles for a mass term, namely the mass derivative of a Feynman graph amounts to inserting a 2-valent vertex into each propagator in turn (compare \cref{ex:1loop_massive}),
\begin{align}\label{longrange_mass_derivative}
\frac{\partial}{\partial m^{2\xi}} \frac{1}{k^{2\xi}+ m^{2\xi}} &= \frac{-1}{\left( k^{2\xi}+m^{2\xi} \right) ^2}= \frac{1} { k^{2\xi}+m^{2\xi} }\cdot(-1)\cdot  \frac{1}{  k^{2\xi}+m^{2\xi }} .
\end{align}
This operation reduces the degree of convergence by $2\xi$, and not, like in \cref{longrange_regulator_derivative}, by two units. Therefore, Feynman graphs with a mass insertion are superficially divergent even if $\xi<1$. 
A calculation similar to \cref{longrange_regulator_derivative} confirms that the long range theory $\xi<1$ does not have divergent wave function renormalization, regardless of whether it is massive or not.

\smallskip

Enforcing $m^{2\xi}=0$ as a renormalization condition means that all tadpoles --- since they are proportional to $m^{2\xi}$ --- are renormalized to zero. However, in the IR-regularized theory, their bare integrals do not vanish by themselves. One has two choices: If counterterms and amplitudes are computed graph-by-graph, then tadpole graphs can be left out from the start. Conversely, if one renormalizes at the level of entire amplitudes, tadpoles need to be included into the bare amplitudes because they are needed as cographs of vertex subdivergences. 

\begin{example}\label{ex:sunrise_renormalized}
	In the massless theory without IR regulator, the sunrise does not have poles  when $\xi \rightarrow 0$ (\cref{ex:sunrise_longrange}). This graph has three isomorphic UV subdivergences, but the cographs of these subdivergences are tadpoles, which vanish in a massless theory.
	
	Conversely, the IR-regularized tropical integral $\tropicalintegral[S]= \frac 6 {\epsilon(2\epsilon-1)}$ (\cref{ex:tropical_twopoint_twoloop})  has a pole at $\epsilon=0$ that stems from the vertex-type subdivergences. Upon renormalization,
	\begin{center}
		\begin{tikzpicture}
			
			\node[vertex](v1) at (.5,0){};
			\node[vertex](v2) at (1.8,0){}; 
			
			\draw[edge, bend angle=50,bend left](v1) to (v2);
			\draw[edge ](v1) to (v2);
			\draw[edge, bend angle=50,bend right](v1) to (v2);
			\draw[edge] (v1) to +(-.3,0);
			\draw[edge] (v2) to +( .3,0);
			
			\node at (2.8,0){\large $-$};
			\node at (3.5,0){$3~\times $};
			\node at (4.2,0){$\renop \Big[$};
			
			\node[vertex](v1) at (4.8,0){};
			\node[vertex](v2) at (5.6,0){}; 
			\draw[edge, bend angle=40,bend left](v1) to (v2);
			\draw[edge, bend angle=40,bend right](v1) to (v2);
			\draw[edge] (v1) to +(-.2, .1);
			\draw[edge] (v1) to +(-.2, -.1);
			\draw[edge] (v2) to +(.2, .1);
			\draw[edge] (v2) to +(.2, -.1);
			
			\node at (6.3,0){$\Big] ~\times$};
			
			\node[vertex](v1) at (7,-.2){};
			
			\draw[edge ](v1) ..controls +(.6,.7) and +(-.6,.7)..  (v1);
			\draw[edge] (v1) to +(-.3,0);
			\draw[edge] (v1) to +( .3,0);
			
			\node at (7.5,-.1){$,$};
		\end{tikzpicture}
	\end{center}
	we obtain a tropical integral that is free of poles at $\epsilon=0$, 
		\begin{align*}
		\tropicalintegral_\ren[S]		&= \frac{6}{2\epsilon (2\epsilon-1)}-3\times \frac{2}{\epsilon}\times \frac{1}{\epsilon-1}= \frac{-6}{(\epsilon-1)(2\epsilon-1)}.
	\end{align*}
	In a short-range theory, the graph would still have a superficial divergence, but, as discussed in \cref{sec:longrange_divergences}, there are no superficial propagator-type divergences in the long-range or tropical theory.
	Note also that the resulting expression  $\tropicalintegral_\ren[S]$ is  distinct from the limit $	\lim_{\xi \rightarrow 0} \feynmanintegral[\Graph]=\frac{3}{(2\epsilon-1) (\epsilon-1)^2}$ of the truly massless sunrise  of \cref{ex:sunrise_longrange}. This is expected: Two distinct procedures to regularize IR singularities  give rise to identical UV poles (namely, no poles in either case), but the resulting finite terms have no particular relation. 
\end{example}

\begin{example}\label{ex:I3_tropical_renormalized}
	 The tropical integral of the graph $I_3$, a sunrise  inserted into a 1-loop multiedge,  has a pole,  $\tropicalintegral[I_3]=\frac{40}{\epsilon}+\ldots$  (\cref{ex:tropical_I3}). Interestingly, this integral does not have a quadratic pole, even if it is overall log-divergent, and the   subgraph by itself has a pole from its subdivergence
	 (\cref{ex:tropical_twopoint_twoloop}). The quadratic pole of the tropical integral $\tropicalintegral[I_3]$ happens to have residue zero, i.e. it is absent although it is structurally allowed. 
	 
	 One could be led to believe that $\tropicalintegral[I_3]$ does not require renormalization of subdivergences. This conclusion would be wrong, as can be understood from computing the same graph in the non-regularized massless theory, where the UV subdivergences are truly absent (as discussed in \cref{ex:sunrise_renormalized}). 
	 In the massless theory,  the integral is\footnotemark{}
	\begin{align*}
	\feynmanintegral[I_3]
	&= p^{-3\xi \epsilon} \frac{\Gamma(\xi(2 \epsilon-1))}{\Gamma(\xi)^4} \frac{\Gamma \left( \xi(1- \epsilon) \right) ^4}{\Gamma \left( 3\xi(1 -  \epsilon)\right)  } \frac{\Gamma(3\xi \epsilon)}{  \Gamma(\xi(1+2\epsilon))} \frac{  \Gamma(\xi(1 -3 \epsilon)  )}{\Gamma (2\xi(1 - 2  \epsilon  ))}.
	\end{align*}
	The tropical limit $\xi \rightarrow 0$ of this non-regularized integral is
	\begin{align*}
	\lim_{\xi \rightarrow 0} \feynmanintegral[I_3] &= \frac{-2(1+2\epsilon)}{\epsilon (\epsilon-1)^3 (3\epsilon-1)} = -\frac{2}{\epsilon}-16 -\ldots, 
	\end{align*}
	Indeed, the simple pole  of the tropical $\tropicalintegral[I_3]=\frac{40}{\epsilon}+\ldots$ is different from this result.   In order to produce the correct UV pole term from the tropical integral, it is necessary to renormalize the vertex subdivergence as in \cref{ex:sunrise_renormalized}:
	\begin{align*}
	\tropicalintegral_\ren[I_3]&=   \frac{40(1+\epsilon)}{\epsilon(4 \epsilon^3+8 \epsilon^2-\epsilon-2)}- 3 \times \frac 2 \epsilon \times \frac{3}{(\epsilon+1)(\epsilon-1)} = - \frac 2 \epsilon - 10 - \ldots .
	\end{align*}
\end{example}

	\footnotetext{	Even in the non-regularized massless theory, $I_3$ does not have IR-divergences as long as the external momentum $p^2\neq 0$ does not vanish, see \cite[Sec.~12.4~Ex.~3]{kleinert_critical_2001}.
		
	The pole term of $I_3$ has also been computed in the IR-regularized long-range  theory in \cite[(C.42)]{benedetti_longrange_2020}. Our massless, non-regularized integral reproduces this pole term  upon setting  $\xi = \frac{d+\varepsilon}{4}$, so that
			\begin{align*}
		p^{ 3\varepsilon} \feynmanintegral[I_3]
		&=  \frac{\Gamma \left( \frac{d-5\varepsilon}{4} \right) \Gamma \left( \frac{d-\varepsilon}{4} \right) ^4 \Gamma \left( \frac{3\varepsilon-d}{4} \right) \Gamma \left( \frac{3 \varepsilon}{2} \right)  }{\Gamma \left( \frac{d-3\varepsilon}{2} \right) \Gamma \left( \frac{3d-3\varepsilon}{4} \right) \Gamma \left( \frac{d+\varepsilon}{4} \right) ^4 \Gamma \left( \frac{d+5\varepsilon}{4} \right)  } = \frac{2 \Gamma \left( -\frac d 4 \right)  }{3 \Gamma \left( \frac d 2 \right) \Gamma \left( \frac{3d}{4} \right) \varepsilon }+ \asyO{\epsilon^0}.
	\end{align*}
	When comparing renormalized amplitudes, notice that the \enquote{renormalized} formula  \cite[(A.7)]{benedetti_longrange_2020} for the beta function does not include the counterterm required for the vertex-type subdivergence of $I_3$. 
}

\subsection{Mulitplicative renormalization in minimal subtraction } \label{sec:MS}

According to  \cref{def:tropical_integral,tropical_omega}, a  $\loopnumber$-loop vertex-type tropical  integral is proportional to $\mu^{-2\loopnumber \epsilon}$ which, in the context of renormalization, should be viewed as scale dependence rather than IR cutoff.  In our convention of the Lagrangian \cref{longrange_lagrangian}, the bare coupling has mass dimension $[g_0]=2\epsilon$. We introduce an arbitrary reference scale $\mu_0$ in order to obtain a dimensionless expansion parameter $u$ according to
\begin{align}\label{def:g0_u}
	g_0 &\defas u \cdot \mu_0^{2\epsilon}.
\end{align}
With this choice, the bare Green functions are power series in $u$, and the scale dependence occurs in the form of ratios  
\begin{align}\label{def:L}
	\frac{\mu^2}{\mu_0^2} &=: e^L.
\end{align}
The physical series is  alternating   (in the physical range $u\geq 0$), while the generating functions $\Gamma_n$ in  \cref{sec:PDE_other} were defined to be non-alternating series of the coupling, and did not include the factor $\mu^{-2\loopnumber \epsilon}$ (\cref{def:tropical_integral}). The precise relation between the physical massless bare 4-point function, and the function $\Gamma_4$, is therefore  
\begin{align}\label{G4_bare}
G_4(u, L, \epsilon) = u\cdot \Gamma_4(-u e^{-\epsilon L}, \epsilon,\kappa=0) &= u  - \sum_{\loopnumber=1}^\infty (-u)^{\loopnumber+1}  \mu_0^{2\epsilon \loopnumber} \sum_{\Graph} \frac{ \mu^{-2\epsilon \loopnumber}\tropicalintegral[\Graph]}{\abs{\Aut(G)}}.
\end{align}
Knowing $G_4$ at $\kappa=0$ is sufficient for the renormalization of the coupling even for the theory at $\kappa \neq 0$ as long as we choose counterterms independent from masses, which can always be done   \cite{collins_structure_1974}. 
Analogously, we only need the  first order in $\kappa$ of the 2-point function, and define  
\begin{align}\label{G2_bare}
	G_2(u, \kappa, L,\epsilon) =\mu^2-\Gamma_2\left(-u  e^{-\epsilon L},\epsilon,\kappa \right) =: \mu^2 G_{2,0} (u,L,\epsilon) - \kappa  G_{2,1} (u,L,\epsilon) + \asyO{\kappa^2}.
\end{align}
Notice that $\mu^2$ and $\kappa$ both have (tropical) mass dimension 2, so that   both of the functions   $G_{2,j}(u, L, \epsilon) $ are dimensionless power series in $u$ and start from a leading term $1+\ldots$. 

We will renormalize the theory multiplicatively on the level of amplitudes.
The general procedure is standard, see e.g. \cite{callan_introduction_1981,gross_applications_1981,balduf_dyson_2024} for reviews. 
  The coupling counterterm $Z_g(g,\epsilon)$ is the ratio between bare and renormalized coupling,
\begin{align}\label{def:renormalized_coupling}
	u=g_0 \mu_0^{-2\epsilon}&= g Z_g(g,\epsilon),
\end{align}
such that the renormalized 4-point function arising from \cref{G4_bare} is 
\begin{align}\label{G4_renormalized}
 	G_{\ren,4}(g,L, \epsilon) &=    G_4 \big(    gZ_g(g, \epsilon),L, \epsilon  \big) .
\end{align}
The minimal subtraction scheme is defined order by order by the condition that all counterterms consist of pole terms in $\epsilon$ exclusively. 
As discussed in \cref{sec:longrange_divergences}, the mass-independent 2-point function $G_{2,0}$ is superficially convergent, it is renormalized merely through the renormalization of $g_0$, which removes vertex subdivergences, and the field strength counterterm $Z_2(g,\epsilon)\equiv 1$ can be left out altogether.  Conversely, the mass coefficient $G_{2,1}$ of the 2-point function \cref{G2_bare} requires an overall counterterm, expressing the renormalization of the mass parameter $\kappa= m^2 \cdot Z_m(g,\epsilon)$, so that
\begin{align}\label{G2m_renormalized}
 G_{\ren,2,1}(g,L,\epsilon) &= Z_m(g,\epsilon)\cdot G_{2,1} \big( gZ_g(g,\epsilon), L, \epsilon\big). 
\end{align}
Having fixed $Z_g$ and $Z_m$, all other $n$-point amplitudes $G_n$ for $n\geq 6$ can be renormalized multiplicatively without needing further counterterms,
\begin{align}\label{Gn_renormalized}
	 G_{\ren, n}(g, m^2, L, \epsilon) = G_n\big(   gZ_g(g,\epsilon), ~m^2 Z_m(g,\epsilon),~L, ~\epsilon \big).
\end{align}

\begin{example}\label{ex:MS_renormalized}
	For conciseness, we  show here only the first few terms of the Laurent expansion in $\epsilon$, the exact expressions for $\loopnumber \leq 3$ have been given in  \cref{ex:gamma_2loop,ex:gamma_3loop}. Using the shorthand  $t\defas e^{-\epsilon L}=\left(\frac {\mu^2}{\mu_0^2}\right)^{- \epsilon}$, the bare Green functions start with
	\begin{align*}
		G_4 &=u -\frac{3}{\epsilon}tu^2  + \left( \frac 9 {\epsilon^2} +\frac 9 \epsilon - \frac{27}{2} +\ldots \right)t^2 u^3 - \left( \frac{27}{\epsilon^3} +\frac{72}{\epsilon^2} -\frac{69}{2\epsilon} -\frac{741}{4}+\ldots \right) t^3 u^4  +\ldots \\
		G_{2,0} &= 1-\left( \frac 12 +  \ldots \right) tu  + \left(  \frac{3}{2\epsilon} +\frac 5 2 + \ldots  \right) t^2u^2  - \left( \frac{9}{2 \epsilon^2} + \frac{15}{\epsilon}+ \frac{227}{8}+\ldots \right) t^3u^3  +\ldots \\
		G_{2,1} &= 1-  \frac 1\epsilon t u  + \left(  \frac{2}{ \epsilon^2} + \frac 3 {2\epsilon} -3 + \ldots  \right) t^2u^2  - \left( \frac{14}{3 \epsilon^3} + \frac{21}{2\epsilon^2} - \frac{11}\epsilon  - \frac{101}4+\ldots \right) t^3u^3  +\ldots.
	\end{align*}
	The pole terms in $G_{2,0}$ are  caused by vertex subdivergences, while $G_{2,1}$   additionally has superficial divergences. The  counterterms start with
	\begin{align*}
		Z_g &= \scalemath{.9}{1 +\frac 3 \epsilon g + \left( \frac{9}{\epsilon^2}-\frac 9 \epsilon \right) g^2 + \left( \frac{27}{\epsilon^3} - \frac{63}{\epsilon^2}+\frac{87}{\epsilon} \right) g^3 + \left( \frac{81}{\epsilon^4}- \frac{621}{2 \epsilon^3} + \frac{774}{\epsilon^2}- \frac{1407}{\epsilon} \right) g^4 + \ldots}, \\
		Z_m &= \scalemath{.9}{1+ \frac 1 \epsilon g + \left( \frac 2 {\epsilon^2} - \frac 3 {2\epsilon} \right)g^2 + \left( \frac{14}{3 \epsilon^3} - \frac{21}{2\epsilon^2} + \frac{10}\epsilon \right)g^3  + \left( \frac{35}{3\epsilon^4} - \frac{183}{4\epsilon^3} + \frac{899}{8\epsilon^2} - \frac{929}{8\epsilon} \right) g^4 + \ldots}.
	\end{align*}
	The renormalized Green functions according to \cref{G4_renormalized,G2m_renormalized}, expressed as series in the renormalized coupling, start with
	\begin{align*}
		&G_{\ren,4} = \scalemath{1}{g+ \left(-\frac 3 \epsilon t + \frac 3 \epsilon \right)g^2 + \left(\left( \frac 9 {\epsilon^2}+ \frac 9 \epsilon - \frac{27}2 \pm \ldots\right) t^2-\frac {18} {\epsilon^2} e^{-2\epsilon L} + \frac 9 {\epsilon^2}-\frac 9 \epsilon  \right) g^2 + \ldots }\\
		&= \scalemath{.93}{g + \left( 3   L -\frac 3 2  \epsilon L^2 + \ldots \right)g^2 + \left( \frac {9(-3-4 L + 2 L^2 )} 2  + 9 \epsilon \left( 1+3  L + 2 L^2 -  L^3 \right)+ \ldots   \right) g^2+\ldots }\\
		&G_{\ren,2,0} = 1- \left( \frac 12 + \frac{1-L}2 \epsilon + \ldots \right) g  + \left( \frac{2-3L}{2} g^2 + \frac{12-14 L + 9 L^2}{4} \epsilon + \ldots \right) g^2 + \ldots \\
		&G_{\ren,2,1} = 1 + \left( L - \frac{L^2}{2} \epsilon \pm \ldots \right)g - \left( 3+3L-2L^2 - \frac{3+12L+6L^2-4L^3}{2} \epsilon + \ldots  \right)g^2 + \ldots 
	\end{align*}
	As expected, these functions are regular at $\epsilon \rightarrow 0$. Their value at $L=0$ is  non-trivial. 
	
\end{example}

\medskip 

We are especially interested in the renormalization group functions $\beta,\gamma_m$ and their critical exponents \cite{callan_broken_1970,symanzik_small_1970,symanzik_massless_1973a,weinberg_new_1973,kleinert_critical_2001}. Through \cref{def:renormalized_coupling}, the renormalized coupling $g$ can be viewed as a function of $g_0, \epsilon, $ and $\mu_0$. Its derivative with respect to the reference scale $\mu_0$ defines the \emph{beta function}.   Analogously, the  \emph{mass anomalous dimension} is the derivative of the renormalized mass $m^2$ with respect to the scale $\mu_0$, 
\begin{align}\label{def:beta_function2}
	\beta(g, \epsilon) &\defas \mu_0 \frac{\partial}{\partial \mu_0} g(g_0, \mu_0, \epsilon), \qquad \gamma_m(g,\epsilon) \defas \mu_0 \frac{\partial}{\partial \mu_0}\ln m^2(g_0, \mu_0 \epsilon).
\end{align} 
The perturbation series depend upon $\mu_0$ and $g_0$ only through the combined quantity $u=g_0 \mu_0^{-2\epsilon}$ (\cref{def:renormalized_coupling}), this entails that $\beta$ and $\gamma_m$ are functions of $g$ and $\epsilon$ alone. One can equivalently compute $\beta$ and $\gamma_m$ as derivatives of counterterms counterterms with respect to $g$. In the MS scheme, the simple pole coefficient   $Z_{g,1}(g)\defas [\epsilon^{-1}] Z_g(g,\epsilon)$ is sufficient,
\begin{align}\label{def:beta_function}
\beta(g,\epsilon) &= \frac{-2\epsilon g}{1+ g \partial_g \ln Z_g(g,\epsilon) } & &\overset{\textnormal{in MS}}= 2 g^2 \partial_g Z_{g,1}(g )-2g\epsilon,  \\
\gamma_m (g)&= \beta(g,\epsilon) \cdot \partial_g \ln Z_m(g,\epsilon) & &\overset{\text{in MS}}= -2g \partial_g Z_{m,1}(g).\nonumber 
\end{align}

\begin{example}\label{ex:beta_function_MS}
	In MS for $N=1$, the series start with
	\begin{align*} 
		\beta &= -2\epsilon g + 6 g^2 -36 g^3 + 522 g^4 - 11256 g^5 + \frac{1224063}{4}g^6 - \frac{97292007}{10} g^7 \pm  \ldots \\
			\gamma_m &= -2 g + 6 g^2 - 60 g^3 + 929 g^4 - \frac{76749}{4}g^5 + \frac{2408967}{5}g^6 - \frac{559460571}{40}g^7 \pm \ldots.
	\end{align*}
	In the tropical theory, all coefficients are rational. Non-tropical $\phi_4^4$ theory involves transcendental numbers, a numerical approximation is \cite[Tables~III,IV]{kompaniets_minimally_2017}
	\begin{align*}
		\beta &\approx  -2\epsilon g + 3 g^2 -5.667 g^3 +32.55 g^4 -271.6 g^5 +2849 g^6 -34776 g^7 + \ldots  \\
		\gamma_m &\approx  -g + 0.833g^2-3.500g^3+19.96 g^4 -150.8 g^5 +1355 g^6 + \ldots.
	\end{align*}
	We observe that both the tropical and the non-tropical theory give rise to alternating series in $g$ whose coefficients grow with the order. The absolute value of coefficients in the tropical theory is larger than in the non-tropical one.  This is expected from the fact that already for primitive graphs, the Hepp bound is much larger than the true period \cite{panzer_hepps_2022,balduf_statistics_2023}. We discuss this further in \cref{sec:correlation}.	
\end{example}

With respect to the so-defined scale dependence, the renormalized $n$-point Green's functions (\cref{Gn_renormalized} satisfy the Callan-Symanzik equation \cite{callan_broken_1970,symanzik_small_1970}
\begin{align}\label{callan_symanzik_equation}
	\left( (n-2)\epsilon +  \beta \cdot   \frac{\partial}{\partial g}- \gamma_m \cdot m^2  \frac{\partial}{\partial m^2}  \right)G_{\ren,n} &= \mu \frac{\partial}{\partial \mu} G_{\ren,n}.
\end{align}
The renormalization group flow is given entirely in terms of the functions $\beta$ and $\gamma_m$, and the normalous scale dependence $(n-2)\epsilon$, but there is no field anomalous dimension. This is expected since the tropical 2-point functions are superficially finite.

\subsection{A kinematic renormalization scheme} \label{sec:kinematic_renormalization}

Instead of minimal subtraction (\cref{sec:MS}), one may also use a kinematic renormalization scheme, where we define the renormalized mass and coupling through  \cref{G4_bare,G2_bare},
\begin{align}\label{def:quasikinematic_scheme}
	g=g(u,\epsilon) &\defas G_4(u, L=0,\epsilon),\qquad 
	m  =m(u,\epsilon)\defas m_0  G_{2,1}( u,L=0,\epsilon ) . 
\end{align} 
This  scheme has been used in \cite{benedetti_longrange_2020} for the long-range theory. Comparing \cref{def:quasikinematic_scheme} with \cref{def:renormalized_coupling}, one reads off the counterterms $Z_g(g,\epsilon)= 1/G_4 $ and  $Z_m = 1/{G_{2,1}} $.

For brevity, we will call \cref{def:quasikinematic_scheme} the \enquote{MOM scheme}, but we remark that it is not the only choice of kinematic renormalization scheme. For example, it would be possible to include the $\kappa$-dependence of $G_4$, or to use the \enquote{BPHZ-scheme} of  \cite{decalan_local_1981,david_large_1988} where $G_2$ and its first derivative are subtracted at zero external momentum in addition to $G_4$.  We leave a detailed comparison of these choices to future work.

To obtain renormalization group functions, it is possible to construct counterterms order by order similar to the MS scheme, and use   the first equations in \cref{def:beta_function}. Alternatively, since \cref{def:quasikinematic_scheme} directly determines the renormalized functions, one can   use  \cref{def:beta_function2}. In that case, it is useful to recall that the $\mu$-dependence resides in $u:= g_0 \mu_0^{-2\epsilon}$, so that the $\mu_0$-derivative can equivalently be expressed as a $u$-derivative,
\begin{align}\label{def:beta_kinematic}
	\beta(g,\epsilon) &= -2\epsilon u  \frac{\partial}{\partial u} g (u,\epsilon)  , \qquad \gamma_m = -2\epsilon u  \frac{\partial}{\partial u}m(u,\epsilon)  . 
\end{align}

\begin{example}\label{ex:beta_kinematic}
	The kinematic beta function at $N=1$ starts with
	\begin{align*}
		\beta(g,\epsilon) &= -2\epsilon g + 6 g^2 - \frac{18(3\epsilon-2)}{(\epsilon+1)(\epsilon-1)}g^3 + \frac{18(-49 + 196 \epsilon - 164 \epsilon^2 - 94 \epsilon^3 +120 \epsilon^4)}{(\epsilon-1)^2(\epsilon+1) (\epsilon+2) (2 \epsilon-1) (2 \epsilon+1)}g^4\\
		&= 6g^2 -36 g^3 +441 g^4 - 8061 g^5 + \frac{735039}{4}g^6 - \frac{39225231}{8}g^7\pm \ldots+\asyO{\epsilon}.
	\end{align*}
	As expected, this function is regular at $\epsilon=0$, but it has a non-trivial dependence on $\epsilon$, and the coefficient $[\epsilon^0]\beta$ is different from the beta function in MS (see \cref{ex:beta_function_MS}).  Likewise, the mass anomalous dimension in the kinematic scheme is,
	\begin{align*}
		\gamma_m &= -2 g + 6 g^2 -51 g^3 + 764 g^4 - \frac{60633}{4}g^5 + \frac{2921649}{8} g^6 - \frac{81162489}{8}g^7 \pm \ldots \\
		& + \left( -12 g^2 + \frac{429}{2}g^2 - \frac{10375}{2}g^4 + \frac{4677693}{32}g^5 - \frac{147908803}{32}g^6 \pm \ldots  \right) \epsilon + \asyO{\epsilon^2}.
	\end{align*} 
\end{example}

\begin{example}\label{ex:beta_N}
	Retaining the dependence on $N$, the beta functions start with  
	\begin{align*}
		\beta^\text{(MS)} &= -2\epsilon g +\frac{2N+16}{3}g^2 - \frac{20N + 88}{3}g^3 + \frac{46N^2+1068N+3584 }{9} g^4 \pm \ldots\\
		\beta^\text{(MOM)} &=  \frac{2N+16}{3}g^2 - \frac{20N + 88}{3}g^3 + \frac{-N^3 +18N^2 +848 N+ 3104 }9 g^4 \pm \ldots + \asyO{\epsilon}.
	\end{align*}
\end{example}

\subsection{Critical exponents}\label{sec:critical_exponents}

The critical coupling $g_\star$ is a function of $\epsilon$, defined by
\begin{align}\label{def:critical_point}
	\beta\big(g_\star(\epsilon), \epsilon\big)&\overset !=0. 
\end{align}
In perturbation theory, this equation can be solved straightforwardly\footnote{While naively solving the system works for low loop orders, in the special case of MS renormalization conditions it is much more efficient to exploit that $\beta(g,\epsilon) =-2\epsilon g + f(g)$, such that $\beta=0$ means $\epsilon(g_\star) = \frac{f(g_\star)}{-2g_\star}$. The latter power series can be   inverted to deliver $g_\star(\epsilon)$.} for the series coefficients of $g_\star(\epsilon)$. Then, the critical exponents describe how a theory behaves close to a critical point \cite{fisher_theory_1967}. 	
The  \emph{correction to scaling exponent}\footnote{We choose the unusual letter $\correctiontoscaling$ since $\omega$ is already used heavily for the degree of divergence (\cref{def:sdd_scaled}).} is the slope of the beta function at a critical point, 
\begin{align}\label{def:correction_to_scaling}
	\correctiontoscaling(\epsilon) &\defas \partial_g \beta(g,\epsilon)\Big|_{g=g_\star(\epsilon)}
\end{align}
Analogously, the \emph{mass critical exponent} is defined as 
\begin{align}\label{def:mass_critical_exponent}
	\nu(\epsilon) &\defas \frac{1}{2+\gamma_m(g_\star(\epsilon))}.
\end{align}
Unlike the renormalization group functions $\beta$ and $\gamma_m$, the critical exponents are independent of the chosen renormalization scheme.

\begin{example}\label{ex:critical_exponents}
	In MS, the critical point (\cref{def:critical_point}) of the tropical theory starts with
	\begin{align*}
		g_\star^{\text{(MS)}} &= \scalemath{.9}{\frac{\epsilon}{3}+ \frac{2 \epsilon^2}{3}- \frac{5 \epsilon^3}{9} + \frac{346 \epsilon^4}{81} - \frac{22367 \epsilon^5}{648} + \frac{412672 \epsilon^6}{1215}- \frac{669955249	\epsilon^7}{174960} + \frac{11743752875 \epsilon^8}{244944} + \ldots}.
	\end{align*}
	 Inserting this into \cref{def:correction_to_scaling,def:mass_critical_exponent},  one finds the critical exponents
	\begin{align*} 
		\correctiontoscaling &= \scalemath{.9}{2 \epsilon -4 \epsilon^2 + \frac{68\epsilon^3}{3} - \frac{1688 \epsilon^4}{9} + 1897 \epsilon^5 - \frac{3555593 \epsilon^6}{162} + \frac{273066547 \epsilon^7}{972} - \frac{ 842771759 \epsilon^8}{216} \pm \ldots} \\
		\nu &= \scalemath{.9}{\frac 12 + \frac{\epsilon}{6} + \frac{2 \epsilon^2}{9}- \frac{7 \epsilon^3}{27}+ \frac{761 \epsilon^4}{324} - \frac{4465 \epsilon^5}{243} + \frac{530399 \epsilon^6}{2916} - \frac{17913625}{8748} \epsilon^7 + \frac{2686869587 \epsilon^8}{104976}\mp \ldots }\\
		&\approx \frac 1 2 + \frac \epsilon 6 + 0.2222 \epsilon^2 - 0.2593 \epsilon^3 + 2.349 \epsilon^4 - 18.37 \epsilon^5 + 181.9 \epsilon^6  -2048. \epsilon^7 \mp \ldots . 
	\end{align*}
	In kinematic renormalization (\cref{sec:kinematic_renormalization}),  the full $\epsilon$-dependence of $\beta(g,\epsilon)$ needs to be taken into account when computing the critical coupling (\cref{def:critical_point}, and
	\begin{align*} 
		g_\star^\text{(MOM)} &= \frac{\epsilon}{3}+\frac{2 \epsilon^2}{3}-\frac{19 \epsilon^3}{8} + \frac{1585}{324}\epsilon^4 - \frac{13915}{324}\epsilon^5 + \frac{6515861}{15552}\epsilon^6 + \ldots. 
	\end{align*}
	As expected, $ g_\star^\text{(MOM)}$ is different from $g_\star^\text{(MS)}$. However, using $g_\star^\text{(MOM)}$   in \cref{def:correction_to_scaling} with the MOM beta function reproduces $\correctiontoscaling$ from the MS calculation, and likewise for $\nu$.
 	For comparison, in non-tropical   $\phi^4_4$ theory, the critical exponents begin with \cite{kompaniets_minimally_2017}
	\begin{align*}
		\psi &\approx  2 \epsilon -2.518\epsilon^2 + 12.95 \epsilon^3 -83.76 \epsilon^4 + 664.0 \epsilon^5-5959 \epsilon^6 + \ldots \\
		\nu  &\approx  \frac 12 +  \frac{\epsilon}6 +0.1728 \epsilon^2 -0.1523 \epsilon^3 + 1.134  \epsilon^4 - 6.945 \epsilon^5 +53.08 \epsilon^6+\ldots .
	\end{align*}
	Once more, except for the   leading terms, the absolute value of coefficients in the tropical theory is larger than in the non-tropical one. 
\end{example}

For the $O(N)$ symmetric theory, the fact that the first coefficient of the beta function (\cref{ex:beta_N}) is proportional to $(N-8)$, together with the definitions \cref{def:correction_to_scaling,def:mass_critical_exponent},  implies that both critical exponents are singular at $N=-8$. All coefficients of $\correctiontoscaling$ and $\nu$ are rational functions in $N$ whose denominator is a power of $(N+8)$. 
Moreover, according to \cref{ren:propagator_vanishing},   at the special value $N=-2$ the mass anomalous dimension $\gamma_m$ vanishes, and the critical exponent $\nu$  coincides with its tree level value $\nu=\frac 12$ to all orders.

\begin{example}
	The $N$-dependent perturbation series start with
	\begin{align*}
		\correctiontoscaling &= 2\epsilon - \frac{12(22+5N)}{(N+8)^2}\epsilon^2 + \frac{12 (8528+3424N+418N^2+23 N^3)}{(N+8)^4}\epsilon^3 + \asyO{\epsilon^4} \\
		\nu &= \frac 12 + \frac{N+2}{2(N+8)}\epsilon + \frac{(N+2)(76+31N+N^2)}{2(N+8)^3}\epsilon^2 \\
		&\qquad + 	\frac{(N+2)(-10016-808N+588N^2+29N^3+N^4)}{2(N+8)^5}\epsilon^3+\asyO{\epsilon^4} .
	\end{align*}
\end{example}

\subsection{Differential equation for renormalized  Green functions}

As discussed in \cref{sec:MS}, renormalization of the tropical amplitudes amounts to multiplicative redefinition of the coupling and mass parameters, but there is no overall field redefinition. This can equivalently be expressed on the level of the quantum effective potential, where it is useful to   work with the generating function as a function of the coupling, $\bar{\mathcal G}$  of \cref{def:Gbar}, instead of the function $\mathcal G$ that depends on loop order. 
The renormalized generating function $\bar {\mathcal G}_\ren$ is related to the unrenormalized $\bar{\mathcal G}$ through multiplicative redefinition of its arguments as in \cref{G4_renormalized,G2m_renormalized},
\begin{align}\label{def:G_renormalized}
	\bar {\mathcal G}_\ren \left( \phi,g,m^2 \right) &= \bar{\mathcal G} \left( \phi,  g Z_g, m^2 Z_m\right) \quad \Leftrightarrow \quad \bar {\mathcal G}(\phi,g_0,\kappa ) = \bar {\mathcal G}_\ren \big(\phi, \mu^{-2\epsilon} g_0 Z_g^{-1} , \kappa Z_m^{-1}\big).
\end{align}

\begin{theorem}\label{thm:G_PDE_renormalized}
	The renormalized quantum effective potential \cref{def:G_renormalized} of the tropical theory satisfies the partial differential equation
	\begin{align*}
		\Big(-\beta(g,\epsilon) \partial_g + (1-\epsilon)\phi \partial_\phi  + \big(2+\gamma_m(g,\epsilon)\big) m^2 \partial_{m^2} + 2\epsilon -4 \Big)\bar{\mathcal G}_\ren &= \frac 1 {1-\partial_\phi^2\bar{\mathcal G}_\ren} + \frac{(N-1) \phi}{\phi - \partial_\phi \bar{\mathcal G}_\ren}-N.
	\end{align*}
\end{theorem}
\begin{proof}
	The unrenormalized $\mathcal G$ satisfies the PDE \cref{lem:G_PDE_g}; it remains to replace the derivatives. The variable $\phi$ is unaffected by renormalization. In order to replace $g_0$-derivatives by $g$-derivatives, notice that the definition of the beta function, in the form \cref{def:beta_kinematic}, implies that $ \frac{\partial g_0}{\partial g} = \frac{-2\epsilon}{\beta}g_0$. In the massless case, this identity allows to rewrite $2\epsilon g_0\partial_{g_0}$ in the unrenormalized PDE. 
	
	In the massive theory, a subtlety arises from the mass renormalization.   Firstly, observe that $\kappa = m^2 \cdot Z_m(g,\epsilon)$ is linear in $m^2$ since we choose the counterterms independent of the mass. This implies that the logarithmic derivative is unchanged, $\kappa \partial_\kappa \bar{\mathcal G}= m^2 \partial_{m^2}\bar{\mathcal G}_\ren$. However, the renormalization of $\kappa$ introduces a new $g$-dependence, so that
	\begin{align*}
		\beta \cdot \partial_g \bar{\mathcal G}_\ren &= \beta \cdot \left(  \frac{\partial g_0}{\partial g}\cdot \partial_{g_0} \bar{\mathcal G}   + \frac{\partial \kappa}{\partial g}\cdot \partial_\kappa \bar{\mathcal G}    \right)
		= \beta \cdot \left(\frac{-2\epsilon}{\beta}g_0 \partial_{g_0} \bar{\mathcal G}   +m^2 \frac{\partial Z_m}{\partial g} \partial_\kappa \bar{\mathcal G}  \right).
	\end{align*}
  With $m^2=\kappa Z_m^{-1}$, the second term becomes $\beta \partial_g \ln Z_m=\gamma_m$ according to \cref{def:beta_function}. But since  $\kappa \partial_\kappa \bar{\mathcal G}= m^2 \partial_{m^2}\bar{\mathcal G}_\ren$, we can solve for the derivative of $\bar{\mathcal G}$ according to 
	\begin{align*}
		2\epsilon g_0 \partial_{g_0}\bar{\mathcal G} &= -\beta \partial_g \bar{\mathcal G}_\ren  + \gamma_m \cdot m^2 \partial_{m^2}\bar{\mathcal G}_\ren.
	\end{align*}
	Inserting this transformation into \cref{lem:G_PDE_g} gives the claimed equation.	
\end{proof}
\Cref{thm:G_PDE_renormalized} holds for all renormalization schemes, but in order to solve it recursively, a concrete scheme needs to be specified.

\FloatBarrier

\section{Numerical results}\label{sec:numerics}

\subsection{Individual graphs} \label{sec:correlation}

\begin{figure}[htb]
	\centering
	\begin{tikzpicture}
		\begin{axis}[
			width=.7\linewidth, height=.5\linewidth,
			title={Contributions to $\beta^\text{MS}$ in the tropical vs. 4-dimensional theory, at $\loopnumber=6$ loops},
			xlabel={tropical theory}, 
			ylabel={4-dimensional theory},
			grid=major,
			ymin=-50, ymax=185,
			minor y tick num=1,
			xmin=-5000, xmax=29000,
			every x tick scale label/.style={
				at={(.98,0)},yshift=-2pt,anchor=north,inner sep=0pt
			}
			]

			\node[red,fill=white] at (axis cs:23000,120){\small $0$};
			\addplot [red, only marks, mark size=2pt] table {countertermComparison/countertermComparisonMS-6-1.txt};
			
			\node[black,fill=white] at (axis cs:11000,30){\small $1$};
			\addplot [black, only marks, mark size=2pt,mark=square*] table {countertermComparison/countertermComparisonMS-6-2.txt};
			
			\node[blue,fill=white] at (axis cs:1000,40){\small $2$};
			\addplot [blue, only marks, mark size=1pt] table {countertermComparison/countertermComparisonMS-6-3.txt};

			\node[orange,fill=white] at (axis cs:8000,80){\small $3$};
			\addplot [orange, only marks, mark size=2pt, mark=x] table {countertermComparison/countertermComparisonMS-6-4.txt};
			
		\end{axis}
	\end{tikzpicture}
	\caption{\textbf{(a)}:  Contribution of individual graphs to the beta function in MS in the tropical theory ($x$-axis) and the 4-dimensional theory ($y$-axis). Numbers and colours indicate coradical degree, red points are primitive graphs (compare \cref{fig:P_H}). Symmetry factors are not taken into account. Graphs with 4 or 5 subdivergences form a dense cluster around the origin and have been left out. Propagator-type  subgraphs are potentially present, but not counted into the coradical degree.}%
	\label{fig:counterterm_correlation}%
\end{figure}
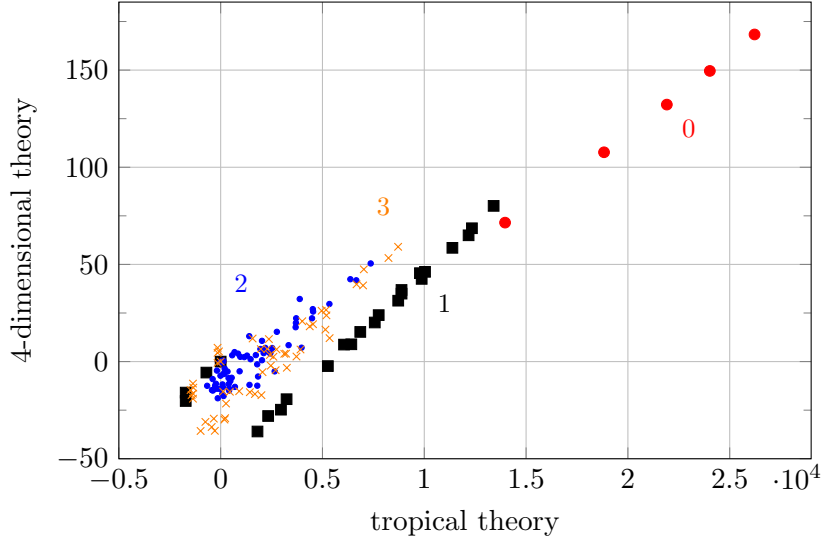

We computed the individual tropical Feynman integrals \cref{def:tropical_combinatorial} of all 1PI vertex and propagator graphs up to 9 loops, and renormalized them graph-by-graph.
According to \cref{def:beta_function}, every $\ell$-loop vertex graph contributes to the beta function in the MS scheme with $2\loopnumber$ times the residue of the $1/\epsilon$ pole of its counterterm. For primitive graphs, this contribution is $-2$ times the period of the graph (\cref{def:period}), whereas in the tropical theory this contribution is $-2$ times the Hepp bound.  As shown in \cref{fig:P_H} above, it had been known that the Hepp bound is closely correlated with the period in the 4-dimensional theory. The corresponding plot in \cref{fig:counterterm_correlation} shows that a strong correlation persists also for graphs with subdivergences. Graphs appear aligned in groups according to their coradical degree (colours in the plot) in the Hopf algebra of renormalization \cite{Kreimer:HopfAlgebraQFT}.\footnote{Since in the tropical theory, propagators do not give rise to subdivergences, the coradical degree accounts for vertex subdivergences only. Concretely, the coradical degree is equal to the maximal size of a collection of vertex subgraphs such that each pair of graphs in the collection is either nested or disjoint.}
We also note an offset in the correlation, such that contributions to the tropical beta function are often still positive for graphs with negative contribution in the non-tropical  theory. 
Cancellations due to renormalization (subtraction of subdivergences) thus appear to be weaker in the tropical theory.

\begin{figure}[htb]
	\begin{subfigure}{.48\linewidth}
		\begin{tikzpicture}
			\begin{axis}[
				width=\linewidth, height=.7\linewidth,
				title={Contributions to $\beta$,  MS, $\loopnumber=8$},
				ymode=log,
				ymin=.5, ymax=100000,
				minor y tick num=1,
				xmin=-22000000, xmax=55000000,
				every x tick scale label/.style={
					at={(.95,0)},yshift=-2pt,anchor=north,inner sep=0pt
				}
				]
			
				\addplot [const plot, fill=blue,draw=blue] table {betaContribution/betaContributionMS-8.txt};

			\end{axis}
		\end{tikzpicture}
		\subcaption{}
		\label{fig:beta_histogram_MS}
	\end{subfigure}
	\begin{subfigure}{.48\linewidth}
	\begin{tikzpicture}
		\begin{axis}[
			width=\linewidth, height=.7\linewidth,
			title={Contributions to $\beta$,  MOM, $\loopnumber=8$},
			ymode=log,
			ymin=.5, ymax=100000,
			minor y tick num=1,
			xmin=-22000000, xmax=55000000,
			every x tick scale label/.style={
				at={(.95,0)},yshift=-2pt,anchor=north,inner sep=0pt
			}
			]
			
			\addplot [const plot, fill=blue,draw=blue] table {betaContribution/betaContributionMOM-8.txt};

		\end{axis}
	\end{tikzpicture}
	\subcaption{}
	\label{fig:beta_histogram_MOM}
\end{subfigure}
	
	\caption{\textbf{(a)}:  Histogram showing the number of graphs, with bin width $2\cdot 10^6$, that contribute to the beta function in MS with a given value, taking into account their symmetry factor. Most graphs contribute zero or a small positive value. There are only a few, and small in magnitude, negative contributions,  but there are a few individual graphs that contribute a large positive value. \textbf{(b)}: Analogous histogram for the beta function in the kinematic scheme. This time, there is a substantial number of graphs that contributes with negative sign. The large positive outliers coincide with MS, these are primitive graphs. }
	\label{fig:beta_histogram}
\end{figure}
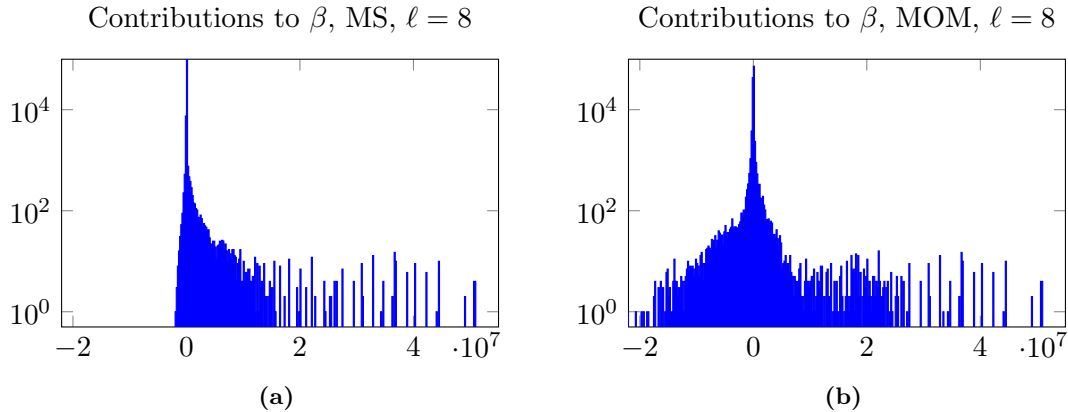

To better understand the contributions of individual graphs, histograms of all 130516 isomorphism classes of 8-loop 1PI vertex-type graphs at 8 loops are shown in \cref{fig:beta_histogram}. The histogram for MS (\cref{fig:beta_histogram_MS}) and the 6-loop data from \cref{fig:counterterm_correlation} show the same pattern: Most graphs contribute with a small positive value. There are a few small negative contributions and a few large positive outliers, the latter are the primitives. More than 83\% of graphs do not contribute to the MS beta function at all, these are in particular all 1-vertex product graphs.

In the MOM scheme (\cref{fig:beta_histogram_MOM}), the primitives still give the largest positive contribution, but in addition, there is now a large number of graphs with substantial positive or negative contributions. By inspection, we found that the largest non-primitive contributions are typically graphs with one large vertex subdivergence (i.e. subdivergence and cograph are vertex graphs of $\approx \frac \loopnumber 2$ loops), whereas the negative contributions of largest magnitude arise from graphs with a single multiedge subdivergence. Only less than 12\% of 1PI vertex graphs contribute zero to the 8-loop tropical beta function in MOM.

It has long been conjectured that primitive graphs dominate the MS beta function of $\phi^4$ theory at large loop orders \cite[\S4.1]{mckane_perturbation_2019}. We are now in a position to seriously probe the analogous claim in the tropical theory, by comparing our results (for the full tropical MS beta function) to estimates for the contributions of only the primitive graphs.
The latter can be obtained by Monte Carlo sampling; accurate data for $\loopnumber \leq 50$  is available from \cite[Tab.~2]{borinsky_tropicalized_2025} based on the novel weighted sampling algorithm introduced in that work. These data points are shown as red markers in \cref{fig:beta_MS_primitive_ratio}. As an independent check, and to investigate $\loopnumber \leq 85$, we estimated the sum of primitive tropical integrals by generating  $>10^6$ random primitive graphs per loop order,  and for each of them selected 1024 random Hepp sectors. Due to the large variance, this naive algorithm delivers inferior accuracy compared to the weighted sampling of \cite{borinsky_tropicalized_2025}, but the results are compatible within error bars, shown in black in \cref{fig:beta_MS_primitive_ratio}.

The plot in \cref{fig:beta_MS_primitive_ratio} shows that 
the primitive graphs contribute a non-zero fraction of the tropical MS beta function, but this fraction stabilizes already from 10 loops upwards to roughly 50\%. This strongly suggests that the conjecture fails in the tropical theory. We also know that it fails in zero dimensions, where the primitive graphs contribute asymptotically only a small fraction (roughly 2\%) of all 1PI vertex graphs \cite[\S6.3]{borinsky_renormalized_2017}.\footnote{The precise fraction in zero dimensions is $e^{-15/4}\approx 0.0235177$.}
In 4 dimensions, the full beta function is only available up to 8 loops \cite{Schnetz:NumbersAndFunctions}, which we now know is too low to reliably probe the asymptotic regime. In conclusion, while we cannot disprove the conjecture in 4 dimensions, we see no signs in its favour.

Furthermore, the conjecture is also violated in a spectacular way if we consider the MOM scheme instead of MS: As we shall point out below, asymptotically the tropical beta function coefficients $\beta_n$ in MOM are smaller than in MS by a factor $n$ (in the case $N=1$). The ratio of primitive contributions (which are scheme independent) over $\beta_n$ in MOM thus approaches infinity as $n\rightarrow\infty$. This means that in MOM, the contribution of all primitive graphs cancels (up to order $1/n$) with contributions from graphs with subdivergences.


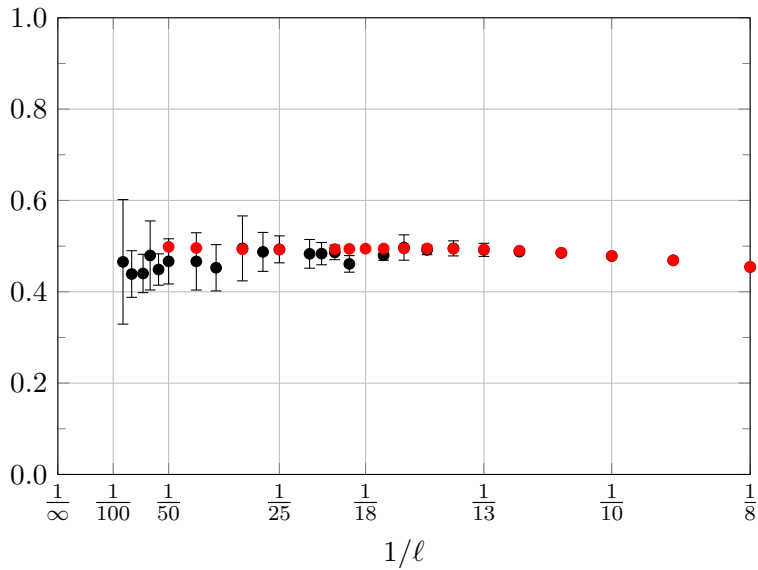
\begin{figure}[htbp]
	\centering
	\begin{tikzpicture}
		\begin{axis}[
			width=.7\linewidth, height=.5\linewidth,
			title={Primitive contribution to the $\loopnumber$-loop tropical beta function $\beta_{\loopnumber+1}$, MS},
			xlabel={$  1/ \loopnumber $}, 
			grid=major,
			xmin=0, xmax=0.125,
			xtick={0, .01, .02, .04,0.0555556,0.0769231, .1, .125},
			xticklabels={$\frac 1 \infty$, $\frac 1 {100}$, $\frac 1 {50}$, $\frac{1}{25}$, $\frac 1 {18}$, $\frac 1 {13}$, $\frac 1 {10}$,$\frac 1 {8}$},
			ymin=0, ymax=1,			
			minor y tick num=1,
			yticklabel style={
				anchor=east,
				/pgf/number format/precision=1,
				/pgf/number format/fixed,
				/pgf/number format/fixed zerofill,
			}
			]
			\addplot [black , only marks, error bars/.cd,
			y dir=both,y explicit] table [y error index=2] {betaMSPrimitiveRatio.txt};
			
			\addplot [red , only marks, error bars/.cd,
			y dir=both,y explicit] table [y error index=2] {betaMSPrimitiveRatioBorinsky.txt};

		\end{axis}
	\end{tikzpicture}
	\caption{Relative contribution of primitive graphs to the tropical beta function in the MS scheme. Black points arise from naive sampling of graphs , error bars indicate 2 standard deviations of sampling uncertainty of primitive graphs. Red points are from \cite{borinsky_tropicalized_2025}, their error bars are barely visible. }%
	\label{fig:beta_MS_primitive_ratio}%
\end{figure}

\subsection{Recursively solving the PDE}

We implemented multiple programs to recursively solve different versions of the tropical loop equation, both in \texttt{Pari/gp} \cite{theparigroup_pari_2025} and in \texttt{C++} using the \texttt{GMP} library for rational arithmetic. For all $\loopnumber\leq 9$, we confirmed that the resulting bare and renormalized  Green's functions $\Gamma^{(\loopnumber)}_n$ coincide with the sums of graphs computed above.\footnote{We used {\nauty} \cite{mckay_practical_2014} to compute the symmetry factors in this sum.} For values up to $n\leq 10$ legs and $\loopnumber\leq 40$ loops, we computed the full $\epsilon$-dependence of $\Gamma^{(\loopnumber)}_n$ as rational functions in $\epsilon$, and note that the degree of the numerator and denominator polynomials in $\epsilon$ is much larger than the loop order.

For higher loop orders, we computed truncated power series expansions in $\epsilon$ with exact rational coefficients, keeping sufficiently many terms of the series to find the coefficient $[\epsilon^1]$ in the target loop order. While the resulting Green's functions are not exact as functions of $\epsilon$, their singular parts are, which means that we obtain the exact rational coefficients of the beta function (\cref{def:beta_function}). For the case $N=1$, we determined the beta function in MS, and the correction to scaling exponent (\cref{def:correction_to_scaling}), up to $\loopnumber =400$ loops.
For other values of $N$ in the range $-12 \ldots 10$, we computed the series to at least 100 loops, in steps of $\Delta N=\frac 15$. Since the coefficients of renormalization group functions are polynomials in $N$, these data points suffice to recover the exact $N$-dependent polynomials up to 100 loops by interpolation.

In order to compute the mass anomalous dimension, we included the order $[\kappa^1]$ in the recurrence at $N=1$ up to 180 loops, and for the other values of $N$ up to 150 or 100 loops. Additionally, we have generated data up to order $\kappa^5$ for $\loopnumber \leq 10$ for the functions $\Gamma_2, \ldots, \Gamma_{16}$. We have verified that the resulting series, in both renormalization schemes, satisfy the Callan-Symanzik equation (\cref{callan_symanzik_equation}).

As is well known \cite{wright_graphs_1971,bender_asymptotic_1978,borinsky_renormalized_2017}, the number of Feynman graphs grows factorially with the loop order. With methods of zero-dimensional QFT (i.e.\ the limit $\epsilon \rightarrow 2$ of tropical field theory, see \cref{thm:zero_dim_PDE_series}), using the programs published together with \cite{balduf_primitive_2024}, one finds that the number of 1PI vertex graphs at 100 loops is larger than $10^{140}$, at 150 loops it is above $10^{230}$ and at 400 loops larger than $10^{800}$.

The data of all these power series is available from the authors, and will be moved to a permanent DOI upon publication of the article.

\subsection{Renormalization group functions in MS}

Recall that a formal power series  $f(g)= \sum_n f_n g^n$ is factorially divergent if there are constants $\left \lbrace s,a, b \right \rbrace \in \mathbb{R}$ such that the leading asymptotic growth  in the limit $n\rightarrow \infty$ is $f_n \sim s \cdot a^{-n-b} \Gamma(n+ b)$. This property can be inspected visually in a Domb-Sykes plot \cite{domb_susceptibility_1957}\cite[Sec.~8.1]{hinch_perturbation_1991}, which amounts to plotting the ratio 
\begin{align}\label{def:rn}
	r_n &:= \frac{f_{n+1}}{n\cdot f_n} \sim \frac 1 a + \frac{b}{a}\frac 1 n + \text{subleading terms}
\end{align}
as a function of $\frac 1 n$. If the data points $r_n$ approach \emph{linearly} a finite intersection with the $y$-axis, then the value of the intersect determines $a$ and its slope determines $b$. For the tropical beta function in the MS scheme, this plot is shown in \cref{fig:beta_MS_r1L}, it suggests $a=-\frac 13$ and $b=\frac 52$, these values are confirmed to high accuracy with methods described in \cite{balduf_asymptotic_2026}. Another interesting feature of \cref{fig:beta_MS_r1L} is that the limiting slope is only attained at very high loop order, while the data at $\loopnumber \leq 20$ loops suggests a \emph{wrong} large-order asymptotics. This indicates a large influence of subleading terms in the asymptotic expansion. A very similar mismatch had been observed for the 4-dimensional $\phi^4$ theory in \cite{balduf_statistics_2023} and discussed in detail for the 0-dimensional and 4-dimensional $\phi^4$ theory in \cite[\S5.3]{balduf_primitive_2024}.

With refined methods, we identify the leading singularities in the Borel plane and deduce the asymptotic form \eqref{beta-asymptotics_general}. Fitting to our data, the first few terms read approximately

\begin{equation}\label{beta_asymptotics_Richardson}
	\beta_n \sim 1.085887 \cdot  (-3)^n  \Gamma \left(n + \frac 5 2 \right) \left( 1 - \frac{3.0217}{n^{1/3}} + \frac{3.36}{n^{2/3}} + \ldots \right), \qquad n \rightarrow \infty,
\end{equation}
in particular, the leading correction decays only very slowly at the rate $n^{-\frac 13}$.

\begin{figure}[h]
	\begin{subfigure}{.48\linewidth}
		\begin{tikzpicture}[baseline]
			\begin{axis}[
				width=\linewidth, height=.75\linewidth,
				title={Ratio $r$ (\cref{def:rn}) of $\beta^\text{(MS)}$},
				xlabel={$1/\loopnumber$}, 
				grid=major,
				ymin=-4.5, ymax=-2.5,
				xmin=0, xmax=0.10,
				xtick={0, .01, .02, .04,0.0555556,0.0769231, .1},
				xticklabels={$\frac 1 \infty$, $\frac 1 {100}$, $\frac 1 {50}$, $\frac{1}{25}$, $\frac 1 {18}$, $\frac 1 {13}$, $\frac 1 {10}$},
				minor y tick num=1,
				yticklabel style={
					anchor=east,
					/pgf/number format/precision=1,
					/pgf/number format/fixed,
					/pgf/number format/fixed zerofill,
				}
				]
				\addplot [black, only marks, mark size=1.2pt] table {betaMSr1L.txt};
				
				\addplot[darkgreen,samples=10,domain=0:0.09,line width=.9pt] {-3*(1+2.5*x)};
				\node[darkgreen,fill=white,rotate=-14] at (axis cs:.075,-3.39){\footnotesize $1/a=-3, b=\frac 52$};
				\addplot[red,samples=10,domain=0:0.11,line width=.9pt] {-2.65*(1 + 6.4*x)};
				\node[red,fill=white,rotate=-32] at (axis cs:.062,-3.94){\footnotesize $1/a \approx -2.65, b \approx 6.4$};

			\end{axis}
		\end{tikzpicture}
		\subcaption{}
		\label{fig:beta_MS_r1L}
	\end{subfigure}
	\begin{subfigure}{.48\linewidth}
		\begin{tikzpicture}[baseline]
			\begin{axis}[
				width= \linewidth, height=.8\linewidth,
				title={Poles of Padé  of $\borel{\beta^\text{(MS)}}(\frac 1 3 \varphi(z))$},
				grid=major,
				ymin=-1.6 , ymax=1.6,
				xmin=-2, xmax=2,
				yticklabel style={
					anchor=east,
					/pgf/number format/precision=1,
					/pgf/number format/fixed,
					/pgf/number format/fixed zerofill,
				},
				xticklabel style={
					/pgf/number format/precision=1,
					/pgf/number format/fixed,
					/pgf/number format/fixed zerofill,
				}
				]
				
				
				\addplot [black, only marks, mark size=1pt] table {polesPadConformalBeta400.txt};
				\addplot [black, only marks, mark=x, mark size=1.2pt] table {zerosPadConformalBeta400.txt};
				
				\node[black,fill=white] at (axis cs:1.8,1.35){$z$};

				\draw[line width=.2pt, blue] (0,0) circle[radius=1];
			\end{axis}
		\end{tikzpicture}
		\vspace{.55cm}
		\subcaption{}
		\label{fig:poles_Borel_beta_conformal}
	\end{subfigure} 
	
	\caption{\textbf{(a)}: Black: growth ratio $r_n$ (\cref{def:rn}) of the tropical beta function in MS, plotted as function of inverse loop number. The data below $\loopnumber \approx 20$ loops suggests  wrong asymptotic growth parameters,   indicated by the red line. Compare the plots in \cite[Figures~10,22]{balduf_primitive_2024}. \textbf{(b)}: Poles (points) and zeros (crosses) of the Padé approximant of the conformally mapped $\borel{\beta}(\frac 13 \varphi(z))$.  The leading singularity at $z=-1$ is separated from further singularities along the circle, and an essential singularity at $u\rightarrow +\infty$ is revealed. The interpretation of this plot is discussed in detail in \cite{balduf_asymptotic_2026}.}
	\label{fig:beta_MS_data}
\end{figure}
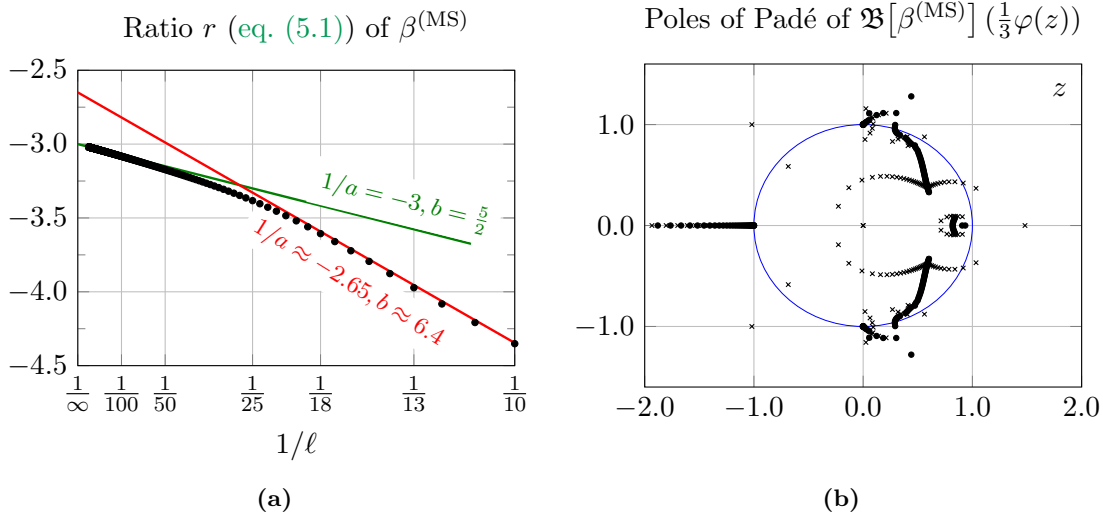

Having established that the coefficients $\beta_n$ grow factorially, one turns to the Borel transform 
\begin{align}\label{Borel_transform_beta}
	\borel{\beta}(u) & \defas \sum_{\loopnumber=1}^\infty \frac{\beta_n}{n!}u^n ,
\end{align}
where $\beta_n=[g^{n}]\beta(g) $ is the $(n-1)$-loop coefficient. 
The factorial growth coefficient $a=-\frac 1 3$ amounts to the location of the dominant singularity in the complex Borel plane. As discussed in detail in \cite{balduf_asymptotic_2026}, the Borel plane structure can be visualized by constructing a Padé approximant of the power series $\borel{\beta}(u)$ and plotting the location of its poles. Singularities with non-integer exponents produce sequences of spurious poles alternating with zeros, this effect often makes it impossible to see further singularities beyond the leading one \cite{nuttall_convergence_1977,martinez-finkelshtein_heine_2011,aptekarev_pade_2015}, \cite[Sec.~2]{costin_conformal_2021}. A well-established solution to this problem is to instead do the Padé approximant in a conformally mapped plane. One possible choice is  a conformal mapping from the sliced complex plane $u \in \mathbb C \setminus[-\frac 13, -\infty)$ to  the unit disk $\mathbb D=\left \lbrace z\in \mathbb C, \abs{z} \leq 1 \right \rbrace $ according to
\begin{align}\label{conformal_mapping}
	u=\varphi(z) &= \frac{4z}{(1-z)^2}, \qquad z=\psi(u) = \frac{-1+\sqrt{1+u}}{1+\sqrt{1 +u}}, \qquad \varphi(\psi(u))=u  \in \mathbb C \setminus[-1, -\infty).
\end{align}
Such mappings have long been used for resummation of perturbative series in QFT, see e.g.\ \cite{thooft_can_1977}, and have been studied more systematically recently in  \cite{costin_physical_2020,costin_conformal_2021,costin_uniformization_2022}.  \Cref{conformal_mapping} maps $u=-1$ to $z=-1$, and all points $u=\infty e^{i\theta}$ are mapped to $z=+1$, where $\frac{\theta}{2}+ \pi$ is the angle of approaching $z=+1$ in the complex plane.

The resulting plot is shown in \cref{fig:poles_Borel_beta_conformal},   mapping the position of singularities back to the non-conformal Borel $u$-plane results in the plot shown in \cref{fig:poles_Borel_beta_conformal_back} in the introduction. The interpretation of qualitative features of such plots is discussed in detail in \cite{balduf_asymptotic_2026}, the main insights for our present case are:
\begin{enumerate}
	\item Beyond the leading one (which gets mapped to $z=-1$), there are further singularities along the negative real $u$-axis, represented by points along the unit circle in $z$. 
	\item We observe a \enquote{front} of poles in $\mathbb D$ that moves to the right with increasing order. This suggests that there are many (probably infinitely many) further singularities along the negative real $u$-axis, which can not all be resolved individually. Alternatively, this front can be viewed as evidence for an essential singularity at $u\rightarrow -\infty$.
	\item There is another \enquote{front} in a quarter circle around the point $z=+1$. It indicates an essential singularity at $u \rightarrow +\infty$. 
	\item There is no evidence for any singularities off the real $u$-axis, or respectively off the unit circle in $z$. 
\end{enumerate}

Series analysis \cite{balduf_asymptotic_2026} reveals that $u=-\frac 13$ is a 6-fold \emph{confluent} singularity with exponents $b+\frac i 3$ for $i=0,\ldots,5$. The strongest (most negative) exponent is $b=-2-\frac N2$.
Resonance (pairs of exponents that differ by unity) gives rise to a logarithmic contribution, and thus the singularity  has the form
\begin{equation}\label{eq:instanton-MS}
	\borel{\beta}(u) = A_0(u) + \frac{A_1(u) + A_2(u) \left(u+\frac 13\right) \ln \left( u+\frac 13\right)}{\left(u+\frac 13\right)^{2+\frac N 2}}.
\end{equation}
Here $A_0$ is holomorphic at $u=-\frac 13$ while $A_1$ and $A_2$ have series expansions in non-negative integer powers of the third root $\sqrt[3]{u+1/3}$. Expanding the singular terms (with $A_1$ and $A_2$) as series in $u$ leads to the asymptotic form \eqref{beta-asymptotics_general} for the coefficients $\beta_n$. Numerical estimates for the first few coefficients of $A$, $B$, and $C$ can be found in \cite{balduf_asymptotic_2026}.

The existing literature on instanton calculations in 4 dimensions assumes much simpler asymptotics for $\beta_n$, expecting only integer power corrections $n^{-k}$. In fact, there are calculations of the first $\frac 1 n$ correction to the leading asymptotics \cite{KomarovaNalimov:FirstCorrectionON,kubyshin_corrections_1984}. The presence of corrections with fractional exponents $ n^{-\frac k 3}$ and logarithms $\log n$ in the supposedly \emph{simpler} tropical theory raises the question if contributions of instantons even in the 4-dimensional theory might in fact be more complicated than traditionally assumed. We therefore suggest the challenge of reproducing \eqref{eq:instanton-MS} purely from instanton calculus, adapted to the tropical theory.

Putting corrections aside, the \emph{leading} asymptotic term arising from \eqref{eq:instanton-MS} is
\begin{equation}\label{eq:MS-leading}
    \beta_n \sim n!\cdot (-3)^n \cdot n^{1+\frac N 2}\cdot \frac{3^{2+\frac N 2}A_1(-\frac 13)}{(1+\frac N 2)!}.
\end{equation}
We compare this with the instanton contributions in the 4-dimensional theory, which are known for minimal subtraction \cite{mckane_nonperturbative_1984,KomarovaNalimov:HigherOrdersON} and and for renormalization by subtraction at symmetric momenta \cite{brezin_perturbation_1977a,mckane_instanton_1978}. In both cases, the leading asymptotics is
\begin{equation}\label{eq:instanton-4d}
    \beta_n \sim n!\cdot (-1)^n \cdot n^{3+\frac N 2}\cdot C,
\end{equation}
where only the constant $C$ varies with the scheme. The factor $(-1)^n$ arises from the location of the singularity in the Borel plane, which is at $u=-1$ in 4 dimensions. Its replacement by $(-3)^n$ suggests that, along the family of long-range models in \cref{longrange_lagrangian}, the instanton moves from $-1$ at $\xi=1$ to $-\frac 13$ at $\xi=0$.
The exponent $3+\frac N 2=1+\frac{D+N}2$ of the power with base $n$ arises from the number of degrees of freedom of instantons. In the tropical limit, the initially $D=4$ translational degrees of freedom disappear, reducing the exponent by two. Hence, at least on this very naive level, \eqref{eq:MS-leading} might indeed be understood from the tropical limit of long-range instantons.



\subsection{Kinematic scheme and renormalons}

For the tropical beta function in the kinematic scheme, we again apply the series analysis methods detailed in \cite{balduf_asymptotic_2026}.
We find again a singularity at $u=-\frac 13$ in the Borel plane, which we thus call the instanton. However, its exponents are increased by $\frac{N+8}9$ compared to MS: 
\begin{equation}\label{eq:instanton-MOM}
	\borel{\beta^\text{(MOM)} }(u) = A_0^\text{(MOM)}(u) + \frac{A_2^\text{(MOM)}(u) + A_2^\text{(MOM)}(u) \left(u+\frac 13\right) \ln \left(u+\frac 13\right)} {\left(u+\frac 13 \right)^{\frac{7N+20}{18}}}
\end{equation}
where as before, $A_0$ is holomorphic whereas $A_1$ and $A_2$ have expansions in $\sqrt[3]{u+1/3}$.

As mentioned before, existing calculations of leading instanton asymptotics \cref{eq:instanton-4d} in 4 dimensions depend on the renormalization scheme only through the constant $C$, but predict the same exponent of $n$ and thus the same (strongest) exponent of the singularity in the Borel plane. The observed scheme dependence of the exponent in the tropical theory was therefore unexpected. Since the value of beta function coefficients---and thus their growth rates---can be altered \emph{arbitrarily}\footnote{except for $\beta_2$ and $\beta_3$ which are fixed (scheme independent)} by tuning the renormalization scheme, it is clear that the exponent can be altered through the scheme, so our finding is not in itself inconsistent.

However, our observation signals that the exponent can differ from the dimension of the instanton parameter space even for standard renormalization schemes defined through kinematic boundary conditions, like our subtraction at zero momentum. This suggests that perhaps also in 4 dimensions, the exponents may differ for standard renormalization schemes. It seems thus worthwhile to implement the symmetric momentum subtraction scheme from \cite{brezin_perturbation_1977a,mckane_instanton_1978} in the tropical theory and analyse its asymptotics. Indeed, the tropical limit can be adjusted to retain momentum dependence, but the recurrences to obtain large loop orders become more complicated. We hope to develop this generalization in the future.

Besides the changes in the instanton singularity, a striking difference of MOM compared to MS  is the existence of additional singularities on the real Borel axis, and in particular the positive real axis, see \cref{fig:poles_Borel_beta_MOM_conformal_back} in the introduction. 
Considerations based on the renormalization group equation and the interplay between Borel transforms and momentum dependence  \cite{gross_dynamical_1974,parisi_singularities_1978,parisi_borel_1979,parisi_physical_1979a,beneke_renormalons_1999} predict a \emph{renormalon} singularity whose location and exponent are determined by the 1-loop and 2-loop coefficients of the beta function, $\beta_2$ and $\beta_3$. Concretely: (see \cite{dunne_instantons_2022} for a concise review)
\begin{align*}
	\beta_n \sim C\cdot \left( \frac{\beta_2}2\right)^n \Gamma \left(n+1+\frac{2\beta_3}{\beta_2^2} \right).
\end{align*}
Since the renormalization group equation also holds in the tropical theory, plugging in the $N$-dependent coefficients from \cref{ex:beta_N} predicts the first renormalon singularity $(a_R-u)^{b_R}$ of the tropical beta function in the Borel plane with parameters
\begin{align}\label{renormalon_prediction}
	a_R &= \frac{2}{\beta_2}= \frac{3}{N+8}, \qquad
    b_R =-1-\frac{2\beta_3}{\beta_2^2}= \frac{-N^2+14N+68}{(N+8)^2}.
\end{align}

Numerically, with details once more given in \cite{balduf_asymptotic_2026}, we find that  the leading singularity on the positive real Borel line is indeed located at $u=a_R=\frac 3{N+8}$, and we thereby identify it with the renormalon. 
Additionally, we find a second singularity on the \emph{negative} real axis, located at $u=-a_R$. The exponents $b_\pm$ of these singularities at $u=\pm a_R$ are
\begin{align}\label{renormalon_exponents}
    b_+=\frac{3N^2+92N+468}{(N+8)^2}
    \qquad\text{and}\qquad
    b_-=\frac{3N^2+12N+12}{(N+8)^2}
\end{align}
In particular, the singularity on the negative axis is stronger.
Since the instanton singularity is located at $u=-\frac 13$,  the renormalon (on the negative axis) dominates asymptotically when $N>1$\footnote{In $\phi^4_4$ theory, where the instanton is at $-1$ and the renormalon at $\frac{6}{N+8}$, the renormalon already dominates for $N>-2$ see \cite{parisi_borel_1979}. In that sense, the renormalon is stronger in $\phi^4_4$ theory than in tropical $\phi^4$ theory. This is consistent with the finding of \cref{sec:correlation}, that cancellation effects are weaker in the tropical theory. }. In the special case $N=1$, the renormalon is located at $u=+\frac 13$, and for $N$ below that, the instanton dominates the asymptotics.

We note that both renormalon exponents in \cref{renormalon_exponents} differ from the prediction \cref{renormalon_prediction}, suggesting that the renormalization group argument does not capture fully the structure of the renormalon singularity. Also it remains at present unclear how the renormalon changes when going from our kinematic scheme to the scheme subtracting at symmetric non-zero momenta.

Another promising future project could be to clarify   which classes of graphs contribute to which properties of renormalons. The   expectation from the literature is that these are graphs containing chains of small subgraphs, but more than one such family has been considered  \cite{zakharov_qcd_1992,faleev_status_1995,beneke_renormalons_1999,clingerman_asymptotic_2025}. The tropical theory might provide a test bed to precisely classify which classes of graphs contribute to which property of renormalons in which renormalization scheme.

\medskip 

For the present work, one crucial outcome is that these renormalon  singularities are absent in the MS scheme. Conceptually, the renormalization group functions are renormalization scheme dependent and therefore it is clear that there \emph{exists} renormalization schemes where they are free of renormalons, but the crucial property is whether the MS scheme has this property (see e.g. \cite[(iv)]{bergere_ambiguities_1984}). Our data confirms that empirically, there are no renormalons in the MS renormalization group functions in the tropical theory. This is consistent with numerical calculations e.g.\ in QED \cite{palanques-mestre_1_1984} and in $\phi^4$ theory.\footnote{David Broadhurst, private communication}


Apart from the existence of renormalon singularities at finite positive location in the Borel plane, it has also been pointed out (e.g.\ \cite[Sec.~IV.2]{decalan_local_1981}) that certain families of graphs produce Borel transforms that grow faster than $e^u$ as $u\rightarrow \infty$ in the Borel plane. These are called \enquote{renormalons at infinity}, and while they do not cause the Borel transform to have singularities at finite $u$, their contributions might cause the Laplace integral to diverge for all finite arguments---thereby spoiling Borel resummability. However, our data suggests only exponential growth $e^u$ for large $u$ in the Borel plane of the tropical MS beta function.

In particular, our observations suggest that the tropical MS beta function is Borel summable. We leve the study of the implications of our findings for Borel summability, and for resummation of critical exponents, to future work.

\newpage 

\begin{multicols}{2}
\printbibliography
\end{multicols}

\end{document}